\documentclass[aos]{imsart}

\RequirePackage{amsthm,amsmath,amsfonts,amssymb}
\RequirePackage[authoryear]{natbib}
\RequirePackage[colorlinks,citecolor=blue,urlcolor=blue]{hyperref}
\RequirePackage{graphicx}

\startlocaldefs
\theoremstyle{plain}

\newtheorem{theorem}{Theorem}[section]
\newtheorem{lemma}{Lemma}
\theoremstyle{remark}

\newtheorem{example}{Example}

\newtheorem{condition}{Condition}
\newtheorem{proposition}{Proposition}
\newtheorem*{rmk*}{remark}
\newtheorem{remark}{Remark}
\newtheorem{corollary}{Corollary}
\newtheorem*{corollary*}{Corollary}

\endlocaldefs
\graphicspath{ {./applications/} }
\usepackage{xr}
\usepackage{enumerate}
 \usepackage{multirow}
\usepackage{bm}
\usepackage{xcolor}
\usepackage{ulem} 

\def\mta{\mt_\A}
\def\mtb{\mt_\B}
\def\tss{{|\mt|}}
\newcommand{\oxw}[1]{{ \overline{ x_{w\cdot}^{#1}}}}
\def\ma{\mathcal A}
\def\beginp{\begin{pmatrix}}
\def\endp{\end{pmatrix}}
\def\tv{\tilde V}

\def\oc{1_\textup{c}}

\def\zin{z \in \{(A_w0), (A_w1)\}}
\def\znotin{z \not\in \{(A_w0), (A_w1)\}}
\def\dwab{d_w(A_wb)}

\def\huwab{\hat U_w(A_wb)}
\def\hvwab{\hat v_w(A_wb)}
\def\reg{regression}
\def\sym{uniform}
\def\aasym{non-uniform}

\def\hys{\hy_*}
\def\mds{* =\nm, \HT, \haj}
\def\hynm{\hy_\snm}
\def\hyhaj{\hy_\shaj}

\def\df{{\dg,\fisher}}
\def\dl{{\dg,\lin}}

\def\dgs{\dg =  \wls, \ag}
\def\dgssi{\dg \in\{ \ols, \wls, \ag\}}
\def\dgss{\dg = \ols, \wls, \ag }
\def\dg{\dagger}
\def\wlsss{\textsc{wls}}
\def\olsss{\textsc{ols}}
\def\cc{\textup{c}}
\def\ac{A_\cc}
\def\bc{B_\cc}

\def\uxyz{T_{xY}(z)}
\def\hx{\hat x}
\def\dmd{\diamond}

\newcommand{\bbo}[1]{1_{#1\times #1}}
\newcommand{\oew}[1]{{ \overline{ \ep_{w\cdot}^{#1}}}}

\newcommand{\oa}[1]{{ \overline{\alpha^{#1}}}}

\def\mn{\mathcal{N}}
\def\ep{\epsilon}
\def\fr{{fully-interacted aggregate regression after adjusting for the whole-plot sizes} }

\def\hy{\hat Y}
\def\begini{\begin{itemize}}
\def\endi{\end{itemize}}

\def\hyht{\byht}
\def\vari{\var_\infty}

\def\by{\bar Y}
\def\oufzw{{{\overline{U^4_{w\cdot}(z)}}}}
\def\oyfzw{{{\overline{Y^4_{w\cdot}(z)}}}}

\def\oxfw{{ \overline{\|x_{w\cdot}\|_4^4}}}

\def\fts{fitting schemes}
\def\wbs{\{(w,b): w = \ot{W}; \ b = 0, 1\}}

\def\oau{\overline{\alpha U(z)}}
\def\oaup{\overline{\alpha U(z')}}

\def\olss{\textsc{ols} }
\def\wlss{\textsc{wls} }

\def\oat{\overline{\alpha^2}}
\def\sm{\text{sm}}

\def\sumsz{\sum_{s:Z_{ws}=z}}

\def\pre{Under the $2^2$ split-plot randomization and Condition  \ref{asym}, we have}

\def\prex{Under the $2^2$  split-plot randomization and Conditions \ref{asym}--\ref{asym2}, we have}
\def\wlw{{\wls,w}}
 
\def\wlws{{\wls,\ws}}
\def\olw{{\ols,w}}
\def\olws{{\ols,\ws}}
\def\tlw{{\ag,w}}

\def\dpp{P}
\def\pphi{\mu}

\def\zws{Z_{\ws}}

\def\ms{\mathcal{S}}
\def\dt{D^\T}

\newcommand{\tbeta}{\tilde\beta}

\newcommand{\ws}{{ws}}
\newcommand{\uws}{U_{ws}}

\newcommand{\lz}{\Lambda_z}
\newcommand{\hsxx}{\hat S_{xx}(z)}
\newcommand{\hsxy}{\hat S_{xY(z)}}
\newcommand{\hqxx}{\hat Q_{xx}(z)}
\newcommand{\hqxy}{\hat Q_{xY(z)}}

\newcommand{\yhtg}{\hY_{\sht,\bG}}

\newcommand{\ysg}{\hY_{*,\bG}}
\newcommand{\ssg}{S_{*,\bG}}
\newcommand{\psig}{\Psi_\bG}
\newcommand{\shtg}{S_{\sht,\bG}}
\newcommand{\shajg}{S_{\shaj,\bG}}
\newcommand{\swg}{S_{w,\bG}}

\newcommand{\hvsg}{\hV_{*,\bG}}

\newcommand{\hvhtg}{\hV_{\sht,\bG}}

\newcommand{\too}{\tilde\Omega}
\newcommand{\tgw}{\tg_\ip}
\newcommand{\tgo}{\tg_\ols}
\newcommand{\tgu}{\tg_{\ag}}

\newcommand{\tgwz}{\tg_{\ip,z}}
\newcommand{\tgoz}{\tg_{\ols,z}}
\newcommand{\tguz}{\tg_{\ag,z}}

\newcommand{\yws}{Y_{ws}}
\newcommand{\byw}{\bY_w}

\newcommand{\uw}{U_w}

\newcommand{\mt}{{\mT}}

\newcommand{\hV}{\hat{ V}}

\newcommand{\ta}{{T_\A}}
\newcommand{\tb}{{T_\B}}
\newcommand{\tainv}{T_\A^{-1}}
\newcommand{\tbinv}{T_\B^{-1}}

\newcommand{\haj}{ { \textup{haj} }}
\newcommand{\HT}{ {\textup{ht}}}
\newcommand{\nm}{ {\textup{sm} }}
\newcommand{\ols}{ {\textup{ols}} }
\newcommand{\wls}{{\textup{wls}}}

\newcommand{\shaj}{{\textup{haj}}}
\newcommand{\sht}{{\textup{ht}}}

\newcommand{\ip}{{\wls}}
\def\ag{{\textup{ag}}}

\newcommand{\htf}{Horvitz--Thompson}

\newcommand{\sshaj}{S_\shaj}
\newcommand{\ssht}{S_\sht}

\newcommand{\hvhaj}{\hV_\shaj}

\newcommand{\snm}{{\textup{\scriptsize sm}}}

\newcommand{\op}{{ o_\mathbb{P}(1)}}

\newcommand{\A}{{\textsc{a}}}

\newcommand{\B}{{\textsc{b}}}

\newcommand{\AB}{{\textsc{ab}}}

\newcommand{\of}{{\ols,\fisher}}
\newcommand{\od}{{\ols,\dmd}}
\newcommand{\wf}{{\ip,\fisher}}
\newcommand{\tf}{{\ag,\fisher}}
\newcommand{\tl}{{\ag,\lin}}

\newcommand{\ol}{{\ols,\lin}}
\newcommand{\wl}{{\ip,\lin}}

\newcommand{\tauA}{\tau_\textsc{a}}
\newcommand{\tauB}{\tau_\textsc{b}}
\newcommand{\tauAB}{\tau_\textsc{ab}}

\newcommand{\HN}{H_{0\textsc{n}}}
\newcommand{\HF}{H_{0\textsc{f}}}

\newcommand{\bb}{ \gamma}

\newcommand{\bbzp}{\bb_{\ip,z}}

\newcommand{\bbzu}{\bb_{\ag,z}}

\newcommand{\tgz}{\tilde\gamma_{\ols,z}}

\newcommand{\tgzu}{\tilde\gamma_{\ag,z}}
\newcommand{\tgzs}{\tilde\gamma_{\dg,z}}

\newcommand{\tbV}{\tilde{ V}}

\newcommand{\tbO}{\tilde{ \Omega}}

\newcommand{\vlsu}{\tilde{ V}_\ag}

\newcommand{\blsx}{\tbeta_{\ols, \lin}}

\newcommand{\blsxu}{\tbeta_\tl}

\newcommand{\blsxf}{\tbeta_{\ols, \fisher}}
\newcommand{\blsxHTf}{\tbeta_{\ip, \fisher}}
\newcommand{\blsxuf}{\tbeta_{\ag,\fisher}}

\newcommand{\bls}{\tbeta_\ols}

\newcommand{\blsHT}{\tbeta_\ip}

\newcommand{\blsu}{\tbeta_\ag}

\newcommand{\yHTn}{\hY_\sht}

\newcommand{\yhajn}{\hY_\shaj}

\newcommand{\bs}{ S}

\newcommand{\tep}{ \epsilon}

\newcommand{\te}{e}

\newcommand{\he}{\hat e}

\newcommand{\fisher}{\textsc{f}}
\newcommand{\lin}{\textsc{l}}

\newcommand{\htau}{\hat{\tau}}

\newcommand{\Uxx}{ \uxx}

\newcommand{\uxx}{ T_{  x x}}

\newcommand{\Wz}{\mathcal{W}(z)}

\newcommand{\Dw}{ D_{\ag,w}}

\newcommand{\sumz}{\sum_{z \in \mT}}
\newcommand{\sumwz}{\sum_{w\in \Wz}}
\newcommand{\sumws}{\sum_{\ws \in\ms }}
\newcommand{\sumwszz}{\sum_{\ws \in \uz}}

\DeclareMathOperator{\diag}{diag}

\newcommand{\hY}{\hat Y}

\newcommand{\hU}{\hat U}

\newcommand{\hbY}{\hat{ Y}}

\newcommand{\bsx}{{ S}_{ x x}}

\newcommand{\sxxw}{\bs_{ x x,w}}

\newcommand{\sxyzw}{\bs_{ xY(z),w}}

\newcommand{\sx}{\bs_{ x x}}

\newcommand{\sxyz}{\bs_{ xY(z)}}
\newcommand{\sxypz}{\bs_{ xY'(z)}}

\newcommand{\bsw}{ S_w}

\newcommand{\swx}{ S_{ x x,w}}

\newcommand{\hbt}{\hat { \tau}}
\newcommand{\mA}{\mathcal{A}}
\newcommand{\mZ}{\mathcal{Z}}
\newcommand{\mN}{\mathcal{N}}
\newcommand{\frt}{\textsc{frt}}

\newcommand{\sumw}{\sum_{w=1}^W}
\newcommand{\sums}{\sum_{s=1}^{M_w}}

\newcommand{\tbt}{\tilde{ \tau}}

\newcommand{\bG}{{\gamma}}

\newcommand{\ynm}{\hY_\snm(z)}
\newcommand{\yhaj}{\hY_\shaj  (z)}
\newcommand{\yht}{\hY_\sht (z)}
\newcommand{\yhtn}{\hY_\sht}
\newcommand{\yHT}{\hY_\sht (z)}
\newcommand{\bynm}{\hbY_\snm}

\newcommand{\byhaj}{\hbY_\shaj}
\newcommand{\byht}{\hbY_\sht}
\newcommand{\byH}{\hbY_*}

\newcommand{\oneHT}{\hat 1_{\sht}(z)}

\newcommand{\var}{\textup{var}}

\newcommand{\sw}{S_w}

\newcommand{\pr}{{\mathbb{P}}}

\newcommand{\uz}{\mathcal{S}(z)}
\newcommand{\Zws}{Z_{ws}}
\newcommand{\Yws}{Y_{ws}}

\newcommand{\Mwb}{M_{wb}}
\newcommand{\Ywszz}{Y_{ws}(z)}

 \newcommand{\pa}[1]{\medskip\noindent\textbf{#1.}}
 \newcommand{\ot}[1]{1, \ldots,#1}
 \newcommand{\zt}[1]{0, \ldots,#1}

\newcommand{\bp}{\Pi }

\newcommand{\itemc}{\item[$\cdot$]}

\newcommand{\hbb}{\hat{ \gamma}}

\newcommand{\tg}{\tilde\gamma}

\def\T{{ \mathrm{\scriptscriptstyle T} }}

\newcommand{\bbU}{\bar{ U}}

\newcommand{\cov}{\textup{cov}}

\newcommand{\bbx}{\bar{ x}}

\newcommand{\bt}{ \tau}

\newcommand{\hbx}{\hat{ x}}
\newcommand{\hbv}{\hat{ v}}

\newcommand{\bY}{{\bar Y}}

\newcommand{\hv}{\hat v}
\newcommand{\hbV}{\hat{ V}}

\newcommand{\bbY}{\bar{Y}}

\newcommand{\Qxx}{ Q_{  x x}}
\newcommand{\qxx}{ Q_{  x x}}
\newcommand{\qxy}{ Q_{  xY}}
\newcommand{\qxyz}{ Q_{  xY(z)}}

\newcommand{\sxx}{ S_{ x x}}

\newcommand{\mT}{\mathcal{T}}
\newcommand{\mF}{\mathcal{F}}

\def\begina{\begin{eqnarray*}}
\def\enda{\end{eqnarray*}}
\def\beginy{\begin{eqnarray}}
\def\endy{\end{eqnarray}}
\def\begine{\begin{enumerate}}
\def\ende{\end{enumerate}}

\usepackage[utf8]{inputenc}
\usepackage[english]{babel}

\begin{document}

\begin{frontmatter}
\title{Reconciling design-based and model-based causal inferences for split-plot experiments}
\runtitle{Design- and model-based inferences for split-plot Experiments}

\begin{aug}
\author[A]{\fnms{Anqi} \snm{Zhao}\ead[label=e1]{staza@nus.edu.sg}}
\and
\author[B]{\fnms{Ding} \snm{Peng}\ead[label=e2,mark]{pengdingpku@berkeley.edu}}
\address[A]{Department of Statistics and Data Science,
National University of Singapore,
\printead{e1}}

\address[B]{Department of Statistics,
University of California, Berkeley,
\printead{e2}}
\end{aug}

\begin{abstract}
The split-plot design arose from agricultural science with experimental units, also known as the sub-plots, nested within groups known as the whole-plots. It assigns different interventions at the whole-plot and sub-plot levels, respectively, providing a convenient way to accommodate hard-to-change factors. By design, sub-plots within the same whole-plot receive the same level of the whole-plot intervention, and thereby induce a group structure on the final treatment assignments. 
A common strategy is to run an ordinary-least-squares (\olsss) regression of the outcome on the treatment indicators coupled with the robust standard errors clustered at the whole-plot level. It does not give consistent estimators for the treatment effects of interest when the whole-plot sizes vary. 
Another common strategy is to fit a linear mixed-effects model of the outcome with Normal random effects and errors. 
It is a purely model-based approach and can be sensitive to violations of  the parametric assumptions. 
In contrast, the design-based inference assumes no outcome models and relies solely on the controllable randomization mechanism determined by the physical experiment. We first extend the existing design-based inference based on the {\htf} estimator to the Hajek estimator, and establish the finite-population central limit theorem for both under split-plot randomization. We then reconcile the results with those under the model-based approach, and propose two regression strategies, namely (i) the weighted-least-squares (\wlsss)  fit of the unit-level data based on the inverse probability weighting and (ii)  the \olss fit of the aggregate data based on whole-plot total outcomes,  to reproduce the Hajek and {\htf} estimators, respectively. 
This, together with the asymptotic conservativeness of the corresponding cluster-robust covariances for estimating the true design-based covariances as we establish in the process, justifies the validity of the {\reg} estimators for design-based inference. 
In light of the flexibility of regression formulation for covariate adjustment, we further extend the theory to the case with covariates and demonstrate the efficiency gain by regression-based covariate adjustment via both asymptotic theory and simulation. Importantly, all design-based properties of the  proposed {\reg} estimators hold regardless of how well the regression equations represent the true data generating process.
\end{abstract}

\begin{keyword}[class=MSC]
\kwd[Primary ]{62K15}
\kwd{62J05}
\kwd{62G05}
\end{keyword}

\begin{keyword}
\kwd{Cluster randomization}  
\kwd{cluster-robust standard error}
\kwd{covariate adjustment}
\kwd{inverse probability weighting}
\kwd{potential outcome}
\kwd{randomization inference}
\end{keyword}

\end{frontmatter}

\section{Introduction}
%
The split-plot design originated from agricultural experiments \citep{yates, yates1937design}  
and affords a convenient way to accommodate hard-to-change factors.
It remains among the most popular designs in industrial and engineering applications \citep{jones}, and is gaining increasing  popularity in social sciences \citep[e.g.,][]{Olken2007,Chong}. It also has deep connections with causal inference with interference \citep[e.g.,][]{HH, BF, imai2020causal}. 

Split-plot randomization subjects all units within the same group to the same level of the whole-plot intervention and poses challenges to subsequent inference of the treatment effects. 
Model-based analyses often require strong assumptions on the  functional forms of the outcome models, 
for example, \citet{LZ}'s marginal model, and sometimes impose additional assumptions on the distributions of the error terms and random effects, for example, the mixed-effects model \citep[e.g.,][]{kempthorne1952design, cox2000theory, wh}.  
The design-based inference, on the other hand, assumes no outcome models and draws its  justification solely from the  randomization mechanism.
\citet{kempthorne1952design} and \citet{hinkelmann2008design} initiated the discussion on the design-based inference for split-plot designs under the assumption of additive treatment effects. Due to the complexity of the randomization distributions, they invoked additional parametric assumptions for statistical inference. 
Under the purely design-based inference framework, 
\cite{ZDMD} discussed causal inference for  split-plot designs without assuming additive treatment effects, and developed the finite-sample exact theory for {\it \sym} split-plot designs where all groups are of the same size and have an equal number of units under each level of the sub-plot intervention.
\cite{MD} extended the discussion to possibly {\aasym} variants, and considered the {\htf} estimator that guarantees unbiased inference.
Both works, however, focused on the finite-sample properties of the proposed estimators and  left their asymptotic distributions, as the theoretical basis for statistical inference, an open question.
To fill this gap, we extend the discussion to the Hajek estimator under possibly {\aasym} split-plot randomization, and derive the asymptotic distributions of both the {\htf} and Hajek  estimators thereunder based on a martingale central limit theorem 
\citep{ohlsson}. 
The result includes the sample-mean estimator under {\sym} designs as a special case, and justifies the design-based inference of possibly {\aasym} split-plot experiments based on the {\htf} and Hajek estimators, respectively. This constitutes our first contribution.

In addition, \cite{ZDMD} made a heuristic link between the design-based inference and the regression-based inference  in the context of  {\sym} split-plot designs, motivating with the idea of the {\it derived linear model} \citep{kempthorne1952design, hinkelmann2008design}. 
We extend their discussion to possibly {\aasym} split-plot designs, and propose two regression formulations to reproduce the Hajek and {\htf} estimators from least squares, respectively. In particular, we demonstrate that the Hajek estimator is numerically identical to the coefficient from the weighted-least-squares (\wlsss) fit with unit data based on the inverse probability of treatment, and the  {\htf} estimator is numerically identical to the coefficient from the ordinary-least-squares (\olsss) fit with aggregate data based on the whole-plot totals. More interestingly, we show that the associated cluster-robust  covariances \citep{LZ} are asymptotically conservative for the true design-based sampling covariances of the Hajek and {\htf} estimators, respectively. 
These results justify the corresponding 
regression-based inferences for split-plot data  from the design-based perspective. 
Although the regression procedures were originally motivated by some outcome modeling assumptions,  their design-based properties hold independent of those assumptions as long as the data arise from the split-plot design. 
The analysis as such is justified by the design of the experiment  rather than the modeling assumptions. 
This constitutes our second contribution on the unification of the  model-based and design-based inferences for split-plot experiments.

Last but not least, the regression formulation offers a flexible way to incorporate covariate information, and promises the opportunity to improve asymptotic  efficiency under complete randomization \citep{Fisher35, Lin13}. 
We extend the discussion to possibly {\aasym} split-plot randomization, and 
establish the design-based properties
of the additive and fully-interacted formulations for covariate adjustment under the unit and aggregate regressions, respectively. 
The \olss estimator based on the fully-interacted aggregate regression, as it turns out, ensures the highest asymptotic efficiency when (i) covariates are relatively homogeneous within whole-plots and (ii) we include the whole-plot size factor as an additional covariate. 
The additive formulation, on the other hand, affords an alternative when the number of whole plots is small. This constitutes our third contribution on the design-based justification of  regression-based covariate adjustment.

We start with the $2^2$ split-plot design to lay down the main ideas, and then extend the results to general factors of multiple levels. Our paper furthers the growing literature on design-based causal inference with various types of experimental data \citep[e.g.,][]{Neyman23, neyman1935statistical, kempthorne1952design, box1955permutation, wu1981robustness, ObsRosenbaum, HH, imai2009essential, Schochet10, Lin13, PostStratYu, sabbaghi2014comments,  DasFact15, MiddletonCl15, CausalImbens, ji2017randomization, DingCLT, fogarty2018mitigating, fogarty2018regression, BF, mukerjee2018using, liu2019regression, abadie2020sampling, pashley,  schochet2021design, su}.

We use the following notation for convenience. 
Let $0_{m}$ and $0_{m\times n}$ be the $m\times 1$ vector and $m \times n$ matrix of  zeros, respectively. 
Let  $1_{m}$ and $1_{m\times n}$ be the $m\times 1$ vector and $m\times n$ matrix of ones, respectively.
Let $I_{m}$ be the $m\times m$ identity matrix. 
We suppress the dimensions when they are clear from the context.
Let $\otimes$ and $\circ$ denote the Kronecker and Hadamard products of matrices, respectively.
Let $1(\cdot)$ be the indicator function.   
Let $\var_\infty$ and $\cov_\infty$ denote  the asymptotic variance and covariance, respectively. 
We use $Y_i \sim x_i$ to denote the least-squares regression of $Y_i$ on $x_i$ and focus on the associated cluster-robust covariance for inference  motivated by \cite{IA}, \cite{BF}, \cite{imai2020causal}, and \cite{su}. 
The terms ``regression'' and ``cluster-robust covariance'' refer to the numeric outputs of the least-squares fit free of any modeling assumptions; we evaluate their 
properties under the design-based framework.

%
\section{Setting}

\subsection{Motivating examples}\label{sec:ex}
%

Consider a study with two interventions or factors of interest  
and a study population nested in different groups.
The split-plot design assigns the two factors 
at the group and unit levels, respectively, providing a convenient way to accommodate hard-to-change factors. 
We give below two examples from neuroscience and economics to add intuition. 

\begin{example}\label{ex:meon}
\citet{pten} conducted a randomized experiment  on $14$ mice   to study the effects of fatty acid delivery and Pten knockdown on soma size of neurons in the brain. 
The randomization of fatty acid delivery was conducted at the mouse level and randomly assigned the mice to three levels of exposure. 
The randomization of Pten knockdown took place at the neuron level and randomly infected neurons within each mouse with an shRNA against Pten or an mCherry control. 
The outcome of interest was measured by the soma size of the neurons following the treatments. 
The number of neurons extracted from each mouse varied depending on the level of infection of each virus. 
This defines a {\aasym} split-plot experiment with fatty acid delivery and Pten knockdown as the whole-plot and sub-plot factors, respectively. 
\end{example}

\begin{example}\label{ex:jpe}
\citet{Olken2007} conducted a randomized experiment on $608$ villages in Indonesia  
to study the effects of two interventions  
on reducing corruption: 
increasing the probability
of external audits (``audits'') and
 increasing participation in accountability
meetings (``participation'').
The villages are nested in subdistricts, which typically contain between $10$ and $20$ villages. 
The randomization for audits was
clustered by subdistrict such that all study villages in a subdistrict
received audits or none did  to circumvent interference. 
The randomization for participation encouragement measures, on the other hand, was done village by village,  and randomly assigned the villages to three levels of encouragement.  
This defines a {\aasym} split-plot experiment with audits and participation  constituting the whole-plot and sub-plot factors, respectively. 
\end{example}

%
\subsection{$2^2$ split-plot randomization, potential outcomes, and causal estimands}\label{me}
%
To simplify the presentation, we start with the $2^2$ split-plot design with two binary factors of interest, $\text{A}, \text{B} \in \{0,1\}$. 
Consider a study population of $N$ units,  $\mathcal{S} = \{\ws: w=1, \dots, W; \ s= \ot{M_w}\}$, nested in $W$ groups  of possibly different sizes,  $M_w \ (w = \ot{W}; \ \sumw M_w = N)$. 
A $2^2$ split-plot design compounds a cluster randomization with a stratified randomization and assigns the treatments in two steps:
\begin{enumerate}
\item
[(I)]
the first step features a cluster randomization and assigns completely at random $W_a$ groups to receive level $a \in \{0,1\}$  of factor A  with $W_0+W_1 = W$;
\item
[(II)]
the second step then runs a stratified randomization and assigns completely at random $M_{wb}$ units in group $w$ to receive level $b\in\{0,1\}$  of factor B   with $M_{w0}+M_{w1} = M_w$  for $w=1,\ldots, W$.
\end{enumerate}
Refer to the first and second steps as the stage (I) and stage (II) randomizations, respectively, with $\{W_a, M_{wb}: a, b=0,1; \ w = \ot{W}\}$ being some prespecified, fixed integers that satisfy $W_a \geq 2$ and $M_{wb} \geq 2$. 
The final treatment received by a unit is the combination of the level of factor A  its group receives in stage (I) and the level of factor B the unit itself receives in stage (II),  indexed by  $(a, b) \in \mt = \{(0,0), (0,1),(1,0),  (1,1) \}$; we abbreviate $(a,b)$ as $(ab)$  when no confusion would arise. 
Refer to each unit as a sub-plot and each group as a whole-plot by convention of the literature on agricultural  experiments. 
Factors A and  B become the whole-plot and sub-plot factors, respectively.

Let $A_{ws} = A_w$ and $B_{ws}$ indicate the levels of factors A and B received by sub-plot $\ws$, respectively, with   $\pr(A_w =a) = W_a/W =p_a$ and $\pr(B_{ws} = b) = M_{wb}/M_w =q_{wb}$ for $a,b=0,1$. 
We suppress the subscript $s$ in  $A_{ws}$ to highlight its identicalness over all sub-plots within the same whole-plot. 
Let 
$\Zws = (A_w, B_{ws})$ indicate the final treatment for sub-plot $\ws$ with
\begina
p_{ws}(z) = \pr(\Zws = z) = \pr(A_w =a) \cdot \pr(B_{ws} = b) = p_a q_{wb} \qquad \text{for} \ \ z=(ab)\in\mt. 
\enda
Refer to $p_\ws(z)$ as the inclusion probability of sub-plot $\ws$ to receive treatment $z = (ab)$. 
It is identical for units in the same whole-plot yet varies across different whole-plots unless $q_{wb}=q_{w'b}$ for $w \neq w'$.

Let $\bar M = N/W$ be the average whole-plot size, and let $\alpha_w = M_w/\bar M$ be the {\it whole-plot size factor} with $\bar \alpha = W^{-1}\sumw \alpha_w = 1$. The sample size under treatment $z= (ab)$ equals $N_z = \sum_{w: A_w = a}M_{wb} $ and is in general stochastic unless $M_{wb}$ is identical across all $w$. 
 We call a split-plot design {\it \sym} if $M_w$ and $ \{M_{wb}: b = 0, 1\}$   are  identical across all $w  = \ot{W}$.

\begin{condition}\label{balanced}
$M_w = M$ and $M_{wb} = M_b$ for all $w=\ot{W}$ and $b = 0,1$.  
\end{condition}

Condition \ref{balanced} ensures that $p_{ws}(z) = p_a q_b  = N_z/N$ is identical for all $ws \in \ms$, with $z = (ab)$, $q_b = M_b/M$ and $N_z = Np_aq_b$. 
\citet{ZDMD} focused on the {\sym}  $2^2$ split-plot design and  established the unbiasedness of the sample-mean estimator. 
Most real-world experiments in social and biomedical sciences, however, are not {\sym} (see, for example, Examples \ref{ex:meon} and \ref{ex:jpe}).

Let $\Ywszz$ be the potential
outcome of sub-plot $\ws$ if assigned to treatment $z$. 
Let $\bar Y(z) = N^{-1} \sumws  \Ywszz$ be the finite-population average, vectorized as $\bbY = (\bY(00), \bY(01),\bY(10),  \bY(11) )^\T$.
Contrasts 
\beginy\label{taud}
\tauA &=& 2^{-1} \{\bY(11)+\bY(10)\} -  2^{-1} \{\bY(01)+\bY(00)\} ,\nonumber\\  
\tauB  &=& 2^{-1} \{\bY(11)+\bY(01)\} -  2^{-1} \{\bY(10)+\bY(00)\},\\ 
\tauAB  &=&  \bY(11) - \by(10) - \by(01) + \by(00) \nonumber
\endy
define the  standard main effects and interaction under $2^2$ factorial designs. Sometimes the interaction is also defined as $\tauAB/2$ \citep{DasFact15}; 
 the difference causes no essential change to our discussion. 
We will discuss inference of the general estimand
$$
\bt  =  G \bbY
$$ 
for arbitrary  coefficient matrix $G$. The standard main effects and interaction in \eqref{taud}  correspond to a special $G = G_0 = (g_{ \A}, g_{ \B}, g_{ \AB})^\T$ with $g_{\A} = 2^{-1} (-1,-1,1,1)^\T$,  $g_{\B}  = 2^{-1} (-1,1,-1,1)^\T$, and $g_{\AB} = (1,-1,-1,1)^\T$.

We consider the  design-based inference of  $\tau$, which conditions on the potential outcomes and views the treatment assignment as the sole source of randomness. 
The observed outcome equals $\Yws = \sumz 1(\Zws = z) \,\Ywszz$ for sub-plot $\ws$. We focus on estimators of the form $\htau = G\hy$, where $\hY$ is some estimator of $\bY$ based on $(Y_{ws}, Z_{ws})_{ws\in\ms}$  and possibly some pre-treatment covariates.
Assume that $a, b \in \{0,1\}$ index the levels of factors A and B in treatment combination $z \in\mt$ throughout unless specified otherwise.
\section{{\htf} and Hajek estimators}\label{est}
We review in this section three design-based estimators for estimating $\bY$. 
Let $\uz =\{ \ws: \Zws = z,\ \ws \in\mathcal{S} \}$ be the set of sub-plots under treatment $z$. 
Let $\Wz$ be the set of whole-plots that contain at least one observation under treatment $z$.
The restriction in the stage (I) randomization ensures that there are only two treatment levels, namely $z =   (A_w0)$ and $z= (A_w1) $, observed in whole-plot $w$.  By definition, $\Wz  =\{w:A_w=a\}$ with $|\Wz| = W_a$  for $z \in \{(a0), (a1)\}$ with level $a$ of factor A.

First, the sample-mean estimator of $\by(z)$ equals 
\begina
\ynm =  |\uz|^{-1} \sum_{\ws \in \uz}   \Yws = |\uz|^{-1}\sumws 1(Z_{\ws}=z) \, \Yws, 
\enda
averaging over all units under treatment $z$. 
It is neither unbiased nor consistent in general.

Second, the {\htf} estimator is unbiased  for $\bar Y(z)$:
\beginy\label{eq:yht}
\yHT
= N^{-1} \sum_{\ws \in \uz} p_{ws}^{-1}(z)\, \Yws =N^{-1} \sumws  \frac{ 1(Z_{ws}=z )}{ p_{ws}(z) } \, Y_{ws}(z).
\endy
Split-plot randomization ensures $p_{ws}(z) =p_aq_{wb}$ for $z = (ab)$, 
and simplifies \eqref{eq:yht} 
to 
\begina
\hat Y_\sht(z)
= 
W_a^{-1} \sumwz   \alpha_w \hY_w(z),
\enda
where $\hY_w(z) = M_{wb}^{-1} \sumsz\Yws$  is the whole-plot sample mean under treatment $z = (ab)$. 
Let $
\bY_w(z) = M_w^{-1} \sums \yws(z)$ be the whole-plot average potential outcome, as the population analog of $\hY_w(z)$. 
Let $U_w(z) = \bar M^{-1} \sums Y_{ws}(z) = \alpha_w \bar Y_w(z)$ be the scaled whole-plot total potential outcome with sample analog $
\hU_w(z) =\alpha_w \hY_w(z)
$. 
We have 
\beginy\label{hat_Y}
 \bar Y(z)  =W^{-1}\sumw \alpha_w \bY_w(z) = W^{-1}  \sumw U_w (z), \qquad \yHT  =   W_a^{-1}  \sumwz \hU_w (z) .
\endy 
This illustrates $\yHT$ as a two-stage sample-mean estimator  of $\bar Y(z)$ 
by first using $\hat U(z) = W_a^{-1}  \sumwz U_w (z)$ to estimate $\bY(z) = W^{-1}  \sumw U_w (z)$ and then using $\hU_w(z)$ to estimate the $U_w(z)$ in $\hat U(z)$ for $w \in \Wz$. 
Standard results ensure that the two steps are unbiased with regard to the whole-plot and sub-plot randomizations, respectively.

A main criticism of  the {\htf} estimator is that it is not invariant to location shifts in general \citep[e.g.,][]{Fuller, MiddletonCl15, su}.
In contrast, the Hajek estimator 
\begina
\yhaj
=     \frac{\yHT}{ \oneHT}, \qquad \text{where} \ \ \oneHT = N^{-1} \sum_{\ws \in \uz} p_{ws}^{-1}(z), 
\enda
normalizes the {\htf} estimator by the sum of the individual weights involved in its definition, and ensures location invariance by construction.
%
%
We can view $\oneHT$ as the {\htf} estimator of constant $1$ when all potential outcomes equal 1.
The Hajek estimator is thus a ratio estimator for $\bar Y(z) = \bar Y(z) / 1$ with the numerator and denominator estimated by $\yHT$ and $\oneHT$, respectively.

This gives us three estimators of $\by$, denoted by $\byH = (\hY_*(00), \hY_*(01),\hY_*(10),\hY_*(11))^\T$ for $\mds$. 
They differ in general but coincide under  {\sym} split-plot designs.

\begin{proposition}\label{prop:hY_bal}
Under Condition \ref{balanced}, we have $\hynm = \hyht = \hyhaj$ with
$$\ynm = \yHT = \yhaj = W_a^{-1} \sumwz \hY_w(z) \qquad \text{for} \ \ z =(ab)\in \mt.$$  
\end{proposition}

We derive the design-based properties  of $\hys \ (\mds)$  under split-plot randomization  in the next section.

\section{Design-based properties under split-plot randomization}\label{sec:design}
\subsection{Finite-sample results for the {\htf} estimator}
Define the scaled between and within whole-plot covariances of $\{\Yws(z), \Yws(z')\}_{\ws \in \mathcal{S}}$ as
\begin{align*}
S(z,z') &=  (W-1)^{-1} \sumw \big\{\alpha_w \bY_w(z)- \bar Y(z) \big\}\big\{\alpha_w \bY_w(z')- \bar Y (z') \big\},\nonumber\\
\sw(z, z') &=  (M_w-1)^{-1} \alpha_w^2\sum_{s=1}^{M_w}  \big\{\Yws(z)- \bY_w(z) \big\}\big\{  \Yws(z')- \bY_w(z')  \big\} ,
\end{align*}
respectively \citep{MD}, 
summarized in 
$S = ( S(z, z')  )_{z,z'\in\mt}$ and $\bsw = ( \sw(z, z') )_{z,z'\in\mt}
$. 
They  measure the between and within whole-plot heterogeneity in potential outcomes after adjusting for the whole-plot sizes. 
Let 
$$H  = \diag(p_0^{-1}, p_1^{-1}) \otimes \bbo{2} - \bbo{4}, \quad  H_w =  \diag(p_0^{-1}, p_1^{-1}) \otimes 
\{ \diag(   q_{w0}^{-1}, q_{w1}^{-1} ) - \bbo{2}\}$$
be two symmetric $4\times 4$ matrices defined by the design parameters. 
 Lemma \ref{Vmat} below quantifies the sampling covariance of $\byht$ in finite samples.

\begin{lemma}\label{Vmat}
Under the $2^2$ split-plot randomization, we have
$$E(\byht) = \bbY, \qquad  \cov( \byht )  =W^{-1}(  H \circ \bs_\sht +\Psi) $$
with $S_\sht  = S$ and $\Psi = W^{-1}\sum_{w=1}^W M_w^{-1} (H_w \circ \bsw)$.
\end{lemma}

Consider $\Psi$ as a summary of $(S_w)_{w=1}^W$ after adjusting for the whole-plot sizes. 
Lemma \ref{Vmat} decomposes the variability in $\byht$ into that due to the stage (I) randomization, namely $ W^{-1}(H \circ \ssht)$, and that due to the stage (II) randomization, namely $W^{-1}\Psi$. 
Lemma S3 in the Supplementary Material quantifies this statement rigorously. 
A direct implication is $\var( g^\T \hY_\sht ) = W^{-1} g^\T (  H \circ \bs_\sht +\Psi) g$ for arbitrary $g \in \mathbb{R}^4$. This gives a more compact matrix form of \citet[][Theorem 1]{MD}. 

Quantification of the Hajek estimator is, on the other hand, hard in finite samples in general.
We thus cast the discussion under an asymptotic framework, and establish  the asymptotic Normality of $\byht$ and $\byhaj$ under split-plot randomization  in Section \ref{sec_clt}.

\subsection{Asymptotic Normality of the {\htf} and Hajek estimators}\label{sec_clt}
To facilitate the discussion, we introduce an intermediate quantity
 $$\qquad \hY_\sht'(z) =    N^{-1}  \sumwszz p^{-1}_{ws}(z) \, \yws'(z), \qquad \text{where} \ \ Y'_{ws}(z) = \yws(z) - \bY(z),$$ 
as the {\htf} estimator defined on the centered potential outcomes $Y'_{ws}(z)$. 
The difference between the Hajek estimator and the true finite-population average equals 
\beginy
\yhaj - \bar Y(z) = \frac{\hY_\sht(z) - \oneHT\bY(z)}{\oneHT} = 
\frac{ \hY_\sht'(z)}{\oneHT}. \label{eq:haj_intuition}
\endy 

Let $\sshaj = (\sshaj(z,z') )_{z,z'\in\mt}$, where $$\sshaj(z, z') = (W-1)^{-1}  \sumw \alpha_w^2   \{ \bar Y_w(z)-\bar Y(z)\} \{ \bar Y_w(z')-\bar Y(z')\},$$
 be the scaled between whole-plot covariance matrix defined on $\{Y_{ws}'(z): z\in\mt\}_{ws\in\ms}$. 
Let $\oa{k} = W^{-1}\sumw \alpha_w^{k}$ be the $k$th moment of $(\alpha_w)_{w=1}^W$ for $k= 1,2,4$ with $\bar\alpha = W^{-1}\sumw \alpha_w= 1$.  
Let $\oyfzw  
=  M^{-1}_w \sums 
 Y^4_{ws}(z)$ be the uncentered fourth moment of $Y_{ws}(z)$ in whole-plot $w$.
We state in Condition \ref{asym} below the regularity conditions for asymptotics under split-plot randomization.   
\begin{condition} \label{asym}
As $W$ goes to infinity, for $a, b = 0,1$ and $z\in \mT$, 
\begin{enumerate}[(i)]
\item\label{alpha} $\oa{2} = O(1); \ \oa{4}  =o(W)$; 
\item\label{paqb} $p_a$ has a limit in $(0,1)$;  $\epsilon \leq \min_{w = \ot{W}} q_{wb} \leq \max_{w = \ot{W}} q_{wb} \leq 1 - \epsilon$ for some $\epsilon \in (0,1/2]$ independent of $W$;
\item\label{po_1}  $\bbY$,  $\ssht$, $\sshaj$, and $\Psi$ have finite limits;   
\item\label{po} $\max_{ w = \ot{W}}  |\alpha_w \bar Y_w(z) - \bar Y(z) |^2/W =o(1)$; 
 \item\label{po_3}  $W^{-1}\sumw \alpha^2_w  \oyfzw  =O(1); \ W^{-2}\sumw \alpha^4_w  \oyfzw  =o(1)$. 
\end{enumerate}
\end{condition}

 For notational simplicity, we will also use $p_a$, $\bY$, $\ssht$, $\sshaj$, and $\Psi$ to denote their respective limiting values when no confusion would arise. 
The exact meaning should be clear from the context.  

With $ \bar \alpha  = 1$, Condition \ref{asym}\eqref{alpha} requires the finite-population variance of the $\alpha_w$'s to be uniformly bounded and thereby protects against the possibility of superlarge whole-plots. 
It also allows for diverging fourth moment yet stipulates the growth rate to be slower than $W$. 

Condition \ref{asym}\eqref{paqb}--\eqref{po_1}, on the other hand, ensure that $\cov(\byht)$ decays at the rate of $W^{-1}$ and thereby guarantee the consistency of $\byht$ for estimating $\bY$. 
We do not need $q_{wb}$ to converge but only be uniformly bounded as long as $\Psi$ has a finite limit; see Lemma S4 in the Supplementary Material for details. 
In the neuroscience experiment in Example \ref{ex:meon}, for example, this imposes bounds on the number of neurons affected by each level of the  Pten knockdown intervention.  
Condition \ref{asym}\eqref{po_3} stipulates the bounded fourth moment condition peculiar to the split-plot randomization.
Provided Condition \ref{asym}\eqref{alpha}, it is satisfied as long as the $\oyfzw$'s are uniformly bounded for all $w$. 

 Importantly, Condition \ref{asym} requires only $W$ goes to infinity and includes both of the following asymptotic regimes as special cases: 
\begin{enumerate}[(i)]
\item   $M_w$ goes to infinity for all $w = \ot{W}$; 
\item   $\{M_w: w=\ot{W}\}$ are  uniformly bounded. 
\end{enumerate}
Recall from Lemma \ref{Vmat} that $H\circ \ssht$ and $\Psi= W^{-1}\sum_{w=1}^W M_w^{-1} (H_w \circ \bsw)$ characterize the variability in $\byht$ due to the stage (I) and stage (II) randomizations, respectively.
Regime (i) ensures that $\Psi=o(1)$, and thus the variability from the stage (I) randomization dominates that from stage (II),  as long as $(S_w)_{w=1}^W$ are uniformly bounded. 
Regime (ii), on the other hand, requires $(S_w)_{w=1}^W$ to have a stable mean to ensure 
that $\Psi$ has a finite limit. A third asymptotic regime is to have $M_w$ go to infinity for all $w = \ot{W}$ while keeping $W$ fixed \citep{LH}. 
Asymptotic Normality is lost under this regime, and we omit it from the ensuing discussion. Overall, it is crucial to have large $W$ for reliable   asymptotic approximations under our framework. This requires a large number of mice in Example \ref{ex:meon} and a large number of subdistricts in Example \ref{ex:jpe}.

\begin{remark} 
Although our theory does not impose any stochastic assumptions on the potential outcomes, we can invoke a working model to gain intuition for the requirements on $\ssht =O(1)$ and $\sshaj=O(1)$ in Condition \ref{asym}\eqref{po_1}.
Consider $N=WM$ units in $W$ equal-sized groups, $\{\ws:w = 1,\dots, W; \ s= 1, \dots, M\}$, with $\Yws(z) \sim [\mu_w, \sigma^2]$ where $\mu_w \sim [\mu_0, \sigma_0^2]$. 
This defines a classical model for characterizing  data nested in clusters.
Denote by $\pr'$ the probability measure induced by the potential outcomes generating process. 
Standard result shows that $S_\sht =O_{\pr'}(1)$ as $W$ and $M$ go to infinity, and  degenerates to $S_\sht = o_{\pr'}(1)$ if $\sigma_0=0$ and $\Yws(z) \sim [\mu_0, \sigma^2]$. 
\end{remark}

Theorem \ref{clt} below states the asymptotic Normality of $\hY_\sht$ and $\hY_\shaj$.

\begin{theorem}\label{clt}
Let 
$\Sigma_* = H \circ S_* +\Psi$ for $*=\HT, \haj$ with $ \Sigma_\sht = W\cov(\byht)$ and $ \Sigma_\shaj =W\cov(\byht')$ in finite samples. 
{\pre}
\begina
\sqrt{W} (\hY_*  -\bbY   ) \rightsquigarrow \mathcal{N}(0, \Sigma_*) \qquad \text{for} \ \ * = \HT, \haj. 
\enda
\end{theorem}

Theorem \ref{clt} ensures the consistency of $\byht$ and $\byhaj$ for estimating $\by$, and establishes $\Sigma_\shaj$ as the asymptotic sampling  covariance of $\sqrt W (\byhaj - \bar Y)$. 
The large-sample relative efficiency between $\byht$ and $\byhaj$ then follows from the comparison of $\Sigma_\sht$ and $\Sigma_\shaj$.

\begin{corollary}\label{V_hajek}
{\pre}
\begina
W\left[ \vari\big\{\yhaj\big\} -  \vari\big\{ \yHT \big\} \right] =  
(p_a^{-1} -1)  \{ S_\shaj(z,z) - S_\sht(z,z)\}  \quad \text{for} \ \ z = (ab)\in\mt
\enda 
 with
\begin{itemize}
\itemc $S_\shaj(z,z) = S_\sht(z,z)$ if (i) $\bar Y(z) = 0$ or (ii) $\alpha_w = 1$ for all $w$; 
\itemc $0=S_\shaj(z,z)   \leq  S_\sht(z,z)$ if $\bar Y_w(z) $ is constant over all $w$; 
\itemc 
  $0= S_\sht(z,z)  \leq  S_\shaj(z,z)$ if $U_w(z) = \alpha_w\bar Y_w(z) $ is constant over all $w$.
\end{itemize}
\end{corollary}

Intuitively, 
$\yhaj$ is asymptotically more efficient than $\yHT$ if the whole-plots have similar average potential outcomes; vice versa if the whole-plots have similar total potential outcomes. 
The whole-plot averages are often more homogeneous than the whole-plot totals in  realistic data generating processes.
This affords another angle for perceiving the advantage of $\yhaj$ over $\yHT$.

%
\subsection{Estimation of the sampling covariances}
%
The expressions of $\Sigma_* \ (*=\HT, \haj)$ involve unobserved potential outcomes. We need to estimate them for the Wald-type inference. 
Let 
\begin{align*}
 \hat S_\sht(z, z') &=   (W_a-1)^{-1}\sum_{w: A_w = a} \{ \alpha_w \hY_w(z) - \yHT\} \{ \alpha_w\hY_w(z') - \hY_\sht(z')\},\\
 \hat S_\shaj(z, z') &=   (W_a-1)^{-1}\sum_{w: A_w = a} \alpha_w^2 \{\hY_w(z) - \yhaj\}\{\hY_w(z') - \hY_\shaj(z')\}
\end{align*}
be the sample analogs of $S_\sht(z,z')$ and $S_\shaj(z,z')$
for $z = (ab)$ and $z'=(ab')$ that share the same level of factor A.
Split-plot randomization assigns all sub-plots within the same whole-plot to receive the same level of factor A,   and thus defies the definition of  $\hat S_\sht(z, z')$ and $\hat S_\shaj(z,z')$ for $z=(ab)$ and $z' = (a'b')$ with $a \neq a'$. 
We use
\begin{eqnarray}\label{Vs}
\hbV_*
= \left( \begin{array}{cc}  W_0^{-1} \left( \begin{array}{cc}    \hat S_*(00,00) &  \hat S_*(00,01)\\
\hat S_*(00,01)&  \hat S_*(01,01)
\end{array}\right) & 0_{2\times 2}\\
0_{2\times 2}&   W_1^{-1} \left( \begin{array}{cc}    \hat S_*(10,10) &  \hat S_*(10,11)\\
\hat S_*(10,11)&  \hat S_*(11,11)
\end{array}\right)
 \end{array} \right)
  \end{eqnarray}
to estimate the sampling covariance of  $\hY_*$ for $* = \HT, \haj$, respectively. 
\cite{MD} introduced $\hbV_\sht$, and we introduce $ \hat V_\shaj$.

\begin{theorem}\label{covEst}
{\pre}
\begina
W\hbV_* - \Sigma_* 
= S_* + \op   \qquad \text{for} \ \ * = \HT, \haj . 
\enda
\end{theorem}

\citet[][Theorem 2]{MD} implied $
 E(\hbV_\sht) - \cov(\byht)  = W^{-1}  S_\sht \geq 0$ such that $\hbV_\sht$ is a conservative estimator of $\cov(\hyht)$ in finite samples, extending \citet{ZDMD} to possibly {\aasym} $2^2$ split-plot randomization.  
Theorem \ref{covEst} extends the discussion to   finite-population asymptotics, and establishes the asymptotic conservativeness of $\hV_*$ for estimating the true sampling covariance of $\hat{Y}_*$  for $* = \HT, \haj$.
This, together with Theorem \ref{clt}, justifies the Wald-type inference of $\tau = G\bY$  based on $\htau = G \hY_*$ with estimated covariance $G\hV_*G^\T$ for $* = \HT, \haj$.

%
\section{Reconciliation with  model-based inference}\label{reconcile}
%
\subsection{Overview}
Despite the nice theoretical properties of the design-based estimators, their reception among practitioners is at best lukewarm due to the dominance of the more convenient model-based counterparts. 
Can these convenient model-based estimators match their design-based counterparts and deliver inferences that are valid from the design-based perspective? The answer is affirmative with the aid of appropriate weighting schemes and cluster-robust covariances. 

Consider 
\begin{eqnarray}
\label{lm_t}
Y_{ws}  \sim  1(Z_{ws}=00) + 1(Z_{ws}=01) + 1(Z_{ws}=10) + 1(Z_{ws}=11)
\end{eqnarray}
as a standard formulation for studying $2^2$ factorial experiments. 
We propose two general strategies,  namely the {\it inverse probability weighting} and {\it aggregate  model}, to recover the Hajek and {\htf}  estimators of $\by$ directly as coefficients from \eqref{lm_t} and its variant, respectively,
and establish the appropriateness of the associated cluster-robust covariances for estimating the true sampling covariances.
The result reconciles the {\reg} estimators with their design-based counterparts free of any modeling assumptions, and ensures the validity of the resulting inferences regardless of how well the regression equations represent the true outcome generating process.

\subsection{Least-squares estimators from  unit and aggregate regressions}
We introduce in this subsection three {\fts}, denoted by ``ols'', ``wls'', and ``\ag'', respectively,  for estimating $\bY$ from least-squares regressions, and establish their respective design-based properties under split-plot randomization. 

First, the ``ols'' fitting scheme represents the dominant choice, and takes  the  \olss coefficients from \eqref{lm_t} to estimate $\bY$. Let $\bls$ denote the resulting estimator.

Next, inspired by the use of inverse probability weighting in constructing $\yHT$ and $\yhaj$,
the ``wls'' fitting scheme weights  $\Yws$ by the inverse of its realized inclusion probability, $ p_{ws}(\Zws)$, in the least-squares fit of \eqref{lm_t},  and estimates $\bY$ by the resulting \wlss coefficients. Let   $\tbeta_\wls$ denote the resulting estimator. 
 
Finally, recall $\hU_w(z) = \alpha_w\hY_w(z)$ as an intuitive estimator of the scaled whole-plot total potential outcome $U_w(z)$  for $w \in\Wz$. 
The restriction in the stage (I) randomization ensures that there are only two treatment levels, namely $z =   (A_w0)$ and $z= (A_w1) $, observed in whole-plot $w$, resulting in a total of $2W$ whole-plot level observations: $\{ \hU_w(A_wb): w = \ot{W}; \ b = 0,1\}$. 
We propose to fit 
\begin{eqnarray}
\label{lm_t_ag}
\hU_w(A_wb)\sim 1(A_wb=00) + 1(A_wb=01) + 1(A_wb=10) + 1(A_wb=11)
\end{eqnarray}
over $\wbs$ for these $2W$ observations as an aggregate analog of \eqref{lm_t}. 
 The ``\ag'' fitting scheme takes  the resulting \olss coefficients to estimate $\bY$. Let $\tbeta_\ag $ denote the resulting estimator. 
The idea of regression based on  aggregate data appeared before: 
 \citet{BF} discussed it in a two-stage experiment  for estimating treatment effects in the presence of interference; \citet{su} recommended it for analyzing the one-stage cluster-randomized experiment. 

This gives us three {\reg} estimators, $\{\tilde \beta_\dagger: \dagger = \ols, \wls, \ag\}$,  
summarized in Table \ref{tb1_Y}. As a convention, we use the tilde symbol ($\tilde{\color{white}\tau} $) to signify outputs from least-squares fits.
Proposition \ref{prop:ls} below states their numeric equivalence with $\bynm$, $\byhaj$, and $\byht$, respectively.

\begin{table}[t]
\caption{Regression estimators of $\bbY$ under the ``\textup{ols}'', ``\textup{wls}'', and ``\textup{ag}'' \fts, respectively, along with their design-based equivalents.
The design-based properties in the last two columns are with regard to general split-plot designs. 
All six estimators coincide under {\sym} split-plot designs. 
\label{tb1_Y}}
\begin{center}
\begin{tabular}{|c|cc|cc|cc|}\hline
fitting&&& regression &  design-based & &\\ 
scheme&model&weight&estimator&equivalent & unbiased  & consistent   \\\hline
$\ols$ & \multirow{2}{*}{\eqref{lm_t}} & 1 & $\bls$ & $\bynm$ &  no & no \\
 $\wls$ &  & $\{p_{ws}(\Zws)\}^{-1}$ & $\tbeta_\wls$& $\byhaj$ &  no & yes  \\\hline
 $\ag $ & \eqref{lm_t_ag}& 1 & $\blsu$ & $\byht$ &  yes & yes    \\\hline
\end{tabular}
\end{center}
\end{table}

\begin{proposition}\label{prop:ls} 
$
\tbeta_\ols =\bynm,
$
$
\tbeta_\wls =\byhaj, 
$ and $ \tbeta_\ag = \byht.$
\end{proposition}

Proposition \ref{prop:ls} is numeric and shows the utility of inverse probability weighting and aggregate model in reproducing the Hajek and {\htf} estimators from least squares, respectively.
The correspondence between the three \fts, $\{ \ols, \wls, \ag \}$, and the three estimation schemes, $\{\nm, \haj, \HT\}$, runs through the following discussion and reconciles the model-based and design-based perspectives. 

\begin{remark}\label{rmk:alternative_weights}
Alternatively, the least-squares fit of 
$\alpha_wY_{ws}  \sim  1(Z_{ws}=00) + 1(Z_{ws}=01) + 1(Z_{ws}=10) + 1(Z_{ws}=11)$
with  weights $\alpha_w^{-1}\{p_{ws}(\Zws)\}^{-1}$ recovers the {\htf} estimator from  scaled unit-level outcomes. 
We exclude it from the discussion due to the unnaturalness in both its weighting and outcome transformation schemes. \cite{luke} reviewed an alternative \textsc{wls}  scheme in the context of stratified experiments, which corresponds to the least-squares fit of \eqref{lm_t} with weights $ N_{\Zws} /  p_{ws}(\Zws) $.
Lemma S13 in the Supplementary Material ensures that this slightly different weighting scheme leads to identical regression coefficients and cluster-robust covariance as those under the ``\wls'' fitting scheme. 
\end{remark}

A key virtue of the regression-based approach is its ability to deliver also estimators of the standard errors via the same least-squares fit. 
Of interest is how these convenient covariance estimators approximate the true sampling covariances from the design-based perspective. 
Denote by  $\tilde V_\dg $ the classic cluster-robust covariance for  $\tbeta_\dg \ (\dg = \ols, \wls, \ag  )$ from the same least-squares fit. 
Theorem \ref{thm_vHats} below shows the asymptotic equivalence of $\tilde V_\dg \ (\dgs)$  with $\hvhaj$ and  $\hbV_\sht$, respectively. 

\begin{theorem}\label{thm_vHats}
Define $\hat{1}_\sht =\diag\{\oneHT\}_{z\in\mt}$. Then 
$$
\tbV_\wls 
=
 \hat{1}^{-1}_\sht\diag\left(  \frac{W_0-1}{W_0}I_2, \frac{W_1-1}{W_1}I_2\right) \hvhaj \,  \hat{1}^{-1}_\sht,\quad 
 \tbV_\ag 
= \diag\left(  \frac{W_0-1}{W_0}I_2, \frac{W_1-1}{W_1}I_2\right)  \hbV_\sht . 
$$
Further assume Condition \ref{asym}. Then $\hat{1}_\sht = I_{4} + \op$ and thus 
$$W(\tilde V_\wls - \hV_\shaj)= \op,\qquad  W(\tilde V_\ag - \hV_\sht)= \op.$$
\end{theorem}

With $\hat{1}_\sht = I_4 + \op$, 
the asymptotic equivalence between the cluster-robust covariances and their design-based counterparts follows directly from the numeric correspondence, and ensures the asymptotic  conservativeness of $\tilde V_\dg$ for estimating the true sampling covariance of $\tbeta_\dg$ for $\dgs$.
This, together with  Proposition \ref{prop:ls}, justifies the Wald-type inference of $\tau = G \bY$ based on point estimator $\tbt_\dg = G\tbeta_\dg$ and estimated covariance $G\tilde V_\dg G^\T$ for $\dg = \wls, \ag $. 
Importantly, the cluster-robust covariance is necessary for valid regression-based inferences because the heteroskedasticity-robust covariance can be asymptotically anti-conservative.

\begin{remark}\label{rmk:hc2}
A number of other options exist for constructing cluster-robust covariances from linear models \citep{LZ, CR}. 
In particular, the HC2 variant of $\tbV_\ag$ recovers $\hbV_\sht$ exactly in finite samples \citep{hc2, BF,imai2020causal}.
The difference between the classic and HC2 estimators vanishes as the sample size goes to infinity. 
 We relegate the details to Section S4.2 of the Supplementary Material. 
 With small $W$, \cite{hc2}, \cite{CR}, \cite{tipton}, and \cite{mac} proposed various confidence intervals to achieve better finite-sample coverage properties. They are likely to improve the classic cluster-robust covariance and its HC2 variant under the design-based framework as well. We leave this to future research.  
 Lastly, $\hat V_\sht$ is unbiased for estimating $\cov(\hy_\sht)$ if $\alpha_w\bar Y_w(z)$ is identical over $w = \ot{W}$ for all $z \in \mt$. 
 Modification  to $\hbV_\sht$ is proposed by \citet{MD} that ensures unbiased estimation of $\cov(\hy_\sht)$ under a different additivity assumption. 
Alternative, likely less common, model specification is needed to recover this variant via least-squares fit. 
\end{remark}

%
\section{Regression-based covariate adjustment}\label{sec:ca}
%
\subsection{Background: covariate adjustment under complete randomization}\label{sec:review}
The regression formulation offers a natural way to incorporate covariates  to further improve the estimation efficiency. 
We briefly review the theory of covariate adjustment under complete randomization  to motivate our extension to  split-plot randomization. 

Consider a treatment-control experiment with two levels of intervention, $\mt=\{0,1\}$, and a study population of $N$ units with potential outcomes $\{Y_i(0), Y_i(1):  i = 1, \dots, N\}$. 
The finite-population average treatment effect equals $\tau = \by(1) - \by(0)$, where $\by(z) = N^{-1}\sum_{i=1}^N Y_i(z)$. 

Denote by $Z_i$ the treatment indicator of unit $i$ under complete randomization. 
The difference-in-means estimator  is unbiased for $\tau$, and equals the coefficient of $Z_i$ from the \olss fit of $Y_i \sim 1+Z_i$. 
Given covariates $ x_i = (x_{i1}, \ldots, x_{iJ})^\T$ for unit $i$ $(i=\ot{N})$, 
\citet{Fisher35} proposed to use  the coefficient of $Z_i$ from the \olss fit of $Y_i \sim 1 + Z_i + x_i$  to estimate $\tau$. 
\cite{Freedman08a} criticized its potential efficiency loss compared to the difference-in-means estimator. 
\citet{Lin13} proposed an improved estimator as 
the coefficient of $Z_i$ from the \olss fit of 
$ Y_i \sim  1 + Z_i +  (x_i - \bar x) + Z_i  (x_i - \bar x)$ with centered covariates and treatment-covariates interactions, and proved that it is at least as efficient as the difference-in-means and \cite{Fisher35}'s estimators asymptotically. 
We call \citet{Fisher35}'s regression the {\it additive} specification and \citet{Lin13}'s regression the {\it fully-interacted} specification hence when no confusion would arise.

We extend below their results to split-plot randomization. 
We will focus on four covariate-adjusted regressions depending on whether we use the unit or aggregate data to form the regression and whether we use the additive or fully-interacted specification for covariate adjustment. 
We will study their design-based properties and compare their efficiency gains over the unadjusted counterparts. 
Given $\htau_1$ and $\htau_2$ as two consistent and asymptotically Normally distributed estimators for $\tau$, we say $\hat\tau_1$ is asymptotically more efficient than $\hat\tau_2$, or equivalently, $\htau_1$ guarantees gains in asymptotic efficiency over $\htau_2$,  if $\cov_\infty(\htau_1) \leq \cov_\infty(\htau_2)$ for all possible values of $\{Y_{ws}(z): z \in\mt\}_{ws\in\ms}$, and the strict inequality holds for at least one set of $\{Y_{ws}(z): z \in\mt\}_{ws\in\ms}$.

%
\subsection{Additive regressions}\label{sec:fisher}
%

Let $x_{ws} = (x_{ws[1]}, \dots, x_{ws[J]})^\T$ be the $J\times 1$  covariate vector for sub-plot $\ws$.
Adding $x_{ws}$ to \eqref{lm_t} yields 
\begin{eqnarray}
\label{dlm4_f}
Y_{ws} \sim \sumz  1( \Zws = z) + x_{ws} \sim d_{ws} +x_{ws}
\end{eqnarray}
as the {\it additive unit regression} over $ \ws \in \ms $ with $d_{ws} =  ( 1(\Zws=00),1(\Zws=01),1(\Zws=10),1(\Zws=11) )^\T$. 
Let $\tbeta_\of$ and $\tbeta_\wf$ denote  the coefficient vectors of $d_\ws$ from the \olss and \wlss fits of \eqref{dlm4_f}, respectively; we use the subscript ``\fisher'' to signify \citet{Fisher35}. 

Let $\hbv_w(z) = \alpha_w \hbx_w(z)$ be the covariate analog of $\hU_w(z) = \alpha_w\hY_w(z)$ with  $\hbx_w(z) = M_{wb}^{-1}\sumsz  x_{ws}$. 
Adding $\hvwab$ to \eqref{lm_t_ag} defines 
\begin{eqnarray}
\hU_w(A_wb) 
\sim
\sum_{z\in\mt} 1(A_wb=z)+  \hvwab \sim d_w(A_wb) + \hvwab
\label{dlm6_f}
\end{eqnarray}
as the {\it additive aggregate regression} over $\wbs$ with $  d_w(z) =   (1(z=00), 1(z=01), 1(z=10), 1(z=11) )^\T$ for $z = (A_w0), (A_w1)$.  
%
%
Let $\tbeta_\tf$ denote  the coefficient vector of $\dwab $ from the \olss fit of \eqref{dlm6_f}. 
This, together with the above $\tbeta_{\ols,\fisher}$ and $\tbeta_{\wls,\fisher}$, defines three covariate-adjusted regression estimators of $\bbY$.

\begin{remark}\label{rmk:covariate}
Oftentimes we may want to include some whole-plot level attributes to both the unit and aggregate regressions; examples include the weights of the mice in Example \ref{ex:meon} and the populations of the subdistricts in Example \ref{ex:jpe}. 
The definition of $x_{ws}$ is flexible enough to accommodate both unit and whole-plot level covariates.
In particular, given $c_w$ as a whole-plot level attribute that is prognostic to the unit outcome, $Y_{ws}$, we can simply let $x_{ws[j]} = c_{w}$ for $s = \ot{M_w}$ to include it as the $j$th covariate in the unit regression. 
The definition of $\hvwab$ then ensures that it enters into the aggregate regression as $\alpha_w c_w$ for all $w$ and $b$. 
\end{remark}

We now derive the design-based properties of the above three covariate-adjusted regression estimators, $\tilde\beta_{\dagger, \fisher} \ (\dagger = \ols, \wls,\ag)$. 
To simplify the presentation, 
we center the covariates to have $\bbx = N^{-1}\sumws x_{ws} = 0_J$. 
Let $\hx_*(z) \ (* = \nm, \HT, \haj)$ be the sample-mean, {\htf}, and Hajek estimators of $\bar x$ based on units under treatment $z$.
Let $\hat x_* =  ( \hbx_*(00), \hbx_*(01), \hbx_*(10), \hbx_*(11)  )^\T$ be the $4\times J$  matrix 
of $\{\hbx_*(z)\}_{z\in\mt}$ for $* = \nm, \HT, \haj$. 
Let  $\tg_\ols$, $\tg_\wls$, and $\tg_\ag$ be the coefficient  vectors  of $x_{ws}$ or  $\hvwab$ from the respective least-squares fits of \eqref{dlm4_f} and \eqref{dlm6_f}. 
Proposition \ref{prop:ols_x_fisher} below parallels Proposition \ref{prop:ls}, and states the numeric correspondence between $\tbeta_{\dagger,\fisher} \ (\dagger = \ols, \wls,\ag)$ and $\hat Y_* \ (* = \sm, \HT, \haj)$.

\begin{proposition}\label{prop:ols_x_fisher}
$
\tbeta_\of
=\bynm-  \hbx_\snm\tgo$, 
$\tbeta_\wf 
= \byhaj -\hbx_\shaj\tgw$, 
and
$
\tbeta_\tf 
= \byht - \hbx_\sht\tgu$.
\end{proposition}

Recall that $\bynm = \tbeta_\ols $, $\byhaj=\tbeta_\wls $, and $\byht = \tbeta_\ag  $ from Proposition \ref{prop:ls}.
Proposition \ref{prop:ols_x_fisher} links the covariate-adjusted $\tbeta_\df$'s back to the unadjusted $\tbeta_\dg$'s, 
and establishes $\tbeta_\of$, $\tbeta_\wf$, and $\tbeta_\tf$ as the sample-mean, Hajek, and {\htf} estimators based on the covariate-adjusted outcomes  $Y_{ws} - x_{ws}^\T\tg_\dg $ for $\dgss$, respectively. 
The correspondence between the {\fts} $\{ \ols, \wls, \ag \}$  and the  design-based estimation schemes  $\{\nm, \haj, \HT\}$  is preserved in these covariate-adjusted fits as well.

The unbiasedness of the {\htf} estimator is in general lost after covariate adjustment due to the correlation between $\hbx_\sht $ and $\tgu$ in $\tbeta_\tf  = \byht - \hbx_\sht\tgu$. 
The consistency of $\hY_*$ and $\hbx_*(z)$, where $* = \HT, \haj$, for estimating $\bY$ and $\bar x$ under mild conditions, on the other hand, ensures the consistency of $\tbeta_\wf $ and $\tbeta_\tf $ so long as $\tg_\wls$ and $\tg_\ag $ have finite probability limits. We formalize the intuition in Theorem \ref{thm_fisher} below. 
The sample-mean analog $\tbeta_\of$, on the other hand, is in general neither unbiased nor  consistent  unless the design is {\sym}.  
We thus deprioritize it in the following discussion  but outline its theoretical guarantees under {\sym} designs in Section S4.1 of the Supplementary Material.

Let $\bar x_w = M_w^{-1}\sum_{s=1}^{M_w} x_{ws}$ and $\oxfw = M^{-1}_w \sums \| x_{ws} \|_4^4$ be the whole-plot average  and uncentered fourth norm of $(x_{ws})_{s=1}^{M_w}$, respectively.
Let 
$\sx$, $\swx$, $\sxyz$, and $\sxyzw$ be the scaled between and within whole-plot  covariances of $(x_{ws})_{\ws\in\ms}$ and $\{x_{ws}, Y_{ws}(z)\}_{\ws\in\ms}$, respectively. 
Let $\sxypz $ be the scaled between whole-plot  covariance of $(x_{ws})_{\ws\in\ms}$ with the centered potential outcomes $\{Y_{ws}'(z)\}_{\ws\in\ms}$; the corresponding scaled within whole-plot covariance coincides with $\sxyzw$. 
To avoid too many formulas in the main paper, we relegate the explicit forms of $\sx$, $\swx$, $\sxyz$, $\sxyzw$, and $\sxypz$ to Section S3.1 of the  Supplementary Material.

Recall $\Psi(z,z)=W^{-1} \sumw M_w^{-1} H_w(z,z)\sw(z,z)$ as the $(z,z)$th element of $\Psi$ from Lemma \ref{Vmat}. Let
$\Psi_{xx}(z,z)$  and $\Psi_{xY}(z,z)$ be the analogs after replacing $\sw(z,z)$ with $\sxxw$ and $\sxyzw$, respectively. 
Let $
\Qxx = (N-1)^{-1}\sumws x_{ws}x_{ws}^\T$ and $ \qxyz  = (N-1)^{-1}\sumws x_{ws}\yws(z)$ be the 
finite-population covariances of $(x_{ws})_{\ws\in\ms}$ and $\{x_{ws}, Y_{ws}(z)\}_{\ws\in\ms}$, respectively.

\begin{condition}\label{asym2}
As $W$ goes to infinity,  
\begin{enumerate}[(i)]
\item\label{po_xy}  $\bsx$, $\Psi_{xx}(z,z) $,  $\sxyz  $, $\sxypz $, $\Psi_{xY}(z,z)$, $Q_{xx}$, and $\qxyz$ have finite limits for all $z \in \mt$;
\item\label{po_x} $\max_{w=\ot{W}} \|\alpha_w\bbx_w\|^2_2/W=o(1)$;
\item  \label{po_x4}  
$W^{-1}\sumw \alpha^2_w   \oxfw  = O(1); \ W^{-2}\sumw \alpha^4_w    \oxfw  = o(1)$.
\end{enumerate}
\end{condition}

Condition \ref{asym2} is the analog of Condition \ref{asym}\eqref{po_1}--\eqref{po_3} for the covariates as potential outcomes unaffected by the treatment. 
The additional requirement on $\Qxx$ and $\qxyz$ ensures the convergence of $\tg_{\wls}$ under split-plot randomization.

Let  $\gamma_\dg $ be the finite probability limit of $\tg_\dg $ for $\dgs  $ under split-plot randomization and Conditions \ref{asym}--\ref{asym2}; we relegate the proof of their existence and explicit forms to Lemma S12 in the Supplementary Material.
Let $S_\wf$ and $\Sigma_\wf$ be the analogs of $\sshaj$  and  $\Sigma_\shaj$  defined on the adjusted potential outcomes $Y_{ws}(z; \gamma_\wls) = \Yws(z) - x_{ws}^\T \gamma_\wls$. 
Let  $S_\tf$  and $\Sigma_\tf$ be the analogs of $\ssht$  and $\Sigma_\sht$  defined on the adjusted potential outcomes $Y_{ws}(z; \gamma_\ag ) = \Yws(z) - x_{ws}^\T \gamma_\ag $. 
Let $\tbV_\df $ be the cluster-robust covariance of $\tbeta_\df $ for $\dgs $.

\begin{theorem}\label{thm_fisher}
{\prex}
\begina 
\sqrt{W} \big(\tbeta_\df   -\bbY  \big) \rightsquigarrow \mathcal{N}(0, \Sigma_\df  ), \qquad W  \tbV_\df  -\Sigma_\df   =S_\df + \op
\enda
with $S_\df  \geq 0$ for $\dgs$.
\end{theorem}

Theorem \ref{thm_fisher} states the asymptotic Normality of $\tbeta_\df  \ (\dagger = \wls, \ag)$  under split-plot randomization, and ensures the asymptotic conservativeness of $\tbV_\df $ for estimating the true sampling covariance. 
This justifies the regression-based inference of $\tau = G\bY$ from the additive regressions with point estimator $G \tbeta_\df $ and estimated covariance $G \tilde V_\df G^\T$ for $\dgs $.
Neither $\tbeta_\wf$ nor $\tbeta_\tf$, however, always guarantees  efficiency gains over their unadjusted counterparts, namely $\tbeta_\dg$ for $\dgs$,  even asymptotically.  
Similar discussion by \citet{Freedman08a} and \citet{Lin13} under complete randomization suggests including the interactions between the treatment indicators and covariates in the regression could be a possible remedy. 
We quantify its design-based properties 
in the next two subsections.

%
\subsection{Fully-interacted  regressions}
%
Modify \eqref{dlm4_f} and \eqref{dlm6_f} with full interactions between the treatment indicators and covariates: 
\beginy
Y_{ws}  &\sim&  d_{ws} + \sumz 1(\Zws=z)\, x_{ws} \label{dlm4}\\
  \huwab & \sim& d_w(A_wb) + \sum_{z \in\mt} 1(A_wb=z )\, \hvwab. \label{dlm6}
\endy
This defines the {\it fully-interacted unit} and {\it  aggregate regressions}, respectively.
Let 
$\tbeta_\ol$ and $\tbeta_\wl$ denote the coefficient vectors of $d_{ws}$ from the \olss and \wlss fits of the unit regression \eqref{dlm4}, respectively. 
Let  $\tbeta_\tl$ denote the coefficient of $\dwab $ from the \olss fit of the aggregate regression \eqref{dlm6}.
We use the subscript ``\lin'' to signify \citet{Lin13}. 
This gives us three more estimators of $\bbY$, $\{\tilde\beta_{\dagger, \lin}: \dg = \ols, \wls, \ag\}$,  summarized in Table \ref{tb2}. 
We now derive their respective design-based properties. 

\begin{table}
\caption{\label{tb2}Nine {\reg} estimators from the unadjusted specifications \eqref{lm_t}--\eqref{lm_t_ag}, the additive specifications \eqref{dlm4_f}--\eqref{dlm6_f}, and the fully-interacted specifications \eqref{dlm4}--\eqref{dlm6}, respectively, under the ``\textup{ols}'', ``\textup{wls}'', and  ``\textup{ag}'' \fts.}
\begin{center}
\renewcommand{\arraystretch}{1.2}
\begin{tabular}{|c|c|c|ccc| }\hline
fitting scheme &base model & weight & unadjusted & additive & fully-interacted  \\\hline
 $\ols$ &\multirow{2}{*}{$ Y_{ws} \sim d_{ws}$} & 1 &  $\bls$ & $\blsxf$  & $\blsx$  \\
 $\wls$ &  & $\{p_\ws(Z_\ws)\}^{-1}$ & $\blsHT$ & $\tbeta_\wf$ & $\tbeta_\wl$ \\\hline
$\ag $ &$\huwab \sim \dwab $ & 1&  $\blsu$ & $\blsxuf$  & $\blsxu$ \\\hline
\end{tabular}
\end{center}
\end{table}

Let $\tg_{\dg,z}$ be  the coefficient of $1(Z_{ws}=z) \, x_{ws}$ or $ 1(A_wb=z)\, \hvwab$ from the corresponding regression for $z \in \mt$ and $\dgssi $. 
Let $\tbeta_\dl (z)$ be the element in $\tbeta_\dl $ that corresponds to treatment $z$. 

\begin{proposition}\label{prop:ols_x} 
$
\tbeta_{\ols,\lin}(z) =\ynm- \hbx_\snm^\T (z) \tgz$, $\tbeta_{\wls,\lin}(z) =\yhaj - \hbx_\shaj^\T(z)\tgwz$,  and $  \tbeta_\tl (z)= \yHT - \hbx_\sht^\T (z) \tgzu$ for $ z\in \mt$.
\end{proposition}

Proposition \ref{prop:ols_x} parallels Proposition \ref{prop:ols_x_fisher}  under the additive regressions. It links the covariate-adjusted $\tbeta_\dl $'s back to the unadjusted $\tbeta_\dg $'s from \eqref{lm_t} and \eqref{lm_t_ag}, respectively, 
and establishes $\tbeta_\ol(z)$, $\tbeta_\wl(z)$, and $\tbeta_\tl(z)$ as the sample-mean, Hajek, and {\htf} estimators based on the covariate-adjusted outcomes $Y_{ws} - x_{ws}^\T\tg_{\dg,z}$ for $\dgss$, respectively. 
A key distinction is that the adjustment is now based on treatment-specific coefficients. 

The correlation between  $\hbx_\sht(z)$ and $\tguz$ likewise leaves the covariate-adjusted {\htf} estimator, $\tbeta_\tl$,   biased in finite samples. 
The fact that $\tgzs$ has a finite probability limit, denoted by $\gamma_{\dg,z}$,  under Conditions \ref{asym}--\ref{asym2}, on the other hand, ensures the consistency of $\tbeta_{\dg,\lin}$  for $\dgs$;  see Lemma S11 in the Supplementary Material. 
Section S4.1 of the Supplementary Material further gives analogous results about $\tbeta_{\ols,\lin}$ under {\sym} designs.

Let $S_\wl$ and $\Sigma_\wl$ be the analogs of $\sshaj$ and $\Sigma_\shaj$  defined on the adjusted potential outcomes $Y_{ws}(z; \gamma_{\wls,z}) = \Yws(z) - x_{ws}^\T \bbzp$. 
Let $S_\tl$ and $\Sigma_\tl$ be the analogs of $ \ssht$ and $\Sigma_\sht$  defined on the adjusted potential outcomes $Y_{ws}(z; \bbzu) = \Yws(z) - x_{ws}^\T \bbzu$. 
Let $\tbV_\dl $ be the cluster-robust covariance of $\tbeta_\dl $ for $\dgs$.

\begin{theorem}\label{thm_lin}
{\prex}
\begina
 \sqrt{W} \big(\tbeta_\dl   -\bbY  \big) \rightsquigarrow \mathcal{N}(0, \Sigma_\dl  ), \qquad W \tbV_\dl -\Sigma_\dl =S_\dl + \op
 \enda
with $S_\dl  \geq 0$  for $\dgs  $.
\end{theorem}

 Echoing the comment after Theorem \ref{thm_fisher}, Theorem \ref{thm_lin} justifies the regression-based inference of $\tau = G\bY$ from the fully-interacted regressions with point estimator $G \tbeta_\dl $ and estimated covariance $G \tilde V_\dl G^\T$ for $\dgs $.

\subsection{Guaranteed gains in asymptotic efficiency}\label{sec:eff}

A natural next question is if the inclusion of the interactions is not just as good but delivers extra gains
in  asymptotic efficiency.
The answer is affirmative when the right covariates are 
used in combination with the aggregate specification. 

Inspired by the utility of cluster-size adjustment in improving asymptotic efficiency under one-stage cluster randomization \citep{MiddletonCl15, su},  one simple extension to the aggregate regressions in \eqref{dlm6_f} and \eqref{dlm6} is to 
also include the centered whole-plot size factor, $(\alpha_w-1)$, as an additional whole-plot level covariate in addition to $\hvwab$; see Remark \ref{rmk:covariate}. 
Intuitively, $\alpha_w$ reflects the size of the whole-plot and is thus prognostic to $\hat U_w(z)$ as the outcome of the aggregate regressions. 
Let $\tbeta_{\ag, \fisher }(\alpha,v)$ and $\tbeta_{\ag, \lin }(\alpha,v)$ be the resulting \olss  coefficient vectors of $\dwab $  under the additive and fully-interacted specifications, respectively; we use the suffix ``$(\alpha,v)$'' to emphasize the components of the corresponding augmented covariates.
Proposition \ref{prop:eff} below states the asymptotic efficiency of $\tbeta_{\tl}(\alpha,v)$.

\begin{proposition}\label{prop:eff}
Under the $2^2$ split-plot randomization and Conditions \ref{asym}--\ref{asym2}, if
\beginy\label{cond:eff}
\Psi_{xx}(z,z) = o(1) \qquad \text{for all} \ \ z\in \mt,
\endy
then 
$\tbeta_{\tl}(\alpha,v)$ has the smallest asymptotic sampling covariance among 
$$
\mathcal{B} =  \{\tbeta_\wls, \, \tbeta_{\wls,\dmd }; \ \tbeta_\ag , \, \tbeta_{\ag ,\dmd },\,  \tbeta_{\ag ,\dmd }(\alpha,v):   \dmd  = \fisher, \lin  \}.
$$
\end{proposition}

Proposition \ref{prop:eff} establishes the optimality of $\tbeta_{\tl}(\alpha,v)$ among the eight consistent {\reg} estimators   in $\mathcal{B}$, highlighting the utility of  including $(\alpha_w-1)$ as an additional covariate in the aggregate regression for ensuring additional asymptotic efficiency.
The asymptotic efficiency over the  unadjusted $\hys \ (* = \HT, \haj)$ 
then follows from the numeric identities in Proposition \ref{prop:ls}. 
Condition \eqref{cond:eff} holds if (i) $x_{ws} = \bar x_w$ or (ii) $S_{xx,w}$ is uniformly bounded while $M_w$ goes to infinity for all $w$. 
 We thus recommend using $\tbeta_{\tl}(\alpha,v)$ when the covariates are relatively homogeneous within whole-plots or when  it is reasonable to consider whole-plot level covariates only. 
The latter ensures gains in asymptotic efficiency  over the unadjusted case even if the unit-level covariates show great heterogeneity within each whole-plot.  
The discussion becomes more complicated when heterogeneous unit-level covariates enter the regression equations. 
We leave the more general theory to future work. 

This guaranteed minimum asymptotic covariance, however, should not be the basis for dismissing the additive regressions completely. 
In particular, the fully-interacted regression \eqref{dlm6} involves 
$|\mT| \times (1+J)$ parameters compared with the $(|\mT| + J)$ parameters in the additive regression \eqref{dlm6_f}, subjecting $\tbeta_{\tl}$ to possibly substantial   finite-sample variability when $J$ is large. 
We thus recommend keeping both strategies in the toolkit and making decisions on a case by case basis contingent on the nature of the design and the abundance of data.

\section{Extensions}

\subsection{Factor-based regressions: practical implementations}

The regressions so far take indicators of the treatments, namely $1(Z_{ws} = ab)$ or $1(A_wb=ab)$, as regressors.
Despite the generality of such formulations and their theoretical guarantees,
they are nevertheless not the dominant choice in practice 
when the goal is to estimate the standard factorial effects as those defined in \eqref{taud}.  {\it Factor-based}  regressions, as the more popular practice, estimate the factorial effects directly by the regression coefficients.  

With the treatment combinations of interest exhibiting a   $2^2$ 
factorial structure, the factor-based approach regresses 
the outcome on the factors themselves via specifications like 
\beginy
\label{lm_f}
\Yws \sim 1 + A_{w} + B_{ws} + A_{w}B_{ws}   \qquad (\ws \in\ms) ,
\endy
and interprets the coefficients of the non-intercept terms as the main effects and interaction, respectively. 
Let $\tbt'_{\A,0}$, $\tbt'_{\B,0}$, and $\tbt'_{\AB,0}$ be the \olss coefficients of $A_{w}$, $B_{ws}$, and $A_{ws}B_{ws}$ from \eqref{lm_f}, respectively. 
Standard least-squares theory ensures
\begin{align*} 
&&\tbt'_{\A,0} &= \hY_\sm(10) - \hY_\sm(00),\nonumber\\
&&\tbt'_{\B,0} &= \hY_\sm(01) - \hY_\sm(00),\\
&&\tbt'_{\AB,0} &=\hY_\sm(11) - \hY_\sm(10) - \hY_\sm(01) + \hY_\sm(00), 
\end{align*}
equaling the sample-mean estimators of $\tau_{\A,0} = \by(10) -\by(00)$, $\tau_{\B,0} = \by(01)-\by(00)$, and $ \tau_{\AB}$ from \eqref{taud}, respectively.  When the goal is to estimate the standard factorial effects $(\tau_\A, \tau_\B, \tau_\AB)$ as defined in \eqref{taud},
an algebraic trick is to  shift  the factor indicators by $1/2$ 
and form the regression as 
\beginy\label{dlm1_a}
Y_{ws} \sim 1 + (A_{w}-1/2) + (B_{ws} -1/2) + (A_{w}-1/2)(B_{ws}-1/2) 
\endy
over $ws \in \ms$. 
Fitting \eqref{dlm1_a} by \olss recovers the sample-mean estimators of $\tau_\A$, $\tau_\B$, and $\tau_\AB$, respectively \citep{ZDfact}.
Denote by $\tbt'_\ols$ and $\tbO'_\ols$ the resulting coefficient vector and cluster-robust covariance of the three non-intercept terms, respectively. 
As a convention, we use the combination of tilde ($\tilde{\color{white}\tau}$) and prime  ($\,'\,$) to signify outputs from factor-based regressions like \eqref{lm_f} and \eqref{dlm1_a}.

Following the intuition from the \olss fit,  
let $\tbt'_\wls $ be the coefficient vector of the non-intercept terms from \eqref{dlm1_a} under the ``wls'' fitting scheme, with $\tbO'_\wls$  
as the associated cluster-robust covariance.
Let 
\begin{eqnarray}
\label{dlm3_a} 
\hU_w(A_wb) \sim 1 + (A_w-1/2) + (b -1/2) + (A_w-1/2)(b -1/2)
\end{eqnarray} 
be the aggregate analog of \eqref{dlm1_a} over $\wbs$, 
with  $\tbt'_\ag$ and $\tbO'_\ag$ as the resulting \olss coefficient vector  
and cluster-robust covariance of the non-intercept terms under fitting scheme ``ag''.

Recall  $G_0$ as the contrast matrix corresponding to $(\tau_\A, \tau_\B, \tau_\AB)^\T = G_0\by$. 
Proposition \ref{correspondence-nox} below follows from the invariance of least squares to non-degenerate linear transformation of the regressors, and  ensures the validity of $(\tbt'_\dg, \tbO'_\dg)$ for the Wald-type inference of the standard factorial effects for $\dg = \wls, \ag$.

\begin{proposition}\label{correspondence-nox}
$\tbt'_\dg  = G_0\tbeta_\dg $ and $ \tbO'_\dg  =
 G_0  \tbV_\dg    G_0^\T$ for $\dg  =  \ols, \wls, \ag$. 
\end{proposition}

Specifications \eqref{dlm1_a}--\eqref{dlm3_a} thus deliver the Hajek and {\htf} estimators of the standard factorial effects, namely $\htau_\shaj =  G_0\byhaj$ and $\htau_\sht =  G_0\hyht$,  directly as regression coefficients. 
We thus recommend using \eqref{dlm1_a}--\eqref{dlm3_a} if the goal is the standard factorial effects and switching back to \eqref{lm_t}--\eqref{lm_t_ag} if otherwise.

The results on covariate adjustment are almost identical to those based on the treatment indicators.
The covariate-adjusted estimator from the \fr ensures asymptotic efficiency under   Conditions \ref{asym}--\ref{asym2} and \eqref{cond:eff}, and is thus our recommendation for estimating the standard factorial effects under the $2^2$ split-plot design. 
We relegate the details to Section 4.3 of the Supplementary Material.

\subsection{ $\ta\times \tb$ split-plot design}\label{sec:ext_tatb}
All discussion so far concerns the $2^2$ split-plot design with  the whole-plot factor and the sub-plot factor both of two levels.
We now extend the result to general split-plot designs with  factors of multiple levels.

 Consider two factors of interest, A and B, of $\ta \geq 2$ and $\tb \geq 2$ levels, respectively, indexed by $a \in \mta = \{ \zt{\ta-1}\}$ and $b \in \mtb =\{ \zt{\tb-1}\}$. 
A general $\ta\times \tb$ split-plot randomization first runs a cluster randomization at the whole-plot level and assigns completely at random $W_a$ of the $W$ whole-plots to receive level $a\in \mT_{\textsc{a}}$ of factor A 
with $\sum_{a\in \mT_\A} W_a = W$. 
It then conducts an independent randomization within each whole-plot and assigns completely at random $\Mwb$ sub-plots in whole-plot $w$ to receive level $b \in \mT_\B $ of factor B 
with $\sum_{b\in \mT_\B }\Mwb = M_w$ for $w = \ot{W}$. 
Refer to the cluster and stratified randomizations as stage (I) and stage (II) of the assignment, respectively. 
The final treatment of a sub-plot is the combination of its whole-plot factor assignment in stage (I) and its sub-plot factor assignment in stage  (II), taking values from $\mT = \{(ab): a \in \mT_{\textsc{a}}, \ b \in \mT_\B \}$. Example \ref{ex:meon} defines a $3 \times 2$ split-plot experiment, and Example \ref{ex:jpe} defines a $2 \times 3$ split-plot experiment.

All notation and results from the $2^2$ case extend to the current setting with minimal modification.
We relegate the details on the design-based inference to Section S4.4 of the Supplementary Material and focus below on the model-based inference from factor-based regressions. 

Renew $\{Y_{ws}(z):  z \in \mt\}_{\ws \in \ms}$ as the potential outcomes under the $T_\A \times T_\B$ design with finite-population averages $\{\bY(z)\}_{z \in \mt}$, vectorized as $\bar Y$. Assume 0 as the baseline levels for both $\mt_\A$ and $\mt_\B$.
We define
\beginy\label{eq:tau_tatb}
&&\tau_{ \A   a} = \tbinv  \sum_{b\in\mt_\B} \{ \bar Y(ab) -\bar Y(0b)\}, \quad\quad
 \tau_{ \B   b} =  \tainv  \sum_{a\in\mt_\A} \{ \bar Y(ab) -\bar Y(a0)\},\\
&& \tau_{ \A  a, \B  b} =  \bar Y(ab)  -\bar Y(0b) - \bar Y(a0) + \bar Y(00) \nonumber 
\endy
as the standard main effects and interactions at non-baseline levels $a = \ot{\ta-1}$  and $b=\ot{\tb-1}$, respectively, 
vectorized as
\begina
\tau = \{\tau_{\A a}, \tau_{\B b}, \tau_{\A a, \B b}: a = \ot{\ta}; \ b = \ot{\tb}\} = G_0\by. 
\enda 
The definitions reduce to $\tau_\A, \tau_\B$, and $\tau_\AB$ from \eqref{taud} when $T_\A  = T_\B=2$.

Inspired by the utility of location-shifted factors in recovering $\tau_\A$ and $\tau_\B$   directly as least-squares coefficients from \eqref{dlm1_a} and \eqref{dlm3_a}, 
consider 
 \beginy\label{os_g}
&&\quad Y_{ws}  \sim  1 + \sum_{a=1}^{\ta-1} { \oc } (A_w=a) + \sum_{b=1}^{\tb-1}{ \oc } (B_{ws}=b) + \sum_{a=1}^{\ta-1}\sum_{b=1}^{\tb-1}{\oc } (A_w = a){\oc } (B_{ws}=b),\\
&&\quad\hU_w(A_wb)  \sim 1 + \sum_{a=1}^{\ta-1} {\oc } (A_w=a) + \sum_{b'=1}^{\tb-1}{\oc } (b=b') + \sum_{a=1}^{\ta-1}\sum_{b'=1}^{\tb-1}{\oc } (A_w = a){\oc } (b=b')  \label{os_g_a}
\endy
 as two generalizations under the $\ta\times \tb$ design with 
 ${\oc } (A_w=a) = 1(A_w=a) - \tainv$, 
 ${\oc } (B_{ws}=b) = 1(B_{ws}=b) - \tbinv$, and  
 ${\oc } (b=b')= 1(b=b') -\tbinv$.

Renew $\tbt'_\dg \ (\dgss)$ as the coefficient vectors of the non-intercept terms from \eqref{os_g} and \eqref{os_g_a} under fitting schemes ``ols'', ``wls'', and ``ag'', respectively. 
Renew $\hY_*$ as the sample-mean, {\htf}, and Hajek estimators of $\bY$ for $* = \nm, \HT, \haj$. 
Proposition \ref{correspondence_g} below states their numeric correspondence paralleling Proposition \ref{correspondence-nox}.

\begin{proposition}\label{correspondence_g}
$\tbt'_\ols = G_0\bynm,$ $ \tbt'_\wls = G_0\byhaj,$ and $ \tbt'_\ag  = G_0\byht.$
\end{proposition}

The unbiasedness and consistency results then follow from the properties of $\hY_*$ for $* = \nm, \HT, \haj$ such that $\tbt'_\wls $ and $\tbt'_\ag $ are both consistent for estimating $\tau$.
The results on the cluster-robust covariances parallel those under the $2^2$ design, and ensure asymptotically conservative estimation of the true sampling covariances. 
This justifies the validity of  regression-based inferences from \eqref{os_g}--\eqref{os_g_a} under {\fts} $\dgs $.

The results on covariate adjustment are almost identical to those under the $2^2$ case and thus omitted.
The covariate-adjusted estimator from the \fr ensures asymptotic efficiency under   the generalized version of  Conditions \ref{asym}--\ref{asym2} and \eqref{cond:eff}. It is thus our recommendation for estimating the standard factorial effects under the general $T_\A\times T_\B$ split-plot design.

\subsection{Fisher randomization test}
The Fisher randomization test targets the strong null hypothesis of no treatment effect on any unit in its original form, 
and delivers finite-sample exact $p$-values regardless of the choice of test statistic \citep{CausalImbens}. 
The theory under complete randomization further demonstrates that the Fisher randomization test with  a robustly-studentized test statistic
is finite-sample exact under the strong null hypothesis, asymptotically valid under the weak null hypothesis of 
zero average treatment effect, 
and allows for flexible covariate adjustment to secure additional power \citep{D20, ZD20FRT}. 
The theory extends naturally to split-plot randomization.

Renew $\mathcal{B} =  \{\tbeta_\wls, \tbeta_{\wls,\dmd}; \, \tbeta_\ag ,  \tbeta_{\ag ,\dmd}, \tbeta_{\ag ,\dmd}(\alpha,v):  \dmd  = \fisher, \lin  \}$ as the collection of {\reg} estimators of $\bY$ that are consistent  under the general $\ta\times\tb$ split-plot design.
We propose to use the Fisher randomization test with test statistic $
t^2(\tbeta) = 
(G \tbeta)^\T   (G\tbV  G^\T)^{-1}  G\tbeta
$
for $\tbeta\in \mathcal{B}$, where $\tbV$ is the associated cluster-robust covariance.

The resulting test is finite-sample exact for testing the strong null hypothesis and asymptotically valid for testing the weak null hypothesis under split-plot randomization  for all $\tbeta \in\mathcal{B}$. Under \eqref{cond:eff}, the test based on $t^2(\tbeta_{\ag,\lin} (\alpha,v))$ has the highest power asymptotically. By duality, we can also construct confidence regions for factorial effects by inverting a sequence of Fisher randomization tests. We relegate the details to Section S4.5 of the Supplementary Material.

\section{Numerical example}\label{sec:application}
We now apply the proposed methods to  
the $3\times 2$ neuroscience experiment from Example \ref{ex:meon}.  
Due to the space limit, we relegate the simulation studies to Section S5 of the Supplementary Material.

We use the data set from \cite{meon}, which consists of $N= 1,143$ neuron level observations nested within $W = 14$ mice. Recall fatty acid delivery (``fa'') and Pten knockdown (``pten'') as the whole-plot and sub-plot factors, respectively, with $(T_\A, T_\B) = (3, 2)$.
The treatment sizes at the whole-plot level are $(W_0, W_1, W_2) = (5, 4, 5)$.
The whole-plot sizes $(M_w)_{w=1}^W$ vary from $13$ to $152$, with the  $q_{w1}$'s ranging from $0.518$ to $0.846$.

Denote by ``fa1'' and ``fa2'' the standard main effects of fatty acid delivery, by ``pten'' the standard main effect of Pten knockdown, and by ``fa1:pten'' and ``fa2:pten'' their interactions, with definitions given by \eqref{eq:tau_tatb}. 
We apply  four  factor-based regression schemes as the combinations of two fitting schemes, namely ``wls'' for  the   unit specification and ``ag'' for   the aggregate specification, and  the presence or absence of covariate adjustment. 
We use $x_{ws} = \alpha_w =  M_w /\bar M$ as the covariate for neurons in mouse $w$, and conduct covariate adjustment using the additive specification due to the small number of whole-plots at $W=14$.

Table \ref{tb:ptn} shows the point estimators, cluster-robust standard errors, $p$-values based on large-sample approximations of the $t$-statistics, and $p$-values based on Fisher randomization tests, respectively. 
Covariate adjustment reduces the standard errors of the estimators of ``fa1'' and ``fa2'' under both the unit and aggregate specifications,  yet has no effect on the estimators of ``pten'' and the  two interactions.  
The identicalness of the unadjusted and additive regressions for estimating ``pten'' and the two interactions is no coincidence but due to the use of whole-plot level covariate, namely $x_{ws} = \alpha_w$, for covariate adjustment. 
 The resulting covariate matrix, as the concatenation of $(x_{ws})_{ws\in\ms}$, is orthogonal to the centered regressors for the sub-plot factor and interactions after accounting for the least-squares weights, leaving their estimation unaffected by the inclusion of covariates. 
Proposition S3 in the Supplementary Material offers a rigorous statement.

The $p$-values from large-sample approximations and Fisher randomization tests concur in most cases at significance level $0.05$. 
Two exceptions are the tests for ``fa2'' under the adjusted unit regression and those for  ``fa1'' under the unadjusted aggregate regression. 
This is likely due to the small number of whole-plots that leaves the asymptotic approximation dubious. 
Based on the theory, the $p$-values based on the Fisher randomization tests should be trusted more given their additional guarantee of finite-sample exactness under the strong null hypothesis. 
This is especially important in the current case given its  small $W$. 

Recall that the {\reg} estimators under the ``wls'' and ``ag'' {\fts} correspond to the Hajek and {\htf} estimators, respectively. 
The estimators and $p$-values under the two fitting schemes concur in most cases with the two exceptions being the $p$-values from large-sample approximations for ``fa1'' under the unadjusted regressions and those for ``fa2'' under the adjusted regressions. 
This is again likely due to the  small $W$. 

Overall the four regression schemes and  two types of $p$-values lead to coherent conclusions: the Pten knockdown increased the soma sizes whereas the effect of  fatty acid delivery, along with its interaction with Pten, is statistically insignificant.

  \begin{table}[ht]
\centering
\caption{\label{tb:ptn}Re-analyzing the data from \citet{meon}. 
``p.normal'' and ``p.frt'' indicate the $p$-values from large-sample approximations and Fisher randomization tests, respectively.}

(a) regression based on unit data\\
\begin{tabular}{|r|rrrr|rrrr|}
  \hline
  &\multicolumn{4}{|c|}{unadjusted} &\multicolumn{4}{|c|}{adjusted}\\\hline
  &  est  &  se  &  p.normal  & p.frt &  est  &  se  &  p.normal  &  p.frt \\ 
  \hline
fa1  & $ 8.30 $ & $ 4.57 $ & $ 0.07 $ & $ 0.14 $ & $ 6.16 $ & $ 3.60 $ & $ 0.09 $ & $ 0.22$ \\ 
  fa2  & $ 5.29 $ & $ 5.31 $ & $ 0.32 $ & $ 0.39 $ & $ 10.36 $ & $ 4.86 $ & $ 0.03 $ & $ 0.11$ \\ 
  pten  & $ 14.00 $ & $ 1.81 $ & $ 0.00 $ & $ 0.00 $ & $ 14.00 $ & $ 1.81 $ & $ 0.00 $ & $ 0.00$ \\ 
  fa1:pten  & $ 8.64 $ & $ 5.24 $ & $ 0.10 $ & $ 0.23 $ & $ 8.64 $ & $ 5.24 $ & $ 0.10 $ & $ 0.23$ \\ 
  fa2:pten  &  $-1.33 $ & $ 2.71 $ & $ 0.62 $ & $ 0.68 $ & $ $-1.33$ $ & $ 2.71 $ & $ 0.62 $ & $ 0.68$ \\ 
   \hline
\end{tabular}\\
\bigskip
 (b) regression based on aggregate data\\
\begin{tabular}{|r|rrrr|rrrr|}
  \hline
     & \multicolumn{4}{|c|}{unadjusted}  & \multicolumn{4}{|c|}{adjusted}\\\hline
  &  est  &  se  &  p.normal  & p.frt &  est  &  se  &  p.normal  &  p.frt \\ 
  \hline
fa1  & $ 9.84 $ & $ 4.22 $ & $ 0.03 $ & $ 0.08 $ & $ 5.66 $ & $ 3.69 $ & $ 0.14 $ & $ 0.21$ \\ 
  fa2  & $ 5.64 $ & $ 6.23 $ & $ 0.38 $ & $ 0.42 $ & $ 8.24 $ & $ 5.53 $ & $ 0.15 $ & $ 0.23 $\\ 
  pten  & $ 13.15 $ & $ 1.95 $ & $ 0.00 $ & $ 0.00 $ & $ 13.15 $ & $ 1.95 $ & $ 0.00 $ & $ 0.00$ \\ 
  fa1:pten  & $ 7.13 $ & $ 5.55 $ & $ 0.21 $ & $ 0.28 $ & $ 7.13 $ & $ 5.55 $ & $ 0.21 $ & $ 0.28$ \\ 
  fa2:pten  &  $-2.51 $ & $ 3.05 $ & $ 0.42 $ & $ 0.49 $ &  $-2.51 $ & $ 3.05 $ & $ 0.42 $ & $ 0.49$ \\ 
   \hline
\end{tabular}

\end{table}

\section{Discussion}\label{sec:discussion}
 
Based on the asymptotic analysis, we recommend using the \olss outputs  from the \fr  if  the sample size permits, and switching to the additive regression if otherwise. The point estimator based on the \olss coefficient is consistent for estimating the finite-population average treatment effect, with the associated cluster-robust covariance being an asymptotically conservative estimator of the true sampling covariance.
The  resulting regression-based inference is valid from the design-based perspective regardless of how well the regression equation represents the true relationship between the outcome, treatments, and covariates.

\begin{acks}[Acknowledgments]
We thank the Associate Editor, two referees, Zhichao Jiang, and Rahul Mukerjee for most insightful comments.
\end{acks}

\begin{funding}
Peng Ding was partially funded by the U.S. National Science Foundation (grant \# 1945136). 
\end{funding}

\begin{supplement}
\stitle{Supplement to ``Reconciling design-based and model-based causal inferences for split-plot experiments''}
\sdescription{ We give the proofs of the results in the main paper, and provide additional results on the special case of uniform designs, extensions, and simulation studies. 
}
\end{supplement}

\bibliographystyle{imsart-nameyear} 
\bibliography{refs_SP}       

\newpage 
\setcounter{equation}{0}
\setcounter{section}{0}
\setcounter{figure}{0}
\setcounter{example}{0}
\setcounter{proposition}{0}
\setcounter{corollary}{0}
\setcounter{theorem}{0}
\setcounter{table}{0}
\setcounter{condition}{0}
\setcounter{lemma}{0}
\setcounter{remark}{0}

\renewcommand {\theproposition} {S\arabic{proposition}}
\renewcommand {\theexample} {S\arabic{example}}
\renewcommand {\thefigure} {S\arabic{figure}}
\renewcommand {\thetable} {S\arabic{table}}
\renewcommand {\theequation} {S\arabic{equation}}
\renewcommand {\thelemma} {S\arabic{lemma}}
\renewcommand {\thesection} {S\arabic{section}}
\renewcommand {\thetheorem} {S\arabic{theorem}}
\renewcommand {\thecorollary} {S\arabic{corollary}}
\renewcommand {\thecondition} {S\arabic{condition}}
\renewcommand {\thepage} {S\arabic{page}}

\setcounter{page}{1}

\begin{center}
\bf \Large 
Supplementary Material to ``Reconciling design-based and model-based causal inferences for split-plot experiments''  
\end{center}

Section \ref{sec:design_app} gives the proofs for the design-based inference under the $2^2$ split-plot design. 

Section \ref{sec:lm_app} gives the proofs for the regression-based inference without covariate adjustment. 

Section \ref{sec:ca_app} gives the proofs for the  regression-based covariate adjustment. 

Section \ref{sec:ext_app} gives the results and proofs for the special case of {\sym} designs and extensions to the HC2 correction for the cluster-robust covariances, covariate adjustment via factor-based regression, the general $\ta\times\tb$ design, and the Fisher randomization test.

Section \ref{sec:simu} gives the results of the simulation studies.

\section{Design-based inference for the $2^2$ split-plot design}\label{sec:design_app}
\subsection{Notation and useful facts}\label{sec:design_notation}
We review in this subsection the key notation and useful facts for verifying the results on design-based inference under the $2^2$ split-plot design.
The results are stated in terms of the general $\ta\times \tb$ design to facilitate generalization. 

Let 
$ a \in \mt_\A= \{0, 1, \ldots, T_\A -1\} $ and $ b  \in \mt_\B = \{0, 1, \ldots, T_\B -1\}$ indicate the levels of factors A and B in treatment combination $z \in\mT = \mt_\A \times \mt_\B$ throughout unless specified otherwise.  The $\mt_\A$ and $\mt_\B$ reduce to $\{0,1\}$ under the $2^2$ split-plot design. Assume lexicographical order of $z$ for all vectors and matrices  when applicable. 

Let $\bar M = N/W$ be the average whole-plot size, and let $\alpha_w = M_w/\bar M$ be the whole-plot size factor for $w = \ot{W}$. 
Let 
$Z_{ws}=(A_w,B_{ws})\in\mt$ indicate the treatment received by sub-plot $\ws$, 
with $
\pr(A_w=a) =  W_a/W= p_a $, $\pr(B_{ws}=b) =  M_{wb}/M_w = q_{wb}$, and $p_{ws}(z) = \pr(Z_{ws}=z)=p_a  q_{wb}$. 
For $z=(ab)$, let 
$\Wz   = \{w:A_w=a\}$ be the set of whole-plots that contain at least one observation under treatment $z$.

Let $\bY_w(z) = M_w^{-1} \sums \yws(z)$ and $\uw(z) = \bar M^{-1} \sums \yws(z)= \alpha_w \byw(z)$
be the whole-plot average potential outcome and the scaled whole-plot total potential outcome for whole-plot $w$ under treatment $z$, respectively. We have 
\begina
\bY(z)= W^{-1}\sumw \alpha_w \bar Y_w(z) = W^{-1}\sumw U_w(z).
\enda
Let $Y'_\ws(z) = Y_\ws(z) - \bar Y(z)$ be the centered potential outcome 
with $\by'_w(z)  = M_w^{-1} \sums \yws'(z) = \by_w(z) -\by(z)$ and $\bar Y'(z) = N^{-1}\sumws \yws'(z)=0$. 

Let $\ssht = S = (S(z,z'))_{z, z'\in\mt}$ and $\sshaj = (S_\shaj(z,z'))_{z, z'\in\mt}$ be the $|\mt|\times |\mt|$  scaled between whole-plot covariance matrices of $\{Y_\ws(z):z\in\mt\}_{\ws\in\ms}$ and  $\{Y'_\ws(z):z\in\mt\}_{\ws\in\ms}$, respectively, with 
\begina
 \ssht(z,z') =S(z,z') &=&  (W-1)^{-1}\sumw \{\alpha_w Y_w(z) - \bar Y(z)\}\{\alpha_w Y_w(z') - \bar Y(z')\},\\
 \sshaj(z,z') &=& 
(W-1)^{-1}\sumw   \alpha_w^2  \{ Y_w(z) -  \bar Y(z)\}\{   Y_w(z') -   \bar Y(z')\}.
\enda
A key observation is that $\ssht(z,z')  = S(z, z')$ equals the finite-population covariance of $\{U_w(z), U_w(z')\}_{w=1}^W$. 
A useful fact is 
\beginy\label{eq:ss_diff}
\qquad \lambda_W  \{ \sshaj(z, z') -  \ssht(z, z') \}
 =\bY(z)\bar Y(z')  (  \oat +1) -  \bY(z')  \ \oau - \bY(z)  \  \oaup  
\endy
for $z, z'\in\mt$, where $\lambda_W = 1-W^{-1}$, 
 $\oa{2} = W^{-1}\sumw \alpha^2_w$, and 
$\oau = W^{-1}\sumw \alpha_w U_w(z)$.

Let $
\hY_w(z) = M_{wb}^{-1} \sumsz \Yws$ and  $\hU_w(z) = \alpha_w \hY_w(z)$
be the sample analogs of $\bY_w(z)$ and $\uw(z)$, respectively, with $\hat Y_w(z) = \hU_w(z) = 0$ for $\znotin$. 
We can decompose $\yht$ as
\begina
\yht 
&=& W_a^{-1}\sum_{w \in \Wz} \hU_w(z) \\
&=& W_a^{-1}\sum_{w \in \Wz} U_w(z)  +  W_a^{-1}\sum_{w \in \Wz} \{\hU_w(z) - U_w(z)\}\\
&=& \mu(z) + \sumw \delta_w(z) \qquad \text{for} \ \  z =(ab) \in \mt,
\enda
where 
\begina
\pphi(z) = W_a^{-1}\sum_{w \in \Wz} U_w(z), \qquad \delta_w(z) = 1(A_w=a) \cdot W_a^{-1} \{\hat U_w(z) - U_w(z)\}.
\enda 
The randomness in $\pphi(z)$ comes solely from the stage (I) randomization. 
Let $\mu = \{\mu(z)\}_{z\in\mt}$ and $\delta_w = \{\delta_w(z)\}_{z\in\mt}$ be the $|\mt|\times 1$ vectorizations of $\mu(z)$ and $\delta_w(z)$, respectively. Then
\beginy\label{eq:decomp}
\byht  =   \pphi + \delta, \qquad \text{where} \quad \delta = \sumw \delta_w. 
\endy

Let 
$\mA = \sigma(A_1, \dots, A_W)$ be the $\sigma$-algebra generated by $(A_w)_{w=1}^W$.
The independence between the stage (I) and stage (II) randomizations ensures that $(\delta_w)_{w=1}^W$ are jointly independent conditioning on $\mathcal A$.  
As a result, we have $E( \delta_w \mid \mA)=E( \delta_w \mid A_w)= 0$, 
$E(\delta \mid \mA) = 0, \ \cov(\delta \mid \mA)= \sumw \cov(\delta_w \mid \ma)$ with $ \cov(\delta_w \mid \ma)= \cov(\delta_w \mid A_w)$, and thus 
\beginy\label{eq:sigA}
\renewcommand{\arraystretch}{1.5}
\begin{array}{l}
E(\delta_w) = 0; \qquad E(\delta) = 0; \\
\cov( \delta )  = \sumw \cov( \delta_w) \qquad \text{with} \ \ \cov( \delta_w) = E\{\cov( \delta_w \mid A_w)\}; \\ 
\cov( \pphi,  \delta ) = E\{\cov( \pphi,  \delta  \mid \mA)\} + \cov\{ E( \pphi \mid \mA), E( \delta  \mid \mA)\} =  0  . 
\end{array}
\endy

Further let $U_{ws}(z) = \alpha_w Y_{ws}(z)$ be the scaled potential outcome at the unit level. 
Then  $  
U_w(z)  = M_w^{-1} \sums \uws(z)$ and $\sw(z,z') = (M_w-1)^{-1} \sums \{\uws(z) - U_w(z)\}\{\uws(z') - U_w(z')\}$ equal the finite-population mean and covariance of $\{\uws(z): z\in \mt\}_{s=1}^{M_w}$ in whole-plot $w$, respectively, with $ \hU_w(z) = M_{wb}^{-1} \sumsz \uws(z)$ as the sample mean under treatment $z$. 
Let 
\begina
H_w = \big(H_w(z,z')\big)_{z,z'\in\mt} =   \diag(p_a^{-1})_{a\in \mt_\A} \otimes \{  \diag(q_{wb}^{-1})_{b\in\mt_\B} - \bbo{T_\B}\}
\enda with $H_w(z,z') = 1(a = a') p_a^{-1} \{q_{wb}^{-1}1(z=z') -1 \} $ for $z = (ab)$ and $z'=(a'b')$. 
Then
\beginy\label{eq:covU}
\cov\{\hU_w(z), \hU_w(z')\mid A_w=a \} 
&=& M_w^{-1} \{q_{wb}^{-1}1(z=z') -1 \} S_w(z,z')\\
&=& p_a  M_w^{-1}H_w(z,z')  S_w(z,z')\nonumber
\endy
for $z = (ab)$ and $z'=(ab')$ that share the same level of factor A. Let 
$\oufzw = M_w^{-1} \sums  U^4_{ws}(z) = \alpha_w^4 \oyfzw $ be the uncentered fourth moment of $U_{ws}(z)$ in whole-plot $w$.

Lastly, 
for two $L \times 1$ vectors $u = (u_1, \dots, u_L)^\T$ and $v = (v_1, \dots, v_L)^\T$, we have 
\beginy\label{cs}
\quad  (u^\T v)^4 \leq \|u\|_2^4 \|v\|_2^4; \qquad \bar u^4 \leq (\overline{u^2})^2 \leq \overline{u^4}, \quad \text{where} \ \ \overline{u^k} = L^{-1} \sum_{l=1}^L u_l^k, 
\endy
by the Cauchy--Schwarz inequality. 
Setting $L=2$ in $\bar u^ 4 \leq \overline{u^4}$ ensures $(u_1+u_2)^4 \leq 8 u_1^4 + 8u_2^4$.

\subsection{Established lemmas}
\begin{lemma}\citep[][Theorems 3 and 5]{DingCLT}\label{Ding17}
In a completely-randomized experiment with $N$ units and $Q$ treatment groups of sizes $N_q \ (q = \ot{Q})$, 
let $ Y_i(q)$ be the $L \times 1$ vector potential outcome of unit $i$ under treatment $q$, and let $ S_{qq'} = (N-1)^{-1} \sum_{i=1}^N \{ Y_i(q) - \bbY(q)\}\{ Y_i(q') - \bbY(q')\}^\T$ be the finite-population covariance for $1\leq q,q'\leq Q$. 
Let $\bt =  \sum_{q=1}^Q  G_q \bbY(q)$ be the finite-population average treatment effect of interest, and let $\hbt =  \sum_{q=1}^Q  G_q \hbY(q)$ be the corresponding moment estimator with $\hy(q) = N_q^{-1}\sum_{i: Z_i = q} Y_i$, where $Z_i$ and $Y_i$ are the treatment indicator and observed outcome for unit $i$, respectively. 
Then
$$\cov(\hbt ) = \sum_{q=1}^Q N_q^{-1}  G_q  S_{qq}  G_q^\T - N^{-1}  S_{\bt}^2,$$
where $S_{\tau }^2$ is the finite-population covariance of $\bt_i =  \sum_{q=1}^Q  G_q Y_i(q)$ for $i = \ot{N}$. 
Further assume that as $N\to\infty$, for all $1\leq q, q' \leq Q$, (i) $ S_{qq'}$ has a finite limit, (ii) $N_q/N$ has a limit in $(0,1)$, and (iii) $\max_{i = \ot{N}} | Y_i(q)-\bbY(q)|^2 /N = o(1)$. 
Then
$
\sqrt N(\hbt - \bt )  \rightsquigarrow \mN( 0,  V) 
$
with
$ V $ denoting the limiting value of $N\cov(\htau) $.
\end{lemma}

\begin{lemma}\citep[][Theorem A.1]{ohlsson} \label{os}
For $W = 1, 2, 3, \dots$,  let $\{\xi_{W,w}: w = \ot{W}\}$ be a martingale difference sequence relative to the filtration $\{\mF_{W,w}: w = \zt{W}\}$, and let $X_W$ be an $\mF_{W,0}$-measurable random variable. 
Set 
$\xi_W = \sumw \xi_{W,w}$. 
Suppose that the following three conditions are fulfilled as $W \to \infty$.
\begin{enumerate}[(i)]
\item\label{os.1} $\sumw E(\xi_{W,w}^4) = o(1).$
\item\label{os.2} For some sequence of non-negative real numbers $\{\beta_W: W = 1, 2, 3, \dots\}$ with $\sup_{W\geq 1} \beta_W <\infty$, we have $ E\big[ \{ \sumw E(\xi_{W,w}^2\mid \mF_{W,w-1}) - \beta_W^2 \}^2\big] = o(1).$
\item\label{os.3} 
$\mathcal{L}(X_W) * \mN(0, \beta_W^2) \rightsquigarrow \mathcal{L}_0$ for some probability distribution $\mathcal{L}_0$, where $*$ denotes convolution. 
\end{enumerate} 
Then $\mathcal{L}(X_W + \xi_W)  \rightsquigarrow \mathcal{L}_0$ as $W\to \infty.$
\end{lemma}

\subsection{New lemmas}

We give in this subsection the key lemmas for verifying the results on design-based inference under the $2^2$ split-plot design.
The lemmas and their proofs extend to the general $\ta\times \tb$ design with minimal modification.

\medskip

\textbf{Decomposition of $\cov(\byht)$ in finite samples.}
Lemma \ref{adam} below separates the parts in $  \cov(\byht)$ that are due to $ \pphi$ and $\delta$ from \eqref{eq:decomp}, respectively. The decomposition furnishes an alternative proof of Lemma \ref{Vmat} relative to \citet{MD}. 
\begin{lemma}\label{adam}
Under the $2^2$ split-plot randomization, we have
$$ \cov(\byht) = \cov( \pphi) +\cov(\delta)$$
with $\cov( \pphi)  = W^{-1} ( H \circ  \ssht)$ and $\cov(\delta)  = \sumw \cov(\delta_w) =  W^{-1}\Psi$, where $\cov( \delta_w) = W^{-2} M_w^{-1}  ( H_w \circ  S_w)$. 
\end{lemma}

\begin{proof}[Proof of Lemma \ref{adam}]
The identities $\cov(\byht) = \cov( \pphi) +\cov(\delta)$ and  $
 \cov( \delta ) = \sumw \cov( \delta_w)$ follow from \eqref{eq:sigA}. 
We verify below the analytic forms of $\cov( \pphi)$ and $\cov(\delta_w)$, respectively. 

For the analytic form of $\cov( \pphi)$, 
define $ U_w(a) = (U_w(a0), U_w(a1) )^\T$
as the vector potential outcome of whole-plot $w$ under $A_w=a \in \{0,1\}$. 
The finite-population mean and covariance of $\{U_w(a)\}_{w=1}^W$ equal $\bbU(a) = W^{-1}\sumw U_w(a) 
=  (\bar Y(a0), \bar Y(a1)  )^\T$  
and 
\begina
 \ssht(a) =  (W-1)^{-1}\sumw \{ U_w(a) - \bbU(a)\}^2 = \left(\begin{array}{cc} \ssht(a0,a0) & \ssht(a0,a1)\\ \ssht(a1,a0) & \ssht(a1,a1)\end{array}\right) \quad(a=0,1),
\enda
respectively. Direct comparison shows that $ \pphi$ equals the sample analog of 
$ \by = (\bar U(0)^\T, \bar U(1)^\T)^\T = \left( I_2,  0_{2\times 2}\right)^\T \bbU(0) +  \left( 0_{2\times 2},  I_2\right)^\T \bbU(1) $
with regard to the stage (I) randomization. 
It then follows from Lemma \ref{Ding17} that 
\begin{align*}
\cov( \pphi)
\, & = \, W_0^{-1}  \left(\begin{array}{c}
 I_2\\ 0_{2\times 2}\end{array}\right)  \ssht(0) \left( I_2,  0_{2\times 2}\right) + W_{1}^{-1}  \left(\begin{array}{c}
 0_{2\times 2}\\ I_2\end{array}\right)  \ssht(1)  \left( 0_{2\times 2},  I_2\right) - W^{-1}  \ssht\\
\, & = \, W^{-1} ( H \circ  \ssht).
\end{align*}
%
%
%
%

For the analytic form of $\cov( \delta_w)$, 
it follows from the definition of $\delta_w$ and \eqref{eq:covU} that 
\beginy\label{eq:deltaw}
 \cov(
 \delta_w \mid A_w=a) = W_a^{-2} p_aM_w^{-1}  \{ H_w(a) \circ  S_w\} = p_a^{-1} W^{-2} M_w^{-1}  \{ H_w(a) \circ  S_w\} 
\endy
with
$ 
H_w(0) =  \diag(p_0^{-1}, 0) \otimes \{ \diag(q_{w0}^{-1}, q_{w1}^{-1} ) - \bbo{2} \}$ and $
H_w(1) =  \diag(0, p_1^{-1}) \otimes \{ \diag(q_{w0}^{-1}, q_{w1}^{-1} ) - \bbo{2}  \}$
corresponding to the upper-left and lower-right $2\times 2$ block matrices of $H_w$, respectively. 
It then follows from \eqref{eq:sigA} and  $ H_w(0) +  H_w(1) =  H_w$ that 
\begina
\cov( \delta_w)= E\{\cov( \delta_w\mid A_w)\} = \sum_{a=0,1} \pr(A_w=a)\cdot \cov( \delta_w\mid A_w = a) =  W^{-2}M_w^{-1}  ( H_w \circ  S_w).
\enda 
\end{proof}
%
\textbf{The weak law of large numbers.}
%
We give in this part the weak law of large numbers for quantifying the probability limits of $\hys$ and $\hat V_* \ (* = \sht, \shaj)$ under some weaker conditions than Condition \ref{asym}, summarized in Condition \ref{asym0_S} below.

 \begin{condition}\label{asym0_S} 
As $W$ goes to infinity, for all $a,b = 0,1$, and $z\in \mt$, 
\begin{enumerate}[(i)]
\item \label{qb} $p_a$ has a limit in $(0,1)$; $\epsilon \leq \min_{w = \ot{W}} q_{wb} \leq \max_{w = \ot{W}} q_{wb} \leq 1 - \epsilon$ for some $\epsilon \in (0,1/2]$ independent of $W$; 
\item\label{asym0}  $\bar Y$ has a finite limit; $S = O(1)$ and $\Psi = O(1)$; 
\item\label{mm}$W^{-2}\sumw \alpha^4_w   \oyfzw   = o(1)$.
 \end{enumerate}
\end{condition}

Condition \ref{asym} ensures that $\{Y_{ws}(z): z\in\mt\}_{ws\in\mathcal S}$, $\{Y_{ws}'(z): z\in\mt\}_{ws\in\mathcal S}$, and the finite population of all ones, namely $\{Y_{ws}(z) = 1: z\in\mT\}_{ws\in\mathcal S}$,  all satisfy Condition \ref{asym0_S}\eqref{asym0}--\eqref{mm}.

\begin{lemma}\label{wlln}
Assume  split-plot randomization. Then
\begin{enumerate}[(i)]
\item\label{wlln_Y} $\byht  -  \bar Y = \op$ 
provided Condition \ref{asym0_S}\eqref{qb}--\eqref{asym0}; 
\item\label{wlln_1} 
$
\hat 1_\sht = \diag\{ \oneHT\}_{z\in\mt} = I_{|\mt|} + \op
$
provided Condition \ref{asym0_S}\eqref{qb} and $\oat = O(1)$.
\end{enumerate}
 \end{lemma}
 
 \begin{proof}[Proof of Lemma \ref{wlln}]
 Standard result ensures $ E(\byht) =\bar Y $. 
 Lemma \ref{Vmat} ensures 
 $\cov(\byht)  =W^{-1}( H \circ \ssht  + \Psi) = o(1)$
 under Condition \ref{asym0_S}\eqref{qb}--\eqref{asym0}.
 The result for $\byht$ then follows from Markov's inequality. 
The result for $\hat 1_\sht$ follows from 
applying statement \eqref{wlln_Y} to the finite population of all ones.  
\end{proof}

\begin{lemma}\label{bounded}
Assume split-plot randomization and Condition \ref{asym0_S}\eqref{qb} and \eqref{mm}. Then
\begin{enumerate}[(i)]
\item\label{bounded.b} $W^{-2}\sumw E \{ \hat U^2_w(z)   \hat U^2_w(z')   \mid A_w = a \} =o(1) \ $ for $z=(ab)$ and $z'=(ab')$; 
\item\label{bounded.c} 
$W^2 \sumw  E(\|\delta_w \|^4_2 \mid A_w = a) = o(1); \quad W^2 \sumw  E(\|\delta_w \|^4_2) = o(1)$.
\end{enumerate}
\end{lemma}

\begin{proof}[Proof of Lemma \ref{bounded}]
For statement \eqref{bounded.b}, 
it follows from 
$
\hat U_w(z)  =  M_{w b}^{-1}\sumsz   U_{ws}(z)$, \eqref{cs}, and $q_{wb} \geq \ep$ by Condition \ref{asym0_S}\eqref{qb}
that 
\begina
\hat U_w^4(z) \leq M_{w b}^{-1} \sumsz U^4_{ws}(z) \leq M_{w b}^{-1} \sums  U^4_{ws}(z) = q_{wb}^{-1} \ \oufzw \leq   \epsilon^{-1} \ \oufzw ,
\enda
recalling $ \oufzw = M_w^{-1} \sums  U^4_{ws}(z)$.  
This ensures 
$$
 W^{-2}\sumw E \{  \hat U^4_w(z)  \mid A_w = a \}  \leq  W^{-2}\epsilon^{-1}\sumw     \oufzw  = o(1)
$$
by Condition \ref{asym0_S}\eqref{mm}. 
The result then follows from $\hat U^2_w(z) \hat U^2_w(z')  \leq 2^{-1}   \{ \hat U^4_w(z) +   \hat U^4_w(z')  \}$.

For statement \eqref{bounded.c},  it suffices to verify the first equality, namely
\beginy\label{eq:d4}
W^2\sumw E (  \|\delta_w\|^4_2 \mid A_w=a ) 
 = o(1).
 \endy
 The second equality then follows from the law of total expectation. 
 
To verify \eqref{eq:d4},  note that $ \|\delta_w\|^2_2 = W_a^{-2} [\{  \hat U_w(a1) - U_w(a1)\}^2 + \{  \hat U_w(a0) - U_w(a0)\}^2 ]$ for  $w$ with $A_w=a$.
This ensures
\begina
 \|\delta_w\|^4_2 =  (\|\delta_w\|^2_2  ) ^2 
 \leq 2W_a^{-4} \{  \hat U_w(a1) - U_w(a1) \}^4 +2W_a^{-4}  \{  \hat U_w(a0) - U_w(a0) \}^4
\enda
for $w$ with $A_w = a$ by the Cauchy--Schwarz inequality and hence
\begina
E( \|\delta_w\|^4_2 \mid A_w=a )   
& \leq& 2W_a^{-4} E \big[ \{  \hat U_w(a1) - U_w(a1) \}^4\mid A_w=a \big] \\
&&+2W_a^{-4} E \big[  \{  \hat U_w(a0) - U_w(a0) \}^4 \mid A_w=a \big]. 
\enda  
A sufficient condition for \eqref{eq:d4} is thus 
\beginy\label{eq:U4}
W^{-2}\sumw E\big[\{  \hat U_w(z) - U_w(z)\}^4  \mid A_w = a\big] = o(1) \qquad \text{for}\ \ z = (ab).
\endy
With $$
\sumw E \big[\{  \hat U_w(z) - U_w(z)\}^4  \mid A_w = a \big] \leq 8\sumw E \{ \hat U^4_w(z)   \mid A_w = a \} +  8\sumw U_w^4(z) 
$$
by 
 \eqref{cs},
\eqref{eq:U4} is guaranteed by statement \eqref{bounded.b}
and 
$
W^{-2}\sumw U_w^4(z) 
\leq W^{-2}\sumw      \oufzw    = o(1)
$
by \eqref{cs} and  Condition  \ref{asym0_S}\eqref{mm}.

\end{proof}

Let 
$$\hat T_{z,z'} =W_a^{-1}\sum_{w: A_w = a}  \hU_w(z)\hU_w(z')$$
for $z=(ab)$ and $z'=(ab')$, as the sample analog of $$\overline{U(z)U(z')} = W^{-1} \sumw U_w(z) U_w(z').$$
 Let 
$\Psi(z,z') = W^{-1} \sumw M_w^{-1}  H_w(z,z') S_w(z,z')$ be the $(z,z')$th element 
of $\Psi$. 
Lemma \ref{lln} below states the convergence of $\hat T_{z,z'}$ to its expectation, affording the basis for computing the probability limits of $\hV_\sht$ and $\hV_\shaj$.
\begin{lemma}\label{lln}
Assume split-plot randomization and Condition \ref{asym0_S}. 
Then
\begina
\hat T_{z,z'} - E(\hat T_{z,z'}) = \op \qquad \text{for} \ \ z=(ab), \ z'=(ab')
\enda
 with $E(\hat T_{z,z'}) =\overline{U(z)U(z')} +  p_a\Psi(z,z') = (1-W^{-1})\ssht(z,z') + \bY(z)\bY(z') + p_a \Psi(z,z')$.
\end{lemma}

\begin{proof}[Proof of Lemma \ref{lln}]
By Markov's inequality,  it suffices to verify the expression of $E(\hat T_{z,z'})$ and $\cov(\hat T_{z,z'}) = o(1)$. 
Let $X_w = 1(A_w = a) \, \hU_w(z)\hU_w(z')$ to write $\hat T_{z,z'}  = W_a^{-1}\sumw   X_w$. 
Let $\mu_w = E(X_w \mid A_w=a) =E\{\hU_w(z)\hU_w(z') \mid A_w=a\}$.

\pa{Expression of $E(\hat T_{z,z'})$} First, it follows from \eqref{eq:covU} that  \begina
\mu_w 
&=&  \cov\{\hU_w(z), \hU_w(z')\mid A_w=a \} + E\{\hU_w(z)\mid A_w=a\} \, E\{\hU_w(z')\mid A_w=a \}\nonumber \\
&=&  p_a  M_w^{-1}H_w(z,z') S_w(z,z') + U_w(z) U_w(z'). 
 \enda
 This, together with 
\beginy\label{ss_exa}
E(X_w) 
= E\{E(X_w \mid A_w) \} = p_a E(X_w \mid A_w=a) = p_a \mu_w, 
 \endy
ensures
\begina
E(\hat T_{z,z'}) &=& W_a^{-1}\sumw E(X_w)  =W ^{-1}\sumw \mu_w =  W^{-1} \sumw U_w(z) U_w(z') + p_a  \Psi(z,z')\\
&=& (1-W^{-1})   \ssht(z,z') + \bY(z) \bY(z')+ p_a  \Psi(z,z'); 
\enda
the last equality follows from 
 $W^{-1} \sumw U_w(z) U_w(z')  = (1-W^{-1})   \ssht(z,z') + \bY(z) \bY(z')$. 

\pa{Limit of $\cov(\hat T_{z,z'})$} 
By \eqref{ss_exa}, 
\begina
E\{ \cov(X_w \mid A_w ) \}  & =& p_a  \cov(X_w \mid A_w =a) = p_a    E(X_w^2  \mid A_w =a) - p_a \mu_w^2,\\
 \cov\{ E(X_w \mid A_w) \}  & =&  E\big[ \{ E(X_w \mid A_w)\}^2 \big] - \{E(X_w)\}^2 \\
 & =&  p_a   \{E(X_w \mid A_w=a) \}^2 - p_a^2 \mu_w^2 = p_a  \mu_w^2  - p_a^2 \mu_w^2,\\
 E \{ E(X_w \mid A_w) \cdot E(X_k \mid A_k) \}  & =& \pr(A_w=A_k=a) \cdot E(X_w \mid A_w=a) \cdot E(X_k \mid A_k=a)\\
   & =&  p_a   \frac{W_a-1}{W-1}  \mu_w \mu_k , \\
 E(X_w) E(X_k)  & =&  p_a^2  \mu_w \mu_k. 
 \enda
This ensures
\begin{align*}
 \cov(X_w) \, & = \,  E\{ \cov(X_w \mid A_w ) \}  + \cov\{ E(X_w \mid A_w) \}\nonumber \\
 \, & = \,  p_a  E(X_w^2   \mid A_w =a)  - p_a^2  \mu_w^2,\\
 \cov(X_w, X_k) 
 \, & = \,  \cov \{ E(X_w \mid A_w), E(X_k \mid A_k) \} +  E \{ \cov(X_w,X_k \mid A_w, A_k )  \} \nonumber\\
 \, & = \,  \cov \{ E(X_w \mid A_w), E(X_k \mid A_k) \}\nonumber\\
\, & = \,  E \{ E(X_w \mid A_w) \cdot E(X_k \mid A_k)\}  - E(X_w) E(X_k)\nonumber\\
\, & = \,  -p_0p_1(W-1)^{-1}\mu_w\mu_k \qquad (w \neq k). 
  \end{align*}
Thus, 
\begina
 W_a^{2} \cov(\hat T_{z,z'})  &=& \sum_{w, k} \cov(X_w, X_k)= 
\sumw \cov(X_w) +  \sum_{w\neq k} \cov(X_w, X_k) \\
&=&  p_a\sumw   E(X_w^2   \mid A_w =a)  - p_a^2 \sumw  \mu_w^2 - \frac{p_0p_1}{W-1}\sum_{w\neq k}\mu_w\mu_k\\
 & =& p_a\sumw   E(X_w^2   \mid A_w =a)  - p_a^2 \sumw  \mu_w^2 + \frac{p_0p_1}{W-1}\sumw \mu^2_w - \frac{p_0p_1}{W-1} \sum_{w,k}\mu_w\mu_k\\
 &\leq &p_a\sumw   E(X_w^2   \mid A_w =a) - \left(p_a^2 - \frac{p_0p_1}{W-1}\right) \sumw \mu^2_w
\\
 &\leq_\infty&  p_a\sumw   E(X_w^2   \mid A_w =a) =p_a\sumw E \{\hU^2_w(z)\hU^2_w(z')\mid A_w =a \} , 
\enda
where $\leq_\infty$ indicates less or equal to as $W\to\infty$. 
The result then follows from Lemma \ref{bounded}\eqref{bounded.b}. 
\end{proof}


\subsection{Proof of the main results}
\begin{proof}[Proof of Theorem \ref{clt}.]
%
We verify below the results for $\byht$ and $\byhaj$, respectively. 

\medskip

\noindent\textbf{Asymptotic Normality of $\byht$.} From \eqref{eq:decomp}, $\sqrt W( \byht -\bbY )=\sqrt W \left( \pphi  - \bbY +\delta \right)$.  
The Cramer--Wold device ensures that $\sqrt W( \byht -\bbY ) \rightsquigarrow \mN(0,\Sigma_\sht)$ as long as 
\beginy\label{eq:cw}
  \eta^\T \sqrt W \left( \pphi  - \bbY +\delta \right)  \rightsquigarrow \mN(0,\eta^\T\Sigma_\sht\eta)
\endy
 for arbitrary non-random $4\times 1$ unit vector $ \eta$.

Let $X =\eta^\T \sqrt W ( \pphi  - \bbY)$, $ \xi_w  =\eta^\T  \sqrt W \delta_w$, and $\xi  = \sumw \xi_w  = \sqrt W \eta^\T \delta $ to write the left-hand side of \eqref{eq:cw} as $X + \xi$. 
Let $\mF_{W,0} =\mathcal{A}$ be the $\sigma$-algebra generated by $(A_w)_{w=1}^W$, and let $\mF_{W,w}$ be the $\sigma$-algebra generated by $(A_w)_{w=1}^W$ and $\{(B_{vs})_{s=1}^{M_v}: v = \ot{w}\}$ for $w = \ot{W}$. 
Then $\mF_{W,0} \subset \mF_{W,1}\subset \dots \subset \mF_{W,W}$ such that $\{\mF_{W,w}:  w = \zt{W}\}$ is a filtration. 
Intuitively, $\mF_{W,0}$ contains all the information on the stage (I) cluster randomization, whereas $\mF_{W,w}$ contains all the  information on the stage (I) cluster randomization plus the subset of stage (II) stratified randomization in the first $w$ whole-plots,  $v = 1, \dots, w$.
We verify below the sufficient condition \eqref{eq:cw} by checking that $( \xi_w )_{w=1}^W$ and $X$  
satisfy the three conditions of Lemma \ref{os} with $\beta_W^2 = \eta^\T\Psi  \eta$ with regard to  filtration $\{\mF_{W,w}: w = \zt{W}\}$.
Technically,  $X = X_W $, $\xi_w = \xi_{W,w}$, and $\Psi = \Psi_W$ all depend on $W$. We suppress the $W$ in the subscripts when no confusion would arise.

For Lemma \ref{os} condition \eqref{os.1}, \eqref{cs} ensures 
\beginy\label{xi4}
\xi_w ^4 =  W^2(\eta^\T\delta_w )^4 \leq W^2 \|  \eta\|^4_2 \cdot  \|\delta_w \|^4_2 =
W^2  \|\delta_w \|^4_2.
\endy
The result then  follows from $
\sumw E(\xi_w ^4)\leq  W^2 \sumw  E(\|\delta_w \|^4_2) = o(1)$
by Lemma \ref{bounded}\eqref{bounded.c}.

For Lemma \ref{os} condition \eqref{os.2}, let $
\sigma = \var (\xi \mid \mF_{W,0})$ and $\sigma_w  =\var(\xi_w \mid \mF_{W,0})=\var(\xi_w \mid A_w)$ with $\sigma = \sumw \sigma_w$. 
It follows from \eqref{eq:sigA} and Lemma \ref{adam} that $  E(\xi \mid \mF_{W,0}) = 0$ and 
$ \var(\xi ) =\eta^\T \Psi  \eta  =  \beta_W^2$. 
This, together with $E(\xi_w ^2\mid \mF_{W,w-1} ) = E(\xi_w ^2\mid \mF_{W,0}) =\sigma_w$, ensures
\begina
\sumw E(\xi_w ^2\mid \mF_{W,w-1} ) = \sigma, \quad 
\beta_W^2  =  \var(\xi ) = E\{\var(\xi  \mid \mF_{W,0})\}  + \var\{ E(\xi \mid \mF_{W,0})\} = E(\sigma)
\enda
such that 
$ E\big[ \{ \sumw E(\xi_w ^2\mid \mF_{W,w-1} ) - \beta_W^2 \}^2\big] 
=E\big[\{\sigma - E(\sigma) \}^2\big]
= \var(\sigma)
$.
Lemma \ref{os} condition \eqref{os.2} is thus equivalent to 
\beginy\label{eq:cond2}
 \var(\sigma) = o(1).
 \endy
To verify \eqref{eq:cond2}, 
view $\sigma_w   =\var(\xi_w \mid A_w)$ 
as the observed value of 
\beginy\label{eq:sig_a}
\sigma_w(a) = \var ( \xi_w  \mid A_w =a )= W \eta^\T \cov\left( \delta_w  \mid A_w =a \right)\eta
\endy
with mean $\bar\sigma(a) = W^{-1}\sumw \sigma_w(a)$, 
variance $S^2_{\sigma(a)} 
= (W-1)^{-1}\sumw \{\sigma_w(a) - \bar\sigma(a)\}^2$, and sample mean $\hat\sigma(a) = W_a^{-1}\sum_{w:A_w=a} \sigma_w(a)$ for $ a=0,1$. 
Standard result ensures $
 \var \{\hat\sigma(a)\}  =W^{-1} p_a^{-1}(1-p_a) S^2_{\sigma(a)}$
such that, with $\sigma = W_0 \hat \sigma(0) + W_1 \hat \sigma(1)$, we have
\begina
\var ( \sigma  ) &=& \var \{  W_0 \hat \sigma(0) + W_1 \hat \sigma(1) \}\nonumber\\
 &\leq &   2\var \{  W_0 \hat \sigma(0) \} + 2 \var \{ W_1 \hat \sigma(1) \}
 = 2 W p_0p_1 (S^2_{\sigma(0)} + S^2_{\sigma(1)}).
\enda
Sufficient condition \eqref{eq:cond2} thus holds as long as $W S^2_{\sigma(a)} = o(1)$ for $a = 0,1$. 
With  $
(W-1)S^2_{\sigma(a)} = \sumw \{\sigma_w(a)\}^2 -W\{ \bar\sigma(a)\}^2$, this is in turn guaranteed by
\beginy\label{eq:ss_ss}
\sumw \{\sigma_w(a)\}^2 =  o(1), \qquad W\{ \bar\sigma(a)\}^2 =  o(1)\qquad(a=0,1).
\endy
We verify below the two sufficient conditions in \eqref{eq:ss_ss}. 

First, 
 \eqref{eq:sig_a} and $E ( \xi _w  \mid A_w =a )=0$ together ensure  $\sigma_w(a) = E ( \xi^2_w  \mid A_w =a )$ and hence
\begina
\{\sigma_w(a)\}^2 =  \{ E ( \xi^2_w  \mid A_w =a ) \}^2 \leq E ( \xi^4_w  \mid A_w =a ) \leq W^2 E ( \|\delta_w \|^4_2\mid A_w =a )
\enda
by Jensen's inequality and \eqref{xi4}. The first equality in \eqref{eq:ss_ss} then follows from 
\begina
\sumw \{\sigma_w(a)\}^2 \leq W^2 \sumw E ( \|\delta_w \|^4_2\mid A_w =a ) = o(1) 
\enda
by Lemma \ref{bounded}\eqref{bounded.c}. 
The second equality in \eqref{eq:ss_ss} follows from \eqref{eq:sig_a} and \eqref{eq:deltaw} as
\begina
\bar\sigma(a) 
= \eta^\T \sumw\cov ( \delta_w  \mid A_w =a ) \eta
=
\eta^\T \left[ p_a^{-1} W^{-2}\sumw M_w^{-1}  \{ H_w(a) \circ  S_w\} \right]\eta = O(W^{-1})  
\enda
by $\Psi = O(1)$. This verifies \eqref{eq:ss_ss} and hence Lemma \ref{os} condition (ii).

Lemma \ref{os} condition \eqref{os.3} then follows from Lemmas \ref{Ding17} and \ref{adam} which ensure $\sqrt W ( \pphi  - \bbY) \rightsquigarrow \mN( 0, H\circ \ssht)$ under Condition \ref{asym}.
The convolution of $\mathcal{L}(X )$ with $ \mN( 0, \eta^\T\Psi  \eta)$ thus  converges in distribution to 
$\mN(0,\eta^\T\Sigma_\sht\eta)$ by the convergence of the characteristic function.

This verifies that $( \xi_w )_{w=1}^W$ and $X$  
satisfy the three conditions in Lemma \ref{os}. The sufficient condition \eqref{eq:cw} then follows from Lemma \ref{os} and ensures the result for $\byht$.

\pa{Asymptotic Normality of $\byhaj$}
Recall from \eqref{eq:haj_intuition} in the main paper that 
$
\byhaj - \bbY = \hat{ 1}_\sht^{-1} \, \byht'
$
with $\hat{ 1}_\sht = \diag\{\oneHT \}_{z\in\mt}$ and  $\byht'$ as the vectorization of $\{\hY_\sht'(z)\}_{z\in \mt}$.
The asymptotic Normality of $\byht$ extends to $\byht'$ as $\sqrt W( \byht' -\bbY ) \rightsquigarrow \mN(0,\Sigma_\shaj)$. 
The result for $\byhaj$ then follows from Slutsky's theorem with $\hat{ 1}_\sht^{-1} =  I_{|\mt|} + \op$ by Lemma \ref{wlln}. 
\end{proof}

\begin{proof}[Proof of Corollary \ref{V_hajek}]
By \eqref{eq:ss_diff}, we have
\begina
(W-1)  \{\sshaj(z,z) -\ssht(z,z) \} = \{\bar Y(z)\}^2\left(\sumw \alpha_w^2 + W\right) - 2 \bar Y(z) \sumw \alpha_w^2 \bar Y_w(z).
\enda
When $\bar Y_w(z) = c$ for all $w = \ot{W}$, we have $\bar Y(z) = W^{-1}\sumw \alpha_w 
\bar Y_w(z) = c$  such that
$
 (W-1)  \{\sshaj(z,z) -\ssht(z,z) \} = c^2  (  W - \sumw \alpha^2_w  )   \leq 0$;  the equality holds if and only if $a_w=1$ for all $w$ or $c=0$. 

When $U_w(z) = c$ for all $w = \ot{W}$, we have $\bar Y(z) = W^{-1}\sumw U_w(z) = c$ such that 
$
(W-1)  \{\sshaj(z,z) -\ssht(z,z) \} = c^2 \sumw (\alpha_w-1)^2 \geq 0$; the equality holds if and only if $a_w=1$ for all $w$ or $c=0$.  
\end{proof}

\begin{proof}[Proof of Theorem \ref{covEst}]
Let $\hat V_*(z,z')$ be the $(z,z')$th element of $\hV_*$ for $* = \HT, \haj$.
Assume  $z=(ab)$ and $z'=(ab')$ with the same level of factor A throughout the proof. 

\pa{Result on $\hbV_\sht$} 
Direct algebra shows that $\hat V_\sht(z,z') = W_a^{-1} \hat S_\sht(z,z') = (W_a-1)^{-1}  \{ \hat T_{z,z'} -  \yHT \hY_\sht(z') \}$ for $z=(ab)$ and $z'=(ab')$.
It then follows from Lemmas \ref{wlln} and \ref{lln} that
\begina
W  \hat V_\sht(z,z') 
&=& p_a^{-1} \{ E(\hat T_{z,z'})  -   \bY(z) \bY(z')  \}+ \op \\
&=&  p_a^{-1}  \{ \ssht(z,z')  + p_a  \Psi(z,z') \}+ \op\\
&=&  \Sigma_\sht(z,z') + \ssht(z,z')+ \op,
 \enda
where $\Sigma_\sht(z,z')=   (p_a^{-1}-1) \ssht(z,z') + \Psi(z,z')$ is the $(z,z')$th element of $\Sigma_\sht$. This verifies the probability limit of $W\hbV_\sht$.

\pa{Result on $\hbV_\shaj$} Direct algebra shows that
\beginy\label{eq:limit_ss}
\quad \quad \quad  (W_a-1)\hat V_\shaj(z,z') 
&=&  \hat T_{z,z'}  +\hy_\shaj(z) \hY_\shaj(z') \left( W_a^{-1}\sum_{w:A_w=a}\alpha_w^2\right)\\
&&- \byhaj(z) \left\{  W_a^{-1}\sum_{w:A_w=a}\alpha_w\hU_w(z')\right\}\nonumber\\
&&-   \byhaj(z') \left\{ W_a^{-1}\sum_{w:A_w=a}\alpha_w\hU_w(z)\right\} \quad \text{for $z=(ab), \ z'=(ab')$.}\nonumber
\endy
We compute below the probability limit of the right-hand side of \eqref{eq:limit_ss}. 

First, apply Lemma \ref{lln} to $\{Y_{ws}(z), Y_{ws}(z')\}_{ws\in\mathcal S}$ with $Y_{ws}(z') = 1$ for all $\ws\in\mathcal{S}$ to see
\begina
W_a^{-1}\sum_{w:A_w=a}\alpha_w\hU_w(z) - \oau = \op.
\enda
This, together with $\oau = O(1)$
by 
$
 |\oau - \bY(z) | =  W^{-1}|\sumw  \alpha_w \{ U_w(z) - \bY(z)\}| \leq 
  \oat  \, \ssht (z,z),
$ ensures
\beginy\label{eq:ss_1}
\byhaj(z')   \left\{W_a^{-1}\sum_{w:A_w=a}\alpha_w\hU_w(z)\right\} -  \bY(z') \, \oau  = \op.
\endy

Likewise for $
W_a^{-1}\sum_{w:A_w=a}\alpha_w^2 - \oat = \op$ by letting $Y_{ws}(z) = Y_{ws}(z')=1$ for all $\ws \in \ms$  in 
Lemma \ref{lln}. 
This, together with $\oat =O(1)$ by Condition \ref{asym}\eqref{alpha}, ensures
\beginy\label{eq:ss_2}
\byhaj (z) \hY_\shaj(z') \left(W_a^{-1}\sum_{w:A_w=a}\alpha_w^2\right) - \bY(z) \bY(z') \, \oat= \op.
\endy
Plug \eqref{eq:ss_1}--\eqref{eq:ss_2} and the probability limit of $\hat T_{z,z'}$ from Lemma \ref{lln} in \eqref{eq:limit_ss} to see
\begina
 (W_a-1)\hat V_\shaj(z,z') 
&=&   \{ \ssht(z,z') + \bY(z)\bY(z') + p_a\Psi(z,z') \}  +\bY(z) \bY(z') \, \oat\nonumber\\
&& -\bY(z)\, \oaup -  \bY(z') \, \oau + \op \\
&=&  p_a\Psi(z,z')+ (1-W^{-1}) S_\haj(z,z') + W^{-1}\ssht(z,z') +  \op
\enda
by \eqref{eq:ss_diff}. 
This ensures $ W \hat V_\shaj(z,z') = p_a^{-1}S_\haj(z,z') + \Psi(z,z')+\op$, and the result follows from $p_a^{-1} S_\shaj(z,z')  = H(z,z')\sshaj(z,z') + \sshaj(z,z')$.
\end{proof}

%
\section{Reconciliation with model-based inference}\label{sec:lm_app}
%
\subsection{Notation and useful facts}\label{sec:notation_lm1}
Assume $\pi_{ws}= N^{-1}\{p_{ws}(\Zws)\}^{-1}$  as the weight for sub-plot $ws$ under fitting scheme ``\wls''. 
It differs from the original weight $\{p_{ws}(\Zws)\}^{-1}$ by a constant factor of $N^{-1}$ and thus does not affect the result of the \wlss fit. 
We have
\beginy\label{eq:ht_pi}
\yHTn(z) = \sumwszz \pi_{ws} \Yws, \qquad \oneHT = \sumwszz \pi_{ws}
\endy
 with 
\beginy\label{eq:piws}
\pi_{ws}= N^{-1}p_a^{-1} q_{wb}^{-1} = \alpha_w W_a^{-1} M_{wb}^{-1}
\endy
for sub-plots under treatment $z=(ab)$.

Recall $d_{ws} = ( 1(\Zws=00),1(\Zws=01),1(\Zws=10),1(\Zws=11) )^\T$ and $\dwab  =  ( 1(A_wb=00),1(A_wb=01),1(A_wb=10),1(A_wb=11) )^\T$ as the regressor vectors in  regressions \eqref{lm_t} and \eqref{lm_t_ag}. Let $\tbeta_\dg (z)$ be the $z$th element in $\tbeta_\dg $ 
corresponding to $1(Z_{ws} = z)$ or $1(A_wb = z)$ for $\dgss$, respectively.
Let
\begina
Y \sim D, \qquad U \sim D_\ag
\enda
be the matrix forms of \eqref{lm_t} and \eqref{lm_t_ag}, respectively, 
where 
 $Y$ and $U$ are the vectorizations of $\{\yws: \ws \in\mathcal{S}\}$ and 
$ \{\hU_w(A_wb): w = \ot{W}; \ b = 0, 1\}$, and $ D$ and $ D_\ag $ are the matrices with rows $\{ d_{ws}: \ws  \in \mathcal{S}\}$ and $\{ \dwab : w = \ot{W}; \ b = 0,1\}$, respectively. 
Let $\bp = \diag(\pi_{ws})_{ws\in\ms}$ be the corresponding weighting matrix under fitting scheme ``wls''. 
Assume lexicographical orders of $ws$ and $(w, b)$  
 throughout unless specified otherwise.  

Let $Y_w= ( Y_{w1}, \dots,  Y_{wM_w})^\T$, $D_w=  ( d_{w1}, \dots,  d_{wM_w})^\T$, $\Pi_w = \diag(\pi_{ws})_{s=1}^{M_w}$,  
$U_w = (\hU_w(A_w0), \hU_w(A_w1) )^\T$,  and
$\Dw = ( d_w(A_w0),  d_w(A_w1) )^\T$
 be the parts in $Y$, $D$, $\Pi$, $U$, and $D_{\ag }$ corresponding to whole-plot $w$, respectively. 
The cluster-robust covariances equal
\beginy\label{v_pi}
&&\tbV_\wls= ( \dt \bp D)^{-1} \left(\sumw  D_w^\T \bp_w \te_\wlw  \te_\wlw ^\T\bp_w  D_w\right)( \dt \bp D)^{-1},\\
&&\vlsu =( \dt _\ag  D_\ag )^{-1}\left( \sumw \Dw^\T\te_\tlw  \te_\tlw ^\T  \Dw \right)( \dt _\ag  D_\ag )^{-1},\label{v_u}
\endy
where $\te_\wlw  = (\te_{\wls,w1} , \dots, \te_{\wls,wM_w} )^\T$   and $\te_\tlw  =  (\te_\tlw (A_w0), \te_\tlw (A_w1) )^\T$ are the  residuals in whole-plot $w$ from \eqref{lm_t} and \eqref{lm_t_ag}, respectively.

Let $  \dpp = \diag(p_0, p_1)\otimes  I_2$ and $  Q_w = I_2 \otimes \diag(q_{w0}, q_{w1})$ 
with $\diag\{ p_{ws}(z)\}_{z\in\mt} = PQ_w$ for all $ws$. 
Let $R =   \diag(r_z)_{z\in\mt}$, where $r_z = N_z/N$. 
Some useful facts are
\beginy\label{matrix0}
\renewcommand{\arraystretch}{1.2}
\begin{array}{lll}
 N^{-1} \dt  D =  R,  & \quad\dt  \bp D  = \hat{ 1}_\sht,  &\quad W^{-1}  D_\ag ^\T  D_\ag  = \dpp,\\
N^{-1}  \dt   Y =  R  \hbY_\snm,  & \quad\dt \bp  Y = \hbY_\sht,  &\quad W^{-1}  D_\ag ^\T  U = \dpp   \hbY_\sht; \\
D_w^\T Y_w =  M_w Q_w \hY_w, &\quad
 D_w^\T\Pi_w Y_w  =  W^{-1} P^{-1} (\alpha_w\hY_w), & \quad
   D_\tlw ^\T U_w = \hat U_w,
\end{array} 
\endy
where $\hY_w$ and $\hU_w$ are the $4\times 1$ vectors of $\{\hY_w(z)\}_{z\in\mt}$ and $\{\hU_w(z)\}_{z\in\mt}$ in lexicographical order of $z$, respectively, with $\hY_w(z) = \hU_w(z) = 0$ for $z \not\in \{(A_w0), (A_w1)\}$  by definition. 
The proof of \eqref{matrix0} follows from direct algebra and is thus omitted.

\subsection{Proof of the main results}

\begin{proof}[Proof of Proposition \ref{prop:ls}]
The result follows from $\tbeta_\ols = (\dt  D)^{-1} \dt  Y $, $\tbeta_\wls = ( \dt  \bp D)^{-1} \dt\bp Y$,  $\tbeta_\ag  = (  D_\ag ^\T  D_\ag )^{-1} D_\ag ^\T U$, and \eqref{matrix0}.
\end{proof}

\begin{proof}[Proof of Theorem \ref{thm_vHats}]
We verify below the numeric expressions of $\tbV_\wls$ and $\vlsu$ in finite samples. The asymptotic equivalence then follows from Lemma \ref{wlln}. 

\medskip

\noindent\textbf{Numeric expression of $\tbV_\wls$.}
Proposition \ref{prop:ls} ensures $\te_\wlws  =  \Yws - \tbeta_\wls(Z_{ws}) = \Yws- \yhajn(\zws)$. 
Let  $\he_\wlw =  (     \he_\wlw(00),     \he_\wlw(01),    \he_\wlw(10),     \he_\wlw(11)   )^\T$ with 
\begina
\he_\wlw(z)  = M_{wb}^{-1}\sum_{s: Z_{ws}=z}\te_\wlws =
\left\{\begin{array}{cl}
 \hY_w(z) - \yhaj &\quad \text{for $\zin$,}\\
 0 & \quad \text{for $\znotin$.} 
\end{array}\right.
\enda
Set $Y_w = \te_\wlw  $ in \eqref{matrix0} 
to see 
$
 D_w^\T \Pi_w \te_\wlw   =   W^{-1} P^{-1} (\alpha_w\he_\wlw)$.
The ``meat'' part of \eqref{v_pi} thus equals
\begina
\sumw  D_w^\T \bp_w \te_\wlw   \te_\wlw  ^\T\bp_w  D_w
= W^{-2} P^{-1}\left(\sumw \alpha^2_w\he_\wlw\he_\wlw^\T\right) P^{-1}
= 
\big( \Omega_\wls(z,z') \big)_{z, z'\in \mt},
 \enda
 where 
 \begina
\Omega_\wls(z,z') &=& W_a^{-1}   W_{a'}^{-1}  \sumw  \alpha_w^2 \he_\wlw(z) \he_\wlw(z') \\
&=& \left\{
\begin{array}{cl} W^{-2}_a (W_a-1) \hat S_\shaj(z,z')  & \ \ \text{for $ z=(ab)$ and $z'=(a'b')$ with $a = a'$},\\
0  & \ \ \text{for $z=(ab)$ and $z' = (a'b')$ with $a\neq a'$.}
\end{array}\right.
\enda
The result for $\tbV_\wls$ then follows from $ \dt \bp D = \hat{ 1}_\sht$ by \eqref{matrix0}. 

\medskip

\noindent\textbf{Numeric expression of $\vlsu$.}
Proposition \ref{prop:ls} ensures $\te_\tlw (z) = \hU_w(z) - \tbeta_\ag (z) = \hU_w(z) - \yht $ for $\zin$. 
Set $U_w = \te_\tlw $ in \eqref{matrix0} 
to see $\Dw^\T \te_\tlw  = \he_\tlw$
with  $\he_\tlw =  (     e_\tlw(00),   e_\tlw(01),   e_\tlw(10),    e_\tlw(11)    )^\T$, where $e_\tlw(z) = 0$ for $\znotin$.   
The ``meat'' part of \eqref{v_u} thus equals
\beginy
 \sumw  D_\tlw ^\T  \te_\tlw \te_\tlw ^\T   D_\tlw  = \sumw \he_\tlw\he_\tlw^\T =  \big(    \Omega_\ag (z,z')  \big)_{z, z'\in \mt}, \label{eq:deed.ag}
 \endy
 where 
 \begina
 \Omega_\ag (z,z') =\sumw \te_\tlw (z) \, \te_\tlw (z') = \left\{  
 \begin{array}{cl}
 (W_a-1) \hat S_\sht(z, z') &\text{for $z=(ab)$ and $z'=(a'b')$ with $a = a'$,}\\
 0  &\text{for $z=(ab)$ and $z' = (a'b')$ with $a\neq a'$.}
 \end{array}\right.
 \enda
The result for $\vlsu$ then follows from $( \dt _\ag  D_\ag )^{-1} = \diag(W_0^{-1}, W_1^{-1})\otimes  I_2$
by \eqref{matrix0}. 
\end{proof}

\section{Regression-based covariate adjustment}\label{sec:ca_app}
\subsection{Notation and lemmas}
Inherit all  notation from Section \ref{sec:notation_lm1}. In addition, 
recall $\bar x_w = M_w^{-1}\sum_{s=1}^{M_w} x_{ws}$ as the whole-plot average covariate vector with $W^{-1}\sumw \alpha_w\bar x_w = \bar x = 0_J$. 
Let $v_{ws} =\alpha_w x_{ws}$, $v_w = \alpha_w\bar x_w$, and $U'_w(z) = \alpha_w\bar Y'_w(z) = \alpha_w\bY_w(z) - \alpha_w\bY(z)$ be the analogs of $U_{ws}(z)= \alpha_w Y_{ws}(z)$ and $U_w(z) = \alpha_w \bar Y_w(z)$ defined on the covariates and centered potential outcomes, respectively. 
We have 
\begina
\sx &=&  (W-1)^{-1}\sumw \alpha^2_w \bar x_w \bar x_w^\T =   (W-1)^{-1} \sumw v_wv_w^\T,\\
\swx &=& (M_w-1)^{-1} \alpha_w^2 \sum_{s=1}^{M_w} (x_{ws}-\bar x_w) (x_{ws}-\bar x_w)^\T =
(M_w-1)^{-1}   \sum_{s=1}^{M_w}  (  v_{ws} -v_w ) ( v_{ws} - v_w )^\T, \nonumber\\
\sxyz  
&=&   S_{Y(z)x}^\T = (W-1)^{-1}\sumw \alpha_w \bar x_w \{\alpha_w \bar Y_w(z) - \bar Y(z)\}
= (W-1)^{-1}\sumw v_w U_w(z), \\
\sxypz  &=&  S_{Y'(z)x}^\T = (W-1)^{-1}\sumw \alpha_w \bar x_w \{ \alpha_w \bar Y'_w(z) \} = 
  (W-1)^{-1}\sumw v_w U'_w(z),\nonumber\\
 \sxyzw &=&   S_{Y(z)x,w}^\T =S_{xY'(z),w} = S_{Y'(z)x,w}^\T = (M_w-1)^{-1} \alpha_w^2 \sum_{s=1}^{M_w} (x_{ws}-\bar x_w) \{Y_{ws}(z)-\bar Y_w(z)\}\\
&=& 
 (M_w-1)^{-1}  \sum_{s=1}^{M_w} ( v_{ws} -v_w )\{ U_{ws}(z)- U_w(z)\} .
\enda
Let  
\begina 
{  \hat S_{xx}(z )} = W_a^{-1} \sumwz \hbv_w(z)\hbv^\T_w(z), \qquad \hat S_{xY(z)} = W_a^{-1} \sumwz \hbv_w(z)\hU_w(z )
\enda
be the sample analogs of $\sxx$ and  $\sxyz $ based on sub-plots under treatment $z=(ab)$.
%

Recall $\qxx   = (N-1)^{-1}\sumws x_{ws}  x_{ws}^\T$ and $\qxyz    =  (N-1)^{-1}\sumws x_{ws}  Y_{ws}(z)$ as the finite-population covariances of $(x_{ws})_{\ws\in\ms}$ and $\{x_{ws}, Y_{ws}(z)\}_{\ws\in\ms}$, respectively. 
Let 
\begina
\hqxx = \lambda_N^{-1}\sumwszz\pi_{ws} x_{ws}  x_{ws}^\T, \qquad \hqxy =\lambda_N^{-1} \sumwszz\pi_{ws} x_{ws} Y_{ws}, 
\enda 
where $\lambda_N = 1-N^{-1}$, be their respective {\htf} estimators based on sub-plots under treatment $z$.

Recall $\Psi(z,z)=W^{-1} \sumw M_w^{-1} H_w(z,z)\sw(z,z)$ as the $(z,z)$th element of $\Psi$ from Lemma \ref{Vmat} with $H_w(z,z) = p_a^{-1}(q_{wb}^{-1}-1)$ for $z = (ab)$. Then 
\begina
  \Psi_{xx}(z,z) = W^{-1}\sumw M_w^{-1} H_w(z,z) \sxxw, \quad \Psi_{xY}(z,z) =  W^{-1} \sumw M_w^{-1} H_w(z,z) \sxyzw
\enda
by replacing $\sw(z,z)$ with $\sxxw$ and $\sxyzw $, respectively.
Let 
\begina
\uxx(z) =  \sx +   p_a  \Psi_{xx}(z,z),\qquad \uxyz = \sxyz + p_a  \Psi_{xY}(z,z).
\enda

\begin{lemma}\label{Axx}{\prex}
\begine[(i)]
\item\label{item:axx1}  $\hat x_\sht(z) =\sumwszz \pi_{ws} x_{ws} = \op$; 
\item\label{item:axx2} $\hqxx - \qxx= \op$, $\ \hqxy - \qxyz  =\op$; $\ \hat S_{xx}(z)  - \Uxx  (z) = \op$, $\ \hat S_{xY(z)} - \uxyz   = \op$  
\ende
 for $z\in\mt$. 
\end{lemma}

\begin{proof}[Proof of Lemma \ref{Axx}]
The result on $\hat x_\sht(z)$ follows from applying Lemma \ref{wlln}\eqref{wlln_Y}  component-wise. We verify below statement \eqref{item:axx2} for scalar covariate $x_{ws} \in \mathbb R$ to simplify the presentation. 

For the result on $\hqxx$, let $\sigma_{ws} = x_{ws}^2$ to write  $Q_{xx} = (N-1)^{-1}\sumws \sigma_{\ws}$. 
Lemma \ref{wlln} ensures that $\hat Q_{xx}(z) - Q_{xx} = \op$ for all $z\in\mt$ as long as 
Condition \ref{asym0_S}\eqref{asym0} holds for the finite population of $\{Y_{ws}(z) =\sigma_{ws}: z\in\mt\}_{\ws \in\ms }$. 
To verify this, let $\bar\sigma_{w} = M_w^{-1}\sums \sigma_{\ws}$ with 
$
 \bar\sigma^2_{w}  \leq  M_w^{-1} \sums \sigma^2_{ws} = M_w^{-1} \sums x_{ws}^4
$
  by \eqref{cs}.
This ensures
\beginy\label{eq:cs_2}
W^{-1}\sumw \alpha^2_w\bar\sigma^2_{w} 
\leq W^{-1}\sumw \alpha^2_w \left( M_w^{-1} \sums x^4_{ws}\right)  = O(1)
\endy
by Condition \ref{asym2}\eqref{po_x4}. 
Condition \ref{asym0_S}\eqref{asym0}  is thus satisfied with 
\begin{enumerate}[$\cdot$]
\item $\bar \sigma = N^{-1}\sumws \sigma_{\ws} = \lambda_N Q_{xx}$  having a finite limit by Condition \ref{asym2}\eqref{po_xy};
\item $S_{\sigma\sigma}= (W-1)^{-1}\sumw \left(\alpha_w\bar\sigma_{w} -\bar \sigma\right)^2 \leq (W-1)^{-1}\sumw \alpha^2_w\bar\sigma^2_{w} = O(1)$ by \eqref{eq:cs_2};    
\item $
\Psi_{\sigma\sigma} = W^{-1}\sumw M_w^{-1} \{H_w \circ (S_{\sigma\sigma,w}\bbo{4})  \} =  W^{-1}\sumw M_w^{-1}S_{\sigma\sigma,w}H_w = O(1)$
given 
\begina
S_{\sigma\sigma,w}=\frac{1}{M_w-1}\alpha^2_w  \left( \sums  \sigma^2_{ws} -M_w \bar\sigma^2_{w} \right) \leq \frac{1 }{M_w-1} \alpha^2_w \left(\sums  \sigma_{ws}\right)^2 = \frac{M_w^2}{ M_w-1} \alpha_w^2 \bar\sigma^2_{w} 
\enda and \eqref{eq:cs_2}.
\end{enumerate}
This verifies  
$\hqxx - \qxx= \op$. 
The proof for $\hqxy - \qxy= \op$ is almost identical  by verifying Condition \ref{asym0_S}\eqref{asym0} for $\{x_{ws}\Yws(z): z\in\mt\}_{\ws\in\ms}$ and thus omitted. 

The  result on $\hat S_{xx}(z)$ follows by applying Lemma \ref{lln} to the finite population with $Y_{ws}(z) = Y_{ws}(z') = x_{ws}$; the corresponding Condition \ref{asym0_S} is ensured by Condition \ref{asym2}. 
Likewise for the result on $\hat S_{xY(z)}$ to follow from letting $Y_{ws}(z') = x_{ws}$ in  Lemma \ref{lln}. 
\end{proof}

For a set of $J\times 1$ vectors $\gamma = (\gamma_z)_{z\in\mt}$, let $Y_{ws}(z; \bb_z) = \Yws(z) -  x_{ws}^\T \bb_z$ be the adjusted potential outcome based on $\bb_z$,
 and let 
$S_{*,\gamma}$,  $\Sigma_{*,\gamma}$, $\swg$, $\psig$, $\ysg = \{\ysg(z)\}_{z\in\mt}$, $\hat U_{w, \gamma}(z) = \alpha_w \hat Y_{w,\gamma}(z)$, and $\hvsg$
be the analogs of $\bs_*$, $\Sigma_*$, $\bs_w$,   $\Psi$, $\hy_* = \{\hy_*(z)\}_{z\in\mt}$, $\hat U_w(z) = \alpha_w \hat Y_w(z)$, and $\hbV_*$
based on $\{Y_{ws}(z; \bb_z): z\in\mT\}_{\ws\in\ms}$ for $*=\HT, \haj$, respectively. We have $\psig = W^{-1}\sumw M_w^{-1}  H_w \circ \swg$ and $
 \Sigma_{*,\gamma} =  H \circ S_{*,\gamma} +\psig$.
With a slight repetition, let 
\beginy\label{eq:adjusted}
\hys(z; \gamma) = \ysg(z), \qquad \hys(\gamma) = \ysg, \qquad \hat U_{w}(z; \gamma)  = \hat U_{w, \gamma}(z). 
\endy 

\begin{lemma}\label{Cov_x}
Assume split-plot randomization and Conditions \ref{asym}--\ref{asym2}.  For $\hat\bG = (\hbb_z)_{z\in \mT}$ with $\hbb_z = \bb_z+\op$, where $\gamma_z$ is some fixed $J\times 1$ vector, we have
$
\sqrt W (\hY_{*,\hat\bG} - \by) \rightsquigarrow \mN( 0,\Sigma_{*,\bG})$ and $W\hbV_{*,\hat\bG} -  \Sigma_{*,\gamma} = \ssg +\op
$ for $* = \HT, \haj$.
\end{lemma}

\begin{proof}[Proof of Lemma \ref{Cov_x}]
We verify below the result for  $*=\HT$.
The proof for $*=\haj$ is almost identical and  thus omitted. 

First, the $(z,z')$th elements of $\shtg$, $\swg$, and $\shajg$ equal
\beginy\nonumber
\ssht(z,z' ;\bG) &=&  \ssht(z,z') - \bb_z^\T S_{xY(z')} -   S_{Y(z)x}  \bb_{z'} + \bb_z^\T \sxx \bb_{z'},\\ 
\sw(z,z';\bG ) &=&  S_w(z,z') - \bb_z^\T S_{xY(z'),w}   -   S_{Y(z)x,w}  \bb_{z'} + \bb_z^\T \sxxw \bb_{z'},\label{eq:swzz}\\
\sshaj(z,z';\bG) &=&  \sshaj(z,z') - \bb_z^\T S_{xY'(z')} -  S_{Y'(z)x} \bb_{z'} + \bb_z^\T \sxx \bb_{z'},\nonumber
\endy
respectively, by direct algebra. 
Conditions \ref{asym}--\ref{asym2} together imply that 
Condition \ref{asym} holds for the finite population of $\{Y_{ws}(z; \bb_z): z\in\mt\}_{\ws \in \ms}$ with fixed $\bG = (\bb_z)_{ z\in \mT}$.
This ensures
\beginy\label{eq:cov1}
 \sqrt W (\yhtg - \bbY) \rightsquigarrow \mN( 0,\Sigma_{\sht,\bG}), \qquad W\hvhtg - \Sigma_{\sht,\gamma}= \shtg +\op
 \endy
by Theorems \ref{clt}--\ref{covEst}.

\pa{Result on $ \hY_{\sht,\hat\bG} $} By Slutsky's theorem, 
\beginy\label{eq:yy}
\sqrt W \{ \yHTn(z; \hbb_z) -\yHTn(z; \bb_z) \} =  - (\hbb_z-\bb_z)^\T   \sqrt W  \hbx_\sht(z) = \op
\endy
given $\hbb_z - \bb_z = \op$ and the asymptotic Normality of $\sqrt W  \hbx_\sht(z)$ by Theorem \ref{clt}.
This ensures $\sqrt W (\hY_{\sht,\hat\bG} - \hY_{\sht,\gamma}) = \op$. 
The result follows from \eqref{eq:cov1} and Slutsky's theorem.

\pa{Result on $W\hbV_{\sht,\hat\bG}$} By \eqref{eq:cov1}, it suffices to verify $W(\hbV_{\sht,\hat\bG} -  \hV_{\sht,\gamma}) = \op$. 
This is in turn guaranteed by 
\beginy\label{eq:ss}
\hat S_\sht(z, z'; \hat\bG) -  \hat S_\sht(z, z'; \bG)   = \op,
\endy where
\begina
&&\hat S_\sht(z, z'; \hat\bG)=(W_a-1)^{-1} \left\{\sum_{w:A_w=a}  \hU_w(z; \hbb_z)\hU_w(z';\hbb_{z'}) -  W_a \, \yHTn(z; \hbb_z)\yHTn(z'; \hbb_{z'})\right\},\\
&& \hat S_\sht(z, z'; \bG)=(W_a-1)^{-1} \left\{ \sum_{w:A_w=a} \hU_w(z; \bb_z)\hU_w(z';\bb_{z'}) -   W_a \, \yHTn(z; \bb_z)\yHTn(z'; \bb_{z'}) \right\},
\enda
for all $z=(ab)$ and $z'=(ab')$ with the same level of factor A. 
Given $ \yHTn(z; \hbb_z) - \yHTn(z; \bb_z) = \op$ by \eqref{eq:yy}, \eqref{eq:ss} holds as long as the difference between the first terms satisfies
\beginy\label{eq:ssss_2}
&&\quad \quad \Delta = W_a^{-1} \left\{ \sum_{w:A_w=a}  \hU_w(z; \hbb_z)\hU_w(z';\hbb_{z'}) -   \sum_{w:A_w=a} \hU_w(z; \bb_z)\hU_w(z';\bb_{z'})\right\} = \op.
\endy

To this end, let $\Delta_w(z) = (\hbb_z-\bb_z)^\T \hbv_w(z)$ to write $\hU_w(z; \hbb_z) = \hU_w(z; \bb_z)-\Delta_w(z)$. Then
$$
 \hat U_w(z; \hbb_z)\hU_w(z';\hbb_{z'}) 
- \hat U_w(z; \bb_z)\hat U_w(z'; \bb_{z'})
 = \Delta_w(z)\Delta_w(z')-\Delta_w(z)\hat U_w(z'; \bb_{z'}) - \Delta_w(z')\hat U_w(z; \bb_z)
 $$
such that
\begina
\Delta &=& W_a^{-1}   \sum_{w:A_w=a} \left\{\Delta_w(z)\Delta_w(z')-\Delta_w(z)\hat U_w(z'; \bb_{z'}) - \Delta_w(z')\hat U_w(z; \bb_z)\right\}\\
&=&  (\hbb_z-\bb_z)^\T   \left\{ W_a^{-1} \sum_{w:A_w=a}   \hbv_w(z)\hbv^\T_w(z')\right\}  (\hbb_{z'}-\bb_{z'})\\
 && - (\hbb_z-\bb_z)^\T  \left\{  W_a ^{-1} \sum_{w:A_w=a}  \hbv_w(z)\hU_w(z'; \bb_{z'})\right\} -(\hbb_{z'}-\bb_{z'})^\T \left\{ W_a ^{-1} \sum_{w:A_w=a}  \hbv_w(z')\hU_w(z; \bb_z)\right\}\\
&=&  \op;
\enda
the last equality follows from 
$
W_a^{-1} \sum_{w:A_w=a} \hbv_w(z)\hbv^\T_w(z') = O_{\mathbb P}(1)$ and 
\begina
W_a^{-1} \sum_{w:A_w=a} \hbv_w(z)\hU_w(z'; \bb_{z'}) 
=  W_a^{-1}\sum_{w:A_w=a} \hv_w(z) \hU_w(z') -  \left\{W_a^{-1}\sum_{w:A_w=a}\hbv_w(z) \hbv^\T_w(z')\right\}\bb_{z'}=O_{\mathbb P}(1)
\enda by Lemmas \ref{lln}--\ref{Axx}. This verifies \eqref{eq:ssss_2} and hence the result. 
\end{proof}

\begin{lemma}\citep[][Theorem 3.1]{schur}\label{lem:schur}
If $G_1$ and $G_2$ are positive semi-definite, then their Hadamard product $G_1\circ G_2$ is also positive semi-definite. 
\end{lemma}

Lastly, let
$$\Psi_{xx}(z,z') =  W^{-1} \sumw M_w^{-1} H_w(z,z') \sxxw, \quad \Psi_{xY}(z,z') =  W^{-1} \sumw M_w^{-1} H_w(z,z') S_{xY(z'),w}  $$
be  extensions of $\Psi_{xx}(z,z)$ and $\Psi_{xY}(z,z)$ to $z, z'\in\mt$. 
\begin{lemma}\label{lem:eff} Under Conditions \ref{asym}--\ref{asym2} and \eqref{cond:eff}, we have
\begina
\Psi_{xx}(z,z') = o(1), \qquad \Psi_{xY}(z,z') = o(1) \qquad (z, z'\in\mt). 
\enda 
\end{lemma}

\begin{proof}[Proof of  Lemma \ref{lem:eff}]
We verify below the result for scalar covariate $x_{ws} \in \mathbb R$ to simplify the presentation. 
Let $u_0 > 0$ and $l_0 > 0$ be some uniform upper and lower bounds of $|H_w(z,z')|$ for all $z,z'\in\mt$ and $w = \ot{W}$ under Condition \ref{asym}\eqref{paqb}, both independent of $W$. That is, $l_0 \leq |H_w(z,z')|\leq u_0$  for all $z,z'\in\mt$ and $w = \ot{W}$.

First, with $H_w(z,z) > 0$ and $l_0^{-1}H_w(z,z)  \geq 1$,  Condition \eqref{cond:eff} ensures 
\beginy\label{eq:bound}
&& \quad\quad  W^{-1} \sumw M_w^{-1}  \swx   \leq   W^{-1} \sumw M_w^{-1} \{l_0^{-1}H_w(z,z)  \} \swx  =   l_0 ^{-1}\Psi_{xx}(z,z)= o(1).
\endy 
The result for $\Psi_{xx}(z,z')$ then follows from $H_w(z,z') = O(1)$. 
 
Second,  the Cauchy--Schwarz inequality ensures $
|\sxyzw| \leq  \swx ^{1/2} \{S_w(z,z) \}^{1/2}$.
This, together with  $|H_w(z,z')|\leq u_0$, suggests
\begina
|\Psi_{xY}(z,z')|
&\leq&   W^{-1} \sumw M_w^{-1} \left|H_w(z,z')\right| \left|  S_{xY(z'),w}   \right|\\
&\leq& u_0 \, W^{-1} \sumw M_w^{-1}      \swx ^{1/2} \{S_w(z',z')\}^{1/2}\\
&\leq& u_0 \, W^{-1} \sumw \left(M_w^{-1} \swx\right)^{1/2}\left\{M_w^{-1}   S_w(z',z')\right\}^{1/2} \\
&\leq& u_0   \left( W^{-1} \sumw  M_w^{-1}  \swx \right)^{1/2} \left\{ W^{-1} \sumw M_w^{-1}  S_w(z',z')  \right\}^{1/2}\\
&=&o(1);
\enda
the last equality follows from \eqref{eq:bound} and the fact that 
$W^{-1} \sumw M_w^{-1}  S_w(z',z') \leq 
l_0 ^{-1}\Psi(z',z') = O(1)$  by similar reasoning. 
\end{proof}

%

\subsection{Results under the fully-interacted regressions}
We verify in this part the results under the fully-interacted regressions \eqref{dlm4} and \eqref{dlm6}.

\begin{lemma}\label{bb_lin}
{\prex}  
$\bb_{\wls,z} = \qxx^{-1} \qxyz$ and $\bb_{\ag ,z} = \Uxx  ^{-1}(z)\uxyz$ for $z\in\mt$.  
\end{lemma}

\begin{proof}[Proof of Proposition \ref{prop:ols_x}  and Lemma \ref{bb_lin}]
We verify below the results for the unit and aggregate regressions, respectively. 

\pa{Unit regression}
For $\dg = \ols, \wls$, the inclusion of full interactions ensures that  $\tilde\beta_{\dg,\lin}(z)$ and $\tilde\gamma_{\dg,z}$ from \eqref{dlm4} 
equal the coefficients of $1$ and $x_{ws}$ from the treatment-specific regression 
\beginy\label{lm:unit_z}
Y_{ws} \sim 1+ x_{ws}  \qquad \text{over} \ \ \ws  \in \uz, 
\endy
respectively, under their respective fitting schemes.
Let $ Y_z$ and $ X_z$ be the concatenations of $Y_{ws}$ and $ x_{ws}$ over $\ws  \in \uz$, respectively.
The matrix form of \eqref{lm:unit_z} equals $Y_z \sim 1_{N_z} + X_z$ for $z \in \mt$. 

That $\tbeta_\ol(z) = \hy_\sm(z) - \hx_\sm^\T(z)\tg_{\ols,z}$  follows from standard results.

Let $\Pi_z = \diag(\pi_{ws})_{\ws \in\uz}$ be the weighting matrix under fitting scheme ``wls''.
The first-order condition of \wlss ensures
$
 G_1  (\tbeta_\wl(z), \tgwz^\T )^\T = G_2$, where  
\begina
G_1 &=& ( 1_{N_z},  X_z)^\T \bp_z ( 1_{N_z},  X_z) =  \left(
\begin{array}{cc}
 1_{N_z}^\T \bp_z  1_{N_z} &  1_{N_z}^\T\bp_z  X_z\\
 X_z^\T \bp_z  1_{N_z} & X_z^\T \bp_z  X_z
\end{array}
\right) 
= \left(
\begin{array}{cc}
\oneHT &\hbx^\T_\sht(z)\\
\hbx_\sht(z) & \lambda_N\hqxx
\end{array}
\right),
\\
G_2 &=& ( 1_{N_z},  X_z)^\T \bp_z  Y_z 
=
\left(
\begin{array}{c}
 1_{N_z}^\T \bp_z  Y_z \\
 X_z^\T \bp_z  Y_z
\end{array}
\right)
= 
\left(
\begin{array}{c}
\hY_\sht(z) \\
\lambda_N \hqxy
\end{array}
\right)
\enda
by \eqref{eq:ht_pi} and $1_{N_z}^\T\bp_z  X_z = \hbx^\T_\sht(z)$. 
Compare the first row  to see
$\oneHT \tbeta_\wl(z) +\hbx^\T_\sht(z) \tgwz = \hY_\sht(z)$ and hence the numeric result for $\tbeta_\wl(z)$. 
The  probability limit then follows from
$(\tbeta_\wl(z), \tgwz^\T )^\T = G_1 ^{-1}G_2$ with
$G_1 = \diag(1, Q_{xx}) + \op$ and $G_2 = (\by(z), Q_{xY(z)}^\T)^\T + \op$ by  Lemmas \ref{wlln} and \ref{Axx}.

\pa{Aggregate regression}
The inclusion of full interactions ensures that $\tilde\beta_{\ag,\lin}(z)$ and $\tguz$ from \eqref{dlm6} equal the \textsc{ols} coefficients of $1$ and $\hat v_w(z)$ from the treatment-specific regression
\beginy\label{lm:ag_z}
\hat U_w (z) \sim 1+ \hat v_w(z)  \qquad\ \text{over} \ \  w \in \Wz,
\endy
respectively. Let $ U_z$ be the vectorization of the $W_a $ observations, namely $\{\hat U_w(z): w \in \mathcal{W}(z)\}$, under treatment $z = (ab)$, and  let  $ \lz$  be the concatenation of the corresponding $\hbv_w(z)$'s.
The matrix form of \eqref{lm:ag_z}  equals $U_z \sim 1_{W_a}+ \lz$ for $z\in\mt$. 
The first-order condition of \olss ensures 
$G_1 (\tbeta_\tl(z),  \tgzu^\T )^\T = G_2$, where
\begina
G_1 &=& W_a^{-1}( 1_{W_a},  \lz )^\T  ( 1_{W_a},  \lz ) 
=
W_a^{-1}\left(
\begin{array}{cc}
 1^\T_{W_a}   1_{W_a} &  1^\T_{W_a}  \lz \\
 \lz ^\T   1_{W_a} & \lz ^\T   \lz
\end{array}
\right) =
 \left(
\begin{array}{cc}
 1& \hbx_\sht^\T(z)\\
 \hbx_\sht(z) &\hsxx
\end{array}
\right)
,\\
G_2 &=&   W_a^{-1}( 1_{W_a},  \lz )^\T   U_z =  W_a^{-1}\left(
\begin{array}{c}
 1_{W_a}^\T   U_z \\
 \lz ^\T   U_z
\end{array}
\right)
 = 
\left(
\begin{array}{c}
\hY_\sht(z) \\
\hsxy
\end{array}
\right). 
\enda
Compare the first row to see
$
\tbeta_\tl (z)+ \hbx_\sht^\T(z)  \tgzu =  \yHT$ and hence the numeric result. 
The  probability limit then  follows from
$ (\tbeta_\tl(z),  \tgzu^\T)^\T = G_1 ^{-1}G_2$ with
$G_1 =\diag\{1, T_{xx}(z)\} + \op$ and $G_2 = (\by(z), T_{xY(z)}^\T)^\T + \op$ by  Lemmas \ref{wlln} and \ref{Axx}. 
\end{proof}

We next verify the asymptotic Normality of $\tbeta_\dl$ and the asymptotic conservativeness of $\tv_\dl$ for $\dg = \wls, \ag$ in Theorem \ref{thm_lin}, respectively.

\begin{proof}[Proof of Theorem \ref{thm_lin}, Part I for $\tbeta_\dl$]
For $\dg = \wls, \ag$, let $ \tilde\gamma_\dl = (\tilde\gamma_{\dg,z})_{z\in\mt}$   with $\tilde\gamma_{\dg,z} = \gamma_{\dg,z} + \op$ by Lemma \ref{bb_lin}.  
Recall the definitions of $\hys(z; \gamma)$ and $\hys(\gamma)$ from \eqref{eq:adjusted}. 
Proposition \ref{prop:ols_x} ensures
\beginy\label{eq:tbeta_wls}
\tbeta_\wl(z) = \yhajn(z; \tg_{\wls,z}), \qquad \tbeta_\wl = \yhajn(\tg_\wl), \qquad \tbeta_\tl = \yhtn(\tg_\tl), 
\endy 
respectively. 
The results on $\tbeta_\dl \ (\dg = \wls, \ag)$ then follow from Lemma \ref{Cov_x}. 
\end{proof}

\begin{proof}[Proof of Theorem \ref{thm_lin}, Part II for $\tv_\dl$]
We verify below the result for $\tilde V_\wl$ from the unit regression \eqref{dlm4}. The proof for $\tilde V_\tl$ from \eqref{dlm6} is almost identical and  thus omitted.

Let $ \chi$   be the concatenation  of $\chi_{ws} =  d_{ws} \otimes  x_{ws} $
over $\ws \in\mathcal{S}$. 
The design matrix of \eqref{dlm4}  equals
$C_\lin = (  D,  \chi)$.
Let $ C_{\lin,w} = ( D_w,  \chi_w)$ be the sub-matrix of $C_\lin$ corresponding to whole-plot $w$. 
Let $\tep_{w} = (\tep_{w1}, \dots, \tep_{wM_w})^\T$, where 
\beginy\label{eq:epws}
\tep_{ws} =  Y_{ws} - \tbeta_\wl(Z_{ws}) -  x^\T_\ws \tg_{\wls,Z_{ws}}
\endy is the residual from the \textsc{wls} fit of \eqref{dlm4}. 
Then $W \tbV_\wl $ equals the upper-left $4\times 4$ matrix of 
\beginy\label{eq:tv}
 ( C_\lin^\T\bp C_\lin)^{-1} \left(W \sumw  C_{\lin,w}^\T \bp_w \tep_w  \tep_{ w}^\T\bp_w  C_{\lin,w} \right) ( C_\lin^\T\bp  C_\lin)^{-1}.
\endy
Let
\begina
\too_\wl  = ( \dt \bp D)^{-1}\left(\sumw  D_w^\T \bp_w\tep_w \tep_w ^\T \bp_w  D_w\right)( \dt \bp D)^{-1} 
\enda
be an intermediate quantity. The result on $\tbV_\wl $ holds as long as 
\beginy\label{ehw_lin:1} 
\quad \quad \textup{(i)} \ \ W \too_\wl -  \Sigma_\wl = S_\wl+\op \quad \text{and} \quad 
\textup{(ii)} \ \ W(\tbV_\wl - \too_\wl) = \op.
 \endy
We verify below these two conditions one by one.

\pa{Condition \eqref{ehw_lin:1}(i)} 
Direct comparison shows that $\too_\wl$ is an analog of $\tbV_\wls$ from \eqref{v_pi},  with the unadjusted $\te_{\wls,w} = (\te_{\wls, ws})_{s=1}^{M_w}$ replaced by $\tep_{w} = (\tep_{ws})_{s=1}^{M_w}$.
By \eqref{eq:tbeta_wls} and \eqref{eq:epws},  
\begina
\tep_{ws} =  
(Y_{ws} -  x^\T_\ws \tg_{\wls,Z_{ws}}) - \yhajn(Z_{ws}; \tg_{\wls,Z_{ws}})
\enda
 is essentially the analog of $\te_{\wls,ws} = Y_{ws} - \yhajn(Z_{ws})$ defined on the adjusted potential outcomes $Y_{ws}(z; \tg_{\wls,z}) = Y_{ws}(z)-  x^\T_\ws \tg_{\wls,z} $. 
This, together with Theorem \ref{thm_vHats}, ensures
\beginy\label{eq:omega_v}
\too_\wl 
= \hat{ 1}^{-1}_\sht\diag\left(  \frac{W_0-1}{W_0} I_2, \frac{W_1-1}{W_1} I_2\right) \hbV_\shaj (\tilde\bG_\wl) \, \hat{ 1}^{-1}_\sht, 
 \endy
where $\hbV_\shaj (\tilde\bG_\wl)$ denotes the value of $  \hV_{\shaj,\gamma}$ at $\gamma = \tg_\wl$. 
With $\tg_{\wls,z}-\bb_{\wls,z} = \op$ by Lemma \ref{bb_lin}, Lemma \ref{Cov_x} ensures 
$ W \hbV_\shaj (\tilde\bG_\wl)  - \Sigma_ \wl   = S_\wl + \op$
by the definitions of $ S_\wl$ and $\Sigma_\wl  $. 
This, together with \eqref{eq:omega_v}, ensures condition \eqref{ehw_lin:1}(i).

\pa{Condition \eqref{ehw_lin:1}(ii)}
Let 
 $G_1 = W \sumw \dt _w \bp_w \tep_w \tep_w^\T\bp_w   \chi_w$ and $G_2 = W \sumw \chi^\T_w \bp_w \tep_w \tep_w^\T\bp_w   \chi_w$. 
The  ``meat'' and ``bread'' parts of the sandwich covariance \eqref{eq:tv} satisfy
\beginy\label{eq:ehw_lin}
W \sumw  C_{\lin,w}^\T \bp_w \tep_w  \tep_{ w}^\T\bp_w  C_{\lin,w}
&=& 
W \sumw\left( \begin{array}{c} \dt _w\\ \chi^\T_w\end{array}\right) \bp_w \tep_w \tep_w^\T\bp_w  
( D_w,  \chi_w)\nonumber\\
&=& 
W  
\beginp
\sumw \dt _w \bp_w \tep_w \tep_w^\T\bp_w   D_w&\ \ \sumw \dt _w \bp_w \tep_w \tep_w^\T\bp_w   \chi_w\\ \sumw \chi^\T_w \bp_w \tep_w \tep_w^\T\bp_w   D_w &\ \  \sumw \chi^\T_w \bp_w \tep_w \tep_w^\T\bp_w   \chi_w\endp\nonumber\\
&=& 
\beginp
( \dt \bp D) (W\too_\wl )  ( \dt \bp D) & \ \ G_1 \\ G_1^\T &\ \ G_2 \endp,\\
 C_\lin^\T\bp  C_\lin = \left( \begin{array}{c}  \dt  \\  \chi^\T \end{array} \right)\bp ( D,  \chi)
&=&
 \left( 
 \begin{array}{cc}  \dt \bp  D & \dt \bp \chi \\
\chi^\T\bp   D & \chi^\T\bp \chi
\end{array}
\right)=  \diag( I_\tss ,  I_\tss\otimes \qxx) + \op, \nonumber
\endy 
respectively. The last equality follows from 
$ \dt \bp D =  I_\tss + \op$ by \eqref{matrix0} and Lemma \ref{wlln}, and 
\begina
&&\dt \bp \chi  
= \sumws\pi_{ws}  d_{ws} \chi_{ws}^\T = \diag\left(\sumwszz \pi_{ws} x_{ws}^\T \right)_{z\in\mt} =  \diag\{\hx_\sht^\T(z)\}_{z\in\mt}= \op, \\ &&\chi^\T\bp \chi = \sumws \pi_{ws}  \chi_{ws}  \chi_{ws}^\T = \diag \left(\sumwszz \pi_{ws}  x_{ws} x_{ws}^\T \right)_{z\in\mt} =  I_\tss  \otimes \qxx + \op
\enda
by 
Lemma \ref{Axx}.
This, together with \eqref{eq:ehw_lin}, ensures that  \eqref{ehw_lin:1}(ii) holds as long as 
\begina
G_k = \big(G_k(z,z') \big)_{z,z'\in\mt} = O_\pr(1) \qquad \text{for} \ \ k = 1, 2. 
\enda
We verify below $G_2= O_\pr(1)$ for scalar covariate $x_{ws} \in \mathbb R$ for notational simplicity.
The proof for $G_1 = O_\pr(1)$ is almost identical and thus omitted. 

First, recall the expression of $\pi_{ws}$ from \eqref{eq:piws}. 
Direct algebra shows that 
\begina
\chi_w^\T \Pi_w \tep_w  
&=& \sums \pi_{ws}\chi_{ws} \tep_{ws} =\sums \pi_{ws}  (d_{ws}\otimes x_{ws}) (1\otimes \tep_{ws}) 
=\sums \pi_{ws}   d_{ws}  \otimes(x_{ws} \tep_{ws})  \\
&=&  \alpha_w\left( W_0^{-1}  \kappa_w (00), \,  W_0^{-1} \kappa_w (01), \,  
 W_1^{-1}   \kappa_w (10), \, W_1^{-1}  \kappa_w (11) \right)^\T,
\enda
where  $\kappa_w (z) = M_{wb}^{-1}\sumsz x_{ws}\tep_{ws}$ with $\kappa_w(z)=0$ if $\znotin$.  
This ensures \begina
G_2(z,z') =
\left\{ \begin{array}{cl}
p_a^{-2} W^{-1} \sum_{w: A_w=a} \alpha_w^2 \kappa_w (z)\kappa_w (z') & \ \  \text{if $z=(ab)$ and $z'=(a'b')$ with $a = a'$,} \\
0  & \ \  \text{if $z=(ab)$ and $z'=(a'b')$ with $a\neq a'$},
 \end{array}\right.
  \enda 
with
\begina
p_a^2 |G_2(z,z')|  \leq  \frac{1}{W} \sumw \alpha_w^2 |\kappa_w (z)\kappa_w (z')|  
 \leq \frac{1}{2W} \sumw \alpha_w^2  \{\kappa_w (z) \}^2+ \frac{1}{2W} \sumw \alpha_w^2\{\kappa_w (z') \}^2
\enda
for $z = (ab)$ and $z'=(ab')$. It thus suffices to verify $ 
 W^{-1} \sumw \alpha_w^2  \{\kappa_w (z) \}^2 
  = O_\pr(1)$.

To this end, 
let $\oxw{k} = M_{w}^{-1} \sums  x_{ws}^k$ and $\oew{k} = M_{w}^{-1} \sums  \epsilon_{ws}^k$ for $k=2,4$. 
Then
\begina
|\kappa_w (z)|\leq M_{wb}^{-1}\sums  |x_{ws}\tep_{ws}|  
 \leq   2^{-1} M_{wb}^{-1}\sums (x_{ws}^2 + \tep_{ws}^2)
=  2^{-1}q_{wb}^{-1} \left( \oxw{2}  +   \oew{2} \right)
\enda
such that, with $q_{wb}\geq \epsilon$ by Condition \ref{asym}, 
\begina
 \{\kappa_w (z) \}^2 
\leq   4^{-1}\epsilon^{-2}\left( \oxw{2}  +   \oew{2}\right)^2
\leq  2^{-1}\epsilon^{-2}  \left\{\left( \oxw{2}\right)^2   +   \left ( \oew{2}\right)^2 \right\}
\leq  2^{-1}\epsilon^{-2}  \left( \oxw{4}  +   \oew{4}  \right),
\enda
where the  last inequality follows from  $ \left( \oxw{2}\right)^2 \leq \oxw{4}  $ and $ \left ( \oew{2}\right)^2 \leq \oew{4}$ by \eqref{cs}. 
This ensures 
$$
 W^{-1} \sumw \alpha_w^2  \{\kappa_w (z) \}^2 
 \leq  2^{-1}\epsilon^{-2}   W^{-1} \sumw \alpha_w^2 \left( \oxw{4}   +  \oew{4} \right) = O_\pr(1);
$$
the last equality follows from $W^{-1} \sumw \alpha_w^2 \oxw{4} = O_\pr(1) $ by Condition \ref{asym2}\eqref{po_x4} and 
$ W^{-1} \sumw \alpha_w^2 \oew{4}   = O_\pr(1)$
by 
$
\tep_{ws}^4 \leq 27   Y_{ws} ^4 +27 \{\tbeta_\wl(Z_{ws})\}^4 +27 (x_{ws}^\T\tg_{\wls,Z_{ws}})^4$ from \eqref{eq:epws} and \eqref{cs}.

\end{proof}

\subsection{Results under the additive regressions}

\begin{lemma}\label{bb_fisher}
{\prex} 
$$
\bb_\wls =  { |\mT|^{-1}}  \qxx^{-1} \sumz \qxyz,\qquad \bb_\ag  =  \left\{\sum_{z=(ab)\in\mt} p_a  \Uxx  (z)\right\}^{-1}\left\{\sum_{z=(ab)\in\mt} p_a   \uxyz \right\}.
$$  
Further assume Condition \ref{balanced}. Then $\tilde\gamma_\ols = \gamma_\ols+ \op$ with $\gamma_\ols = Q_{xx}^{-1}\sumz r_z Q_{xY(z)}$, recalling $r_z = N_z/N$. 
\end{lemma}


\begin{proof}[Proof of Proposition \ref{prop:ols_x_fisher} and Lemma \ref{bb_fisher}] 
We verify below the numeric results in Proposition \ref{prop:ols_x_fisher} and the asymptotic results in Lemma \ref{bb_fisher} together. Let $ X$ and $ \Lambda$ be the concatenations of $\{x_{ws}:\ws \in \mathcal{S}\}$ and $\{\hvwab: w= \ot{W}; \ b = 0,1\}$, respectively.
The design matrices of  \eqref{dlm4_f} and \eqref{dlm6_f} equal 
$C_\fisher = (  D,  X)$ and $ C_{\ag ,\fisher} = ( D_\ag ,  \Lambda)$, respectively. 
Recall that $ R =  \diag(r_z)_{z\in\mt}$ with $r_z =N_z/N$ and $\hat x_* =  ( \hbx_*(00), \hbx_*(01), \hbx_*(10), \hbx_*(11) )^\T$ for $* = \sm, \HT, \haj$. 
Direct algebra shows that
\beginy\label{matrix0_x}
N^{-1}  \dt   X =  R  \hx_\snm, \qquad \dt \bp  X = \hx_\sht, \qquad  W^{-1}  D_\ag ^\T  U = \dpp   \hx_\sht
\endy
analogous to \eqref{matrix0}. 
Lemma \ref{Axx} further ensures
\beginy\label{matrix0_xy}
X^\T\bp X  = 
|\mt| \, \qxx + \op,&\quad &
X^\T\bp Y = 
\sumz \qxyz + \op,\\
W^{-1}  \Lambda^\T \Lambda = 
\sum_{z=(ab)\in\mt} p_a  \Uxx  (z) + \op, &\quad&
W^{-1} \Lambda^\T  U
= \sum_{z=(ab)\in\mt} p_a  \uxyz  + \op. \nonumber
\endy  

\pa{Results on $(\blsxf, \tg_\ols)$}  
The first-order condition of \olss ensures $
G_1 (\tbeta_{\ols,\fisher}^\T, \tg_\ols^\T)^\T  = G_2$,  
where 
\begina 
G_1 =  C_\fisher^\T  C_\fisher = 
 \left( 
 \begin{array}{cc}  \dt   D & \dt  X \\
 X^\T   D &  X^\T X
\end{array}
\right) =\left( 
 \begin{array}{cc} N R & N R    \hbx_\snm \\
N \hbx_\snm^\T R      &   X^\T X
\end{array}
\right),
\quad 
 G_2  =C_\fisher^\T  Y =  \left( \begin{array}{c}  \dt  Y \\  X^\T Y \end{array} \right)
=  \left( \begin{array}{c} N R\hbY_\snm\\  X^\T Y \end{array} \right).
\enda
by  \eqref{matrix0} and \eqref{matrix0_x}. 
The numeric result follows by comparing the first row.
The probability limit follows  from 
$(\tbeta_\of^\T, \tg_\ols^\T)^\T =  (N^{-1}G_1)^{-1} (N^{-1} G_2)$, where $N^{-1}G_1 = \diag(R,  \qxx)+\op$ and $N^{-1}G_2 =( (R\by)^\T, \sumz r_z\qxyz^\T)^\T + \op$  under Condition \ref{balanced} by 
$\hat x_\sm =\hat x_\sht= \op$  and 
\begina
&& N^{-1}X^\T  X = N^{-1}\sumws    x_{ws}  x_{ws}^\T = \sumz r_z \left(N_z^{-1} \sumwszz   x_{ws} x_{ws}^\T\right) =\qxx + \op, \\
&&N^{-1}X^\T Y = N^{-1}\sumws    x_{ws}  Y_{ws} = \sumz r_z \left\{ N_z^{-1} \sumwszz   x_{ws} Y_{ws}(z) \right\} = \sumz r_z \qxyz + \op
\enda
from Lemma \ref{Axx}.

\pa{Results on $(\blsxHTf, \tg_\wls)$} 
The first-order condition of \wlss ensures $
 G_1 (\tbeta_\wf^\T, \tg_\wls^\T)^\T  =  G_2$, 
where 
\begina
G_1 = C_\fisher^\T\bp  C_\fisher 
 =  \left( 
 \begin{array}{cc}  \dt \Pi  D & \dt \Pi X \\
 X^\T \Pi  D &  X^\T\Pi X
\end{array}
\right) 
=
 \left( 
 \begin{array}{cc} \hat 1_\sht &  \hx_\sht \\
  \hx_\sht^\T &  X^\T\bp X
\end{array}
\right), \quad 
G_2 = C_\fisher^\T\bp  Y 
 =
  \left( \begin{array}{c}  \dt \Pi  Y \\  X^\T \Pi Y \end{array} \right)
 =  \left( 
 \begin{array}{c} \byht \\
 X^\T\bp   Y
\end{array}
\right) 
\enda
by  \eqref{matrix0} and \eqref{matrix0_x}. 
The numeric result follows by comparing the first row. 
The probability limit follows from  $(\tbeta_\wf^\T, \tg_\wls^\T)^\T =  G_1^{-1}  G_2$, where $G_1 = \diag(I_{|\mt|}, |\mt| \, Q_{xx})+\op$ and $G_2 = (\by^\T, \sumz \qxyz^\T)^\T + \op$ by $\hx_\sht = \op$, $\byht=\by+\op$, and \eqref{matrix0_xy}.  
 
\pa{Results on $(\blsxuf, \tg_\ag)$} The first-order condition of \olss ensures $
G_1 (\tbeta_{\ag ,\fisher}^\T, \tg_\ag^\T )^\T  =  G_2$, 
where 
\begina
G_1 = W^{-1} C_{\ag ,\fisher}^\T   C_{\ag ,\fisher} 
=
 \left( 
 \begin{array}{cc}   P  & P\hx_\sht \\
(P\hx_\sht)^\T  &  \Lambda^\T \Lambda
\end{array}
\right),
 \qquad 
G_2 = W^{-1}  C_{\ag ,\fisher}^\T   U = \left( 
 \begin{array}{c}  P\byht \\
  \Lambda^\T  U
\end{array}
\right) 
\enda
by \eqref{matrix0} and \eqref{matrix0_x}. 
The numeric result follows by comparing the first row. The probability limit follows from $(\tbeta_{\ag ,\fisher}^\T, \tg_\ag^\T )^\T  = G_1^{-1} G_2$, where 
\begina
G_1 = \beginp
P & \\
&  \sum_{z=(ab)\in\mt} p_a  \Uxx  (z) 
\endp +\op, \qquad 
G_2  = \beginp
P\by\\
\sum_{z=(ab)\in\mt} p_a  \uxyz\}
\endp + \op
\enda by $\hx_\sht = \op$, $\byht=\by+\op$, and \eqref{matrix0_xy}. 
\end{proof}

\begin{proof}[Proof of Theorem \ref{thm_fisher}]
The proof is similar to that of Theorem \ref{thm_lin} and omitted. 
\end{proof}

\subsection{Guaranteed gains in asymptotic efficiency}\label{sec:eff_app}
Let $\hat c_w(A_wb) = (\alpha_w-1,  \hat v_w^\T(A_wb) )^\T$ be the augmented whole-plot level covariate vector corresponding to $\tilde\beta_{\ag,\dmd}(\alpha,v) \ (\dmd = \fisher, \lin)$.
Let $\tilde\gamma_{\ag}(\alpha,v)$ and $\tilde\gamma_{\ag,z}(\alpha,v)$ be the analogs of $\tilde\gamma_{\ag}$ and $\tilde\gamma_{\ag,z}$, respectively, based on $\hat c_w(A_wb)$. 
The underlying augmented unit-level covariate vector equals $(1-\alpha_w^{-1}, x_{ws}^\T)^\T$ for unit $ws$, and satisfies Condition \ref{asym2} and \eqref{cond:eff} as long as $(x_{ws})_{ws\in\ms}$ satisfies Condition \ref{asym2} and \eqref{cond:eff}. 
All results on $\{\tilde\beta_{\ag,\dmd}, \tilde\gamma_{\ag}, \tilde\gamma_{\ag,z}\}$ so far thus extend to   $\{\tilde\beta_{\ag,\dmd}(\alpha,v), \tilde\gamma_{\ag}(\alpha,v), \tilde\gamma_{\ag,z}(\alpha,v)\}$ as well. 

Recall $\gamma_\dg $ and $\gamma_{\dg ,z}$ as the probability limits of $\tilde \gamma_\dg $ and $\tilde \gamma_{\dg ,z}$, respectively, for $\dg = \wls, \ag$.
Let
$\Psi_{\dg,\fisher}$ and $\Psi_{\dg,\lin}$, 
be the value of $\Psi_\gamma$ when $\gamma_z = \gamma_\dg,  \gamma_{\dg,z}$, respectively, 
Let 
$\gamma_{\ag }(\alpha,v)$ and $\gamma_{\ag ,z}(\alpha,v)$ be  the probability limits of $\tilde\gamma_{\ag}(\alpha,v)$ and $\tilde\gamma_{\ag,z}(\alpha,v)$, respectively, with $\{S_{\ag,\dmd} (\alpha,v), \Psi_{\ag,\dmd}(\alpha,v), \Sigma_{\ag,\dmd}(\alpha,v)\}$ as the corresponding analogs of $\{S_{\ag,\dmd}, \Psi_{\ag,\dmd}, \Sigma_{\ag,\dmd}\}$ for $\dmd = \fisher, \lin$. They are essentially the analogs of $\{\ssht, \Psi, \Sigma\}$ defined on the  
 adjusted potential outcomes $\Yws(z) - {(1-\alpha_w^{-1},x_{ws}^\T)} \gamma_{\ag }(\alpha,v)$ and $\Yws(z) - {(1-\alpha_w^{-1},x_{ws}^\T)} \gamma_{\ag,z}(\alpha,v)$, respectively. 

\begin{proof}[Proof of Proposition \ref{prop:eff}]
For two sequences of  symmetric matrices $(A_N)_{N=1}^\infty$ and $(B_N)_{N=1}^\infty$, write $A_N \leq_\infty B_N$ if the limiting value of $(B_N-A_N)$ is positive semi-definite. 
The result is equivalent to 
\begina
\begin{array}{ll}
\Sigma_{\tl}(\alpha,v) \leq_\infty \Sigma_\sht, \,  
\Sigma_\shaj,  \,   \Sigma_{\wls,\dmd}, \,
\Sigma_{\ag,\dmd}, \,  \Sigma_{\ag ,\fisher}(\alpha,v) \qquad \text{for} \ \ {\dmd  = \fisher, \lin}
\end{array}
\enda
with $\Sigma_\sht = H\circ \ssht +\Psi, \ \Sigma_\shaj = H\circ \sshaj +\Psi, \ \Sigma_{\wls,\dmd} = H\circ S_{\wls,\dmd}  +\Psi_{\wls,\dmd}$, and 
\begina
\Sigma_{\ag,\dmd} = H\circ S_{\ag,\dmd}  +\Psi_{\ag,\dmd}, \qquad 
\Sigma_{\ag,\dmd}(\alpha,v) = H\circ S_{\ag,\dmd}(\alpha,v) +\Psi_{\ag,\dmd}(\alpha,v).
\enda
Because $ H  \geq  0$, Lemma \ref{lem:schur}
ensures that it suffices to verify
\begin{enumerate}[(i)]
\item\label{sd:2} $\Psi_{\wls,\dmd} - \Psi = o(1), \ \Psi_{\ag,\dmd} - \Psi = o(1), \ \Psi_{\ag,\dmd}(\alpha,v) - \Psi = o(1)$; 
\item\label{sd:1} $S_{\tl}(\alpha,v) \leq_\infty S_\sht, \, S_\shaj, \, S_{\wls,\dmd}, \,  S_{\ag,\dmd}, \,  S_{\tf}(\alpha,v)$
\end{enumerate}
 for $\dmd  = \fisher, \lin$ under Conditions \ref{asym}--\ref{asym2} and \eqref{cond:eff}. 
We verify below \eqref{sd:2} and \eqref{sd:1} respectively.

\pa{Sufficient condition \eqref{sd:2}} 
Let $\Psi(z,z'; \bG)$ be the $(z,z')$th element of $\Psi_\gamma$. For arbitrary fixed $\gamma = (\gamma_z)_{z\in\mt}$, it follows  from \eqref{eq:swzz} and Lemma \ref{lem:eff} that 
\begina
\Psi(z,z'; \bG)
&=& W^{-1}\sum_{w=1}^W M_w^{-1} H_w(z,z')  \sw(z,z';\bG )\\
&=&  W^{-1}\sum_{w=1}^W M_w^{-1} H_w(z,z')  \{ S_w(z,z') - \bb_z^\T S_{xY(z'),w}   -    \bb_{z'}^\T \sxyzw  + \bb_z^\T \sxxw \bb_{z'}  \}\\
&=&  \Psi(z,z') - \bb_z^\T\Psi_{xY}(z,z') -  \bb_{z'}^\T\Psi_{xY}(z',z) +  \bb_z^\T \Psi_{xx}(z,z') \bb_{z'}\\
&=&  \Psi(z,z')  + o(1) 
\enda
 such that $\Psi_\gamma = \Psi + o(1)$. 
This verifies the result for $\Psi_{\dg,\dmd} \ (\dg = \wls, \ag; \ \dmd = \fisher, \lin)$.

The result for $\Psi_{\ag, \dmd}(\alpha,v)$ then follows from the same line of reasoning as that for $\Psi_{\ag, \dmd}$ given Condition \eqref{cond:eff} also holds for the augmented unit-level covariates $(1-\alpha_w^{-1}, x_{ws})$. 

\pa{Sufficient condition \eqref{sd:1}} Let 
\begina
\begin{array}{ll}
e_{1,w}(z)= U_w(z) - \bar Y(z), &\quad e_{2,w}(z) =U_w(z) - \alpha_w \bar Y(z),\\
e_{3,w}(z) =U_w(z) - \bar Y(z) -  v_w^\T \theta_z, & \quad e_{4,w}(z) = U_w(z) - \alpha_w\bar Y(z) - v_w^\T \theta_z,\\
e_{5,w}(z) = U_w(z) - \bar Y(z) -  c_w^\T \theta_z', 
\end{array}
\enda
where $v_w = \alpha_w \bar x_w$ and $c_w = (\alpha_w - 1, v_w^\T)^\T$ are the population analogs of $\hvwab = \alpha_w \hat x_w(A_wb)$ and $\hat c_w(A_wb)$, respectively, and $\theta_z \in \mathbb R^J$ and $\theta_z' \in \mathbb R^{J+1}$ are arbitrary  vectors.
Let $\gamma_{c,z}$ be  the coefficient of $c_w$  from the \olss fit of $U_w(z)$ on $(1, c_w)$ over $w= \ot{W}$ with 
\begina
e_{6,w}(z) = U_w(z) - \bar Y(z) -  c_w^\T \gamma_{c,z} 
\enda as the corresponding residual. 
Let $S_k =  ( S_k(z,z')  )_{z, z'\in\mt}$ be the finite-population covariance matrix of $\{e_{k,w}(z): z\in\mt\}_{w=1}^W$ for $k = \ot{6}$ with $
S_k(z,z') = (W-1)^{-1}\sumw e_{k,w}(z) e_{k,w}(z') = (W-1)^{-1} \{e_k(z)\}^\T e_k(z)$, where  $e_k(z) = ( e_{k,1}(z), \dots, e_{k,W}(z) )^\T$.  

Note that $e_{k, w}(z)-e_{6, w}(z)$ is a linear combination of $c_w$ for all $k =\ot{5}$. 
The theory of least squares ensures  $\{e_k(z) - e_6(z)\}^\T e_6(z') = 0$ for all $z, z'\in\mt$ and $k =\ot{5}$.
Let  $e_{k-6}(z) = e_k(z) - e_6(z)$ be a shorthand for $k = \ot{5}$. 
Then $e_k(z) =  e_6(z) + e_{k-6}(z)$ with 
\begina
S_k(z,z') &=&  (W-1)^{-1}  \{ e_6(z) + e_{k-6}(z) \}^\T  \{e_6(z') + e_{k-6}(z')\}\\
&=&  S_6(z,z') + (W-1)^{-1} \{e_{k-6}(z)\}^\T\{e_{k-6}(z')\},\\
S_k - S_6 &=&   (W-1)^{-1}\left( \begin{array}{c}e_{k-6}(00)\\e_{k-6}(01)\\e_{k-6}(10)\\e_{k-6}(11)\end{array}\right)\big(  e_{k-6}(00), e_{k-6}(01), e_{k-6}(10), e_{k-6}(11)\big) \geq 0. 
\enda
This, together with 
\begini
\itemc $S_1 = \ssht$; $S_2 = \sshaj$; 
\itemc $S_3 = S_{\ag,\fisher}$ and $S_{\ag,\lin}$ for $\theta_z = \gamma_\ag $ and $\gamma_{\ag ,z}$, respectively; 
\itemc $S_4 = S_{\wls,\fisher}$ and $S_{\wls,\lin}$ for $\theta_z =\gamma_\wls$ and $\gamma_{\wls,z}$, respectively;  
\itemc$S_5 = S_{\ag,\fisher}(\alpha,v)$ and $S_{\ag,\lin}(\alpha,v)$ for $\theta_z' = \gamma_\ag (\alpha,v)$ and $\gamma_{\ag ,z}(\alpha,v)$, respectively,
\endi
ensures
\beginy\label{eq:sd_S}
S_6 \leq   \ssht, \     \sshaj, \  S_{\wls,\dmd},\   S_{\ag,\dmd} , \   S_{\ag,\dmd}(\alpha,v) \qquad \text{for} \ \ \dmd  = \fisher, \lin.
\endy

In addition,
let $S_{xx}(\alpha,v)$, $S_{xY(z)}(\alpha,v)$, $\Psi_{xx}(z,z; \alpha,v)$, and $\Psi_{xY}(z,z; \alpha,v)$ be the analogs of $S_{xx}$, $S_{xY(z)}$, $\Psi_{xx}(z,z)$, and $\Psi_{xY}(z,z)$ under the augmented covariates $(1-\alpha_w^{-1},x_{ws}^\T)$, respectively. 
Then
\begina
S_{xx}(\alpha,v) = (W-1)^{-1}\sumw c_w c_w^\T, \qquad S_{xY(z)}(\alpha,v) = (W-1)^{-1}\sumw c_w U_w(z)
\enda
with $\Psi_{xx}(z,z; \alpha,v)  = o(1)$ and $\Psi_{xY}(z,z; \alpha,v) = o(1)$ by applying  Lemma \ref{lem:eff} to the augmented covariates.  Lemma \ref{bb_lin} ensures $\tg_{\ag ,z}(\alpha,v) = \{S_{xx}(\alpha,v)\}^{-1} S_{xY(z)}(\alpha,v) + \op$  and hence 
\begina
\gamma_{\ag ,z}(\alpha,v) = \{S_{xx}(\alpha,v)\}^{-1} S_{xY(z)}(\alpha,v) + o(1).
\enda
This, together with $\gamma_{c,z} =(\sumw c_w c_w^\T)^{-1}\{\sumw c_w U_w(z) \} = \{S_{xx}(\alpha,v)\}^{-1} S_{xY(z)}(\alpha,v) $ by standard \olss results, ensures $\gamma_{\ag ,z}(\alpha,v) - \gamma_{c,z}  = o(1)$ and hence $ 
S_\tl(\alpha,v) - S_6 = o(1)$. Condition \eqref{sd:1} then follows from \eqref{eq:sd_S}. 

\end{proof}

\section{Special case and extensions}\label{sec:ext_app}
\subsection{Uniform designs}\label{sec:balanced}
We outline in this subsection the unification of the three fitting schemes under Condition \ref{balanced}. 
The results clarify the theoretical guarantees by the sample-mean estimator and the corresponding ``ols'' fitting scheme under {\sym} designs.

Let $S$, $\Sigma$, and $\hat V$ be the common values of   
$S_*$,  $\Sigma_*$, and $\hat V_*$  for $* = \HT,\haj$ under Condition \ref{balanced}, respectively. Let $\hY$ be the common value of $\hys \ (* = \sm, \HT, \haj)$ under Condition \ref{balanced} by Proposition \ref{prop:hY_bal}. 
Corollary \ref{corr:clt_bal} below justifies the Wald-type inference of $\tau$ based on $(\hY, \hat V)$. 

\begin{corollary}\label{corr:clt_bal}
Assume the $2^2$ split-plot randomization and Conditions \ref{balanced}--\ref{asym}. Then 
$
\sqrt{W}  (\hY  -\bbY  ) \rightsquigarrow \mathcal{N}(0, \Sigma)$ and $W   \hbV -  \Sigma   = S  + \op$ 
with $S \geq 0$. 
\end{corollary}


Let $\tbV_\ols$, $\tbV_{\ols,\fisher}$, and $\tbV_{\ols,\lin}$ be the  cluster-robust covariances of $\tbeta_\ols$, $\tbeta_\of$, and $\tbeta_\ol$ from the \olss fits of the unit regressions \eqref{lm_t}, \eqref{dlm4_f}, and \eqref{dlm4}, respectively. 
Recall $\tg_\ols$ as the coefficient  vector of $x_\ws$ from the \olss fit of \eqref{dlm4_f}. 
Recall $\gamma_\ols$ as the probability limit of $\tg_\ols$  under Conditions \ref{balanced}--\ref{asym2} by Lemma \ref{bb_fisher}. 
Let $S_\of$  and $\Sigma_\of$ be the analogs of $S$ and $\Sigma$ defined on the adjusted potential outcomes $\Yws(z; \gamma_\ols) = \Yws(z) - x_{ws}^\T\gamma_\ols$, respectively.

\begin{proposition}\label{prop:ols} 
Under the $2^2$ split-plot experiment and Condition \ref{balanced}, we have
\begina
\quad \quad\quad  \tbeta_\dg = \hY, \quad \quad
\tbV_\dg 
 =\diag\left(  \frac{W_0-1}{W_0}I_2, \frac{W_1-1}{W_1}I_2\right)\hbV \qquad(  \dg = \ols, \wls, \ag),
 \enda 
and $\{\tbeta_\ol, (\tgoz)_{z\in\mt}\} = \{\tbeta_\wl, (\tgwz)_{z\in\mt}\}$. 
Further assume Conditions \ref{asym}--\ref{asym2}. Then 
\begina
\sqrt{W}  (\tbeta_\od  -\bbY ) \rightsquigarrow \mathcal{N}(0, \Sigma_\od ), \qquad  W   \tbV_\od-\Sigma_\od =S_\od+ \op
\enda
for $\dmd = \fisher, \lin$ 
with $(S_{\ols, \lin}, \Sigma_{\ols,\lin}) = (S_{\wls, \lin}, \Sigma_{\wls,\lin})$. 
 \end{proposition}

\begin{proof}[Proof of Proposition \ref{prop:ols}]
We verify below the numeric and asymptotic results, respectively.

\pa{Numeric results}
Condition \ref{balanced} ensures $\alpha_w = 1$ and  $q_{wb} = M_b/M =q_b$ for all $w$ such that 
\beginy\label{eq:bp}
\quad \quad N = WM, \qquad N_z = W_a M_b, \qquad  p_{ws}(z) = p_a q_b  = N_z/N, \qquad \pi_{ws} = N_{Z_{ws}}^{-1}
\endy
for all $\ws \in \ms$ and  $z=(ab)\in\mt$.
The numeric equivalence between $\tbeta_\dg$  follows from Propositions \ref{prop:hY_bal} and \ref{prop:ls}. 
That between $\{\tbeta_\ol, (\tgoz)_{z\in\mt}\} = \{\tbeta_\wl, (\tgwz)_{z\in\mt}\}$  follows from the equivalence between the ``ols'' and ``wls'' {\fts} under the treatment-specific regression  \eqref{lm:unit_z} with  $\pi_{ws} = N_{Z_{ws}}^{-1}$ being constant for all units under the same treatment. 
The results on $\tbV_* \ (* = \wls, \ag)$ follow from Theorem \ref{thm_vHats}. We verify below the result on $\tbV_\ols$. 

First, 
\beginy\label{eq:v_ols}
\tbV_\ols = ( \dt   D)^{-1} \left(\sumw  D_w^\T   \te_\olw  \te_\olw ^\T   D_w\right)( \dt    D)^{-1},
\endy
where $ \te_\olw  = (\te_{\ols,w1} , \dots, \te_{\ols,wM_w } )^\T$ with $\te_{\ols, \ws}  = \Yws - \tbeta_\ols(Z_{ws}) = \Yws - \hat Y(Z_{ws})$. 
Let  $\he_\olw(z)  = M_b  ^{-1}\sum_{s: Z_{ws}=z}\te_\olws$, with  $\he_\olw(z) =  \hY_w(z) - \hY(z)$ for $\zin$ and $\he_\olw(z) = 0$  for  $\znotin$. 
Set $Y_w = \te_\olw$ in \eqref{matrix0} 
to see $
 D_w^\T \te_\olw =   M Q \he_\olw$, where  $\he_\olw = (   \he_\olw(00),    \he_\olw(01),   \he_\olw(10),   \he_\olw(11)   )^\T$ and  $Q= I_2\otimes\diag(q_{0},q_{1})$ is the common value of $ Q_w$ over all $w$ under Condition \ref{balanced}. This ensures
\beginy
\sumw  D_w^\T   \te_\olw \te_\olw^\T   D_w =  M^2  Q \left(\sumw  \he_\olw \he_\olw^\T\right)Q  = \big(\Omega_\ols(z,z')\big)_{z,z'\in\mt}  \label{eq:deed}
\endy
with 
\begina
\Omega_\ols(z,z') &=& M_b   M_{b'} \sumw \he_\olw(z) \he_\olw(z')\\& =&\left\{
\begin{array}{cl} M_b   M_{b'} (W_a-1)\hat S(z,z') &\ \  \text{for $ z=(ab)$ and $z'=(a'b')$ with $a = a'$},\\
0  & \ \  \text{for $z=(ab)$ and $z' = (a'b')$ with $a\neq a'$.}
\end{array}\right.
\enda
The numeric result on $\tilde V_\ols$ then follows from \eqref{eq:v_ols}, \eqref{eq:deed}, and 
$ \dt   D  = \diag(W_0M_0, W_0M_1, W_1M_0, W_1M_1)$ by \eqref{matrix0} and \eqref{eq:bp}. 

\pa{Asymptotic results}
Let $\hx$ and $\hy_\gamma$ be the common values of $\hx_* \ (* = \sm, \HT, \haj)$ and $\hy_{*,\gamma} \ (* = \HT, \haj)$ under Condition \ref{balanced}, respectively. 
Proposition \ref{prop:ols_x_fisher} ensures $\tbeta_\of = \hy - \hx \tilde\gamma_\ols = \hy_\gamma$ for $\gamma = (\gamma_z)_{z\in\mt}$ where $\gamma_z = \tilde\gamma_\ols$ for all $z \in \mt$. 
The asymptotic Normality of $\tbeta_\of$ then follows from Lemma \ref{Cov_x} with $\tilde\gamma_\ols = \gamma_\ols+\op$ by Lemma \ref{bb_fisher}. 
The asymptotic Normality of $\tbeta_\ol$ follows from $\tbeta_\ol = \tbeta_\wl$ as we just showed. 
The asymptotic conservativeness of $\tilde V_{\ols, \dmd} \ (\dmd = \fisher, \lin)$ follows from the same reasoning as that for $\tilde V_\wl$ in the proof of Theorem \ref{thm_lin}.
\end{proof}

%
\subsection{HC2 correction for the cluster-robust covariance estimators}
%
The classic cluster-robust covariances recover $\hbV_\sht$ and $\hbV_\shaj$ only asymptotically by Theorem \ref{thm_vHats}. 
We verify in this subsection the exact recovery of $\hbV_\sht$ by the HC2 correction in finite samples \citep{CR}. 
Let 
 \begina
&&\tilde V_{\ols2}=( \dt   D)^{-1} \left\{\sumw  D_w^\T ( I -  P_\olw )^{-1/2}  \te_\olw  \te_\olw ^\T ( I -  P_\olw )^{-1/2} D_w\right\}( \dt    D)^{-1},\\
&&\tilde V_{\ag2} = ( \dt _\ag  D_\ag )^{-1} \left\{\sumw \Dw^\T( I -  P_\tlw )^{-1/2} \te_\tlw  \te_\tlw ^\T  ( I-  P_\tlw )^{-1/2}  \Dw\right\}( \dt _\ag  D_\ag )^{-1}
\enda
be the HC2 variants of $\tbV_\ols$ and $\vlsu$, respectively, with $ P_\olw  =  D_w( \dt  D)^{-1} D_w^\T$ and $ P_\tlw  = \Dw( \dt _\ag  D_\ag )^{-1}\Dw^\T$ for $w = \ot{W}$.

\begin{theorem}\label{vHats_prime}
$\tilde V_{\ag2}   = \hbV_\sht$. Under Condition \ref{balanced}, 
$\tilde V_{\ols2}
 = \tilde V_{\ag2}   =  \hbV
$. 
\end{theorem}

\begin{proof}[Proof of Theorem \ref{vHats_prime}]
We verify below the results on $\tilde V_{\ag2} $ and $\tilde V_{\ols2}$, respectively. 

\pa{Proof of $\tilde V_{\ag2}   = \hbV_\sht$} It follows from $
 D_\ag ^\T  D_\ag  = \diag(W_0, W_1)\otimes  I_2$
by \eqref{matrix0} and $D_\tlw  = 1(A_w=0) \cdot ( I_2,  0_{2\times 2}) + 1(A_w=1) \cdot ( 0_{2\times 2},  I_2)$ by definition 
that $ P_\tlw  = D_\tlw ( D_\ag ^\T D_\ag )^{-1} D_\tlw ^\T= W_{A_w}^{-1}  I_2$.
This ensures
\begina
&&  \sumw \Dw^\T( I -  P_\tlw )^{-1/2} \te_\tlw  \te_\tlw ^\T  ( I-  P_\tlw )^{-1/2}  \Dw\\
&=&  \sumw (1-W_{A_w}^{-1})^{-1}  \Dw^\T\te_\tlw  \te_\tlw ^\T   \Dw 
=\sum_{a=0,1} \frac{W_a}{W_a-1} \left( \sum_{w:A_w=a} D_\tlw ^\T \te_\tlw \te_\tlw ^\T D_\tlw  \right)
 \enda
such that 
\begina\tilde V_{\ag2}  = \left\{\diag\left( \frac{W_0}{W_0-1}, \frac{W_1}{W_1-1} \right)\otimes  I_2\right\} \tbV_\ag 
\enda
by the form of $\sum_{w:A_w=a} D_\tlw ^\T \te_\tlw \te_\tlw ^\T D_\tlw $ from \eqref{eq:deed.ag}.
The result then follows from 
Theorem \ref{thm_vHats}.

\pa{Proof of $\tilde V_{\ols2}
 = \hbV
$ under Condition \ref{balanced}} 
Recall $(M, M_0, M_1)$ as the common value of $(M_w, M_{w0}, M_{w1})$ for all $w$ under Condition \ref{balanced}. 
Assume without loss of generality that the first $M_0$ sub-plots in whole-plot $w$ receive level $0$ of factor B. 
Then $
 D_w =  ( D_{w,00},  D_{w,01},  D_{w,10},  D_{w,11})$, 
where $ D_{w,z} = ( 1(Z_{w1}=z), \dots, 1(Z_{wM_w}=z) )^\T$ with
\begina
\renewcommand{\arraystretch}{1.4}
\begin{array}{llll}
 D_{w,00} = ( 1_{M_0}^\T,   0^\T_{M_1})^\T,  &\ \ D_{w,01} = ( 0_{M_0}^\T,  1^\T_{M_1})^\T, &\ \ D_{w,10}= D_{w,11} =  0_M & \ \ \text{if $ A_w=0 $;}\\
D_{w,00}= D_{w,01} =  0_M,   &\ \  D_{w,10} = ( 1_{M_0}^\T,  0^\T_{M_1})^\T, &\ \   D_{w,11} = ( 0_{M_0}^\T,  1^\T_{M_1})^\T &\ \ \text{if $A_w=1$.}
\end{array}
\enda
This, together with $ \dt  D = \diag(N_z)_{z\in\mt}$ where  $N_z = W_aM_b$, ensures
\begina
 P_\olw  =  D_w( \dt   D)^{-1} D_w^\T =  \sumz N_z^{-1}  D_{w,z} D_{w,z}^\T = W_{A_w}^{-1}\left(
\begin{array}{cc} M_0^{-1} \bbo{M_0}&  0 \\ 0  & M_1^{-1} \bbo{M_1}\end{array}\right).
\enda

Note that the two non-zero columns of $ D_w$, namely 
$( 1^\T_{M_0},  0^\T_{M_1})^\T$ and $( 0^\T_{M_0},  1^\T_{M_1})^\T$, are both eigen-vectors of  $ P_\olw $, corresponding to the same eigen-value $W_{A_w}^{-1}$. 
They are thus also the eigen-vectors of  $( I- P_\olw )^{-1/2}$, corresponding to the same eigen-value $(1-W_{A_w}^{-1})^{-1/2}$, such that
$
(  I -  P_\olw )^{-1/2} D_w =(1-W_{A_w}^{-1})^{-1/2} D_w
$ \citep[][Section C.3]{imai2020causal}. This ensures   
\begina
\sumw  D_w^\T ( I -  P_\olw )^{-1/2}  \te_\olw  \te_\olw ^\T ( I -  P_\olw )^{-1/2} D_w=\sum_{a=0,1} \frac{W_a}{W_a-1} \left(\sum_{w:A_w=a} D_w^\T \te_\olw \te_\olw ^\T D_w\right) 
\enda
and hence
$$\tilde V_{\ols2} = \left\{\diag\left( \frac{W_0}{W_0-1}, \frac{W_1}{W_1-1} \right)\otimes  I_2\right\}   \tbV_\ols$$
by the form of $\sum_{w:A_w=a} D_w^\T \te_\olw \te_\olw ^\T D_w$ from \eqref{eq:deed}.
The result for $\tilde V_{\ols2}$ then follows from Proposition \ref{prop:ols}.

\end{proof}

\subsection{Covariate adjustment via factor-based regressions}\label{sec:ca_factor}
We give in this subsection the details on covariate adjustment under factor-based regressions. 

Let
$f_{ws} = ( A_{w}-1/2 ,   B_{ws} -1/2 ,  (A_{w}-1/2)(B_{ws}-1/2) )^\T$ 
and 
$ f_{w}(A_wb) = ( A_{w}-1/2,  b -1/2,  (A_{w}-1/2)(b-1/2) )^\T$ 
be the vectors of the non-intercept regressors of \eqref{dlm1_a} and \eqref{dlm3_a}, respectively. 
The additive and fully-interacted extensions of \eqref{dlm1_a} and \eqref{dlm3_a} equal
\beginy\label{lm_f_unit}
Y_{ws}   &\sim&
 1+ f_{ws} + x_{ws},\\
\label{lm_f_ag} \hU_w(A_wb)  &\sim& 1+ f_w(A_wb)  + \hvwab, \\
\label{lm_l_unit}
Y_{ws}  &\sim & 1 + f_{ws} + x_{ws} +  f_{ws}  \otimes x_{ws},\\
\label{lm_l_ag}
\hU_w(A_wb) 
&\sim& 1+ f_w(A_wb) + \hvwab+ f_w(A_wb) \otimes  \hvwab,
\endy
respectively. 
Let $\tbt'_{\dg,\fisher}$ and $\tbt'_{\dg,\lin}$ be the coefficient vectors of $f_{ws}$ and $f_w(A_wb)$ from \eqref{lm_f_unit}--\eqref{lm_l_ag} under {\fts} $ \dg = \ols,\wls, \ag$, respectively, with cluster-robust covariances $\tilde\Omega'_{\dg, \dmd }$ for $(\dg,\dmd ) \in \{\ols, \wls, \ag\}\times \{\fisher,\lin\}$.  
Let $\tbt'_{\ag,\dmd }(\alpha, v)$ and $\tilde\Omega'_{\ag,\dmd }(\alpha, v)$ be the variants of $\tbt'_{\ag,\dmd }$ and $\tilde\Omega'_{\ag,\dmd }$ after including the centered whole-plot size factor, $(a_w-1)$, as an additional whole-plot level covariate in addition to $\hvwab$. 
Proposition \ref{correspondence} below follows from the invariance of least squares to non-degenerate transformation of the regressors and,  together with Theorems \ref{thm_fisher}--\ref{thm_lin} and Proposition \ref{prop:eff},  ensures the optimality of $\tbt'_{\ag,\lin}(\alpha, v)$ for estimating the standard factorial effects $(\tau_\A, \tau_\B, \tau_\AB)^\T = G_0\by$ among $
 \{\tbt'_\wls, \tbt'_{\wls,\dmd }; \, \tbt'_\ag ,  \tbt'_{\ag ,\dmd }, \tbt'_{\ag ,\dmd }(\alpha,v):  \dmd  = \fisher, \lin \}$.  

\begin{proposition}\label{correspondence}
For $\dg = \ols,\wls,\ag $  and $\dmd  = \fisher, \lin$, we have 
\begina
\renewcommand{\arraystretch}{1.1}
\begin{array}{ll}
 \tbt'_{\dg,\dmd} = G_0\tbeta_{\dg,\dmd}, &\qquad \tbO'_{\dg,\dmd} =
G_0  \tbV_{\dg,\dmd}  G_0^\T;\\
\tbt'_{\ag,\dmd}(\alpha, v) = G_0\tbeta_{\ag,\dmd}(\alpha, v), &\qquad \tbO'_{\ag,\dmd}(\alpha, v) =
G_0  \tbV_{\ag,\dmd}(\alpha, v)G_0^\T  .
\end{array}
\enda
\end{proposition}
 

Recall $\tbt_\dagger'$ as the coefficient vectors of the non-intercept terms from the unadjusted factor-based regressions \eqref{dlm1_a} and \eqref{dlm3_a}, respectively, for estimating $(\tauA, \tauB, \tauAB)$.  
Let $(\tbt'_{\dagger, \textsc{b}}, \tbt'_{\dagger, \textsc{ab}})$ and $(\tbt'_{\dagger, \fisher, \B}, \tbt'_{\dagger, \fisher, \AB})$ be the elements of $\tbt'_{\dagger}$  and $\tbt_{\dagger, \fisher}'$ corresponding to $(\tauB, \tauAB)$, respectively. 
Proposition \ref{prop:no change} below states the invariance of $(\tbt_{\dagger, \textsc{b}}', \tbt_{\dagger, \textsc{ab}}')$ to additive covariate adjustment when the $x_{ws}$'s are identical within each whole-plot for $\dg = \wls, \ag$.
The result extends to the sub-plot effects and interactions under the general $T_\A\times T_\B$ design with minimal modification of the notation. 
This illustrates a key difference between the additive and fully-interacted adjustments under factor-based specifications; see Section \ref{sec:simu} for examples from simulation studies. 

\begin{proposition}\label{prop:no change}
For $\dg = \wls, \ag$, $(\tbt'_{\dagger, \textsc{b}}, \tbt'_{\dagger, \textsc{ab}}) = (\tbt'_{\dagger, \fisher, \B}, \tbt'_{\dagger, \fisher, \AB})$ if $x_{ws} = x_{w}$ for all $ws \in \ms$. 
\end{proposition}

\begin{proof}[Proof of Proposition \ref{prop:no change}]

Let $A_\cc$ and $B_\cc$ be the $N\times 1$ vectors of the centered factor indicators $(A_{ws}-1/2)_{\ws\in\ms}$ and $(B_{ws}-1/2)_{\ws\in\ms}$, respectively, where $A_{ws}=A_w$. 
Let $X = (x_{11}, \dots, x_{W,M_W})^\T$ with $x_{ws} =  x_{w}$. 
The design matrices of the unadjusted and additive factor-based unit regressions \eqref{dlm1_a} and \eqref{lm_f_unit} equal 
\begina
C = (1_N, \ac, \bc, A_\cc\circ B_\cc), \qquad 
C_\fisher = (1_N, \ac, X, \bc, A_\cc\circ B_\cc),
\enda
respectively, after a reshuffle of the column order in $C_\fisher$. 
Recall $\Pi = \diag(\pi_{ws})_{ws\in\ms}$ as the weighting matrix under the ``wls'' fitting scheme. 
Direct algebra shows that $C_1 = (1_N, \ac)$, $C_{\fisher, 1} = (1_N, \ac, X)$, and $C_{2} = (\bc, \ac\circ\bc)$ satisfy
\begina
C_1^\T \Pi C_2 = 0, \qquad C_{\fisher,1}^\T \Pi C_2 = 0
\enda
due to $\sums \pi_{ws}(B_{ws}-1/2) = 0$ for all $w$ following from \eqref{eq:piws}.
This ensures
\begina
C^\T \Pi C= 
\beginp 
C_{1}^\T \Pi C_{1} &0\\
0& C_{2}^\T \Pi C_{2}
\endp, 
\qquad 
C_\fisher^\T \Pi C_\fisher = 
\beginp 
C_{\fisher, 1}^\T \Pi C_{\fisher, 1} &0\\
0& C_{2}^\T \Pi C_{2}
\endp,
\enda
implying $(\tbt'_{\wls, \textsc{b}}, \tbt'_{\wls, \textsc{ab}})^\T =( C_{2}^\T \Pi C_{2})^{-1}C_{2}^\T \Pi Y =(\tbt'_{\wls,\fisher, \textsc{b}}, \tbt'_{\wls,\fisher, \textsc{ab}})^\T$. 

The proof for $\dagger = \ag$ follows from identical reasoning and is  thus omitted. 
\end{proof}
 
\subsection{General $T_\text{A}\times T_{\text{B}}$ split-plot design}\label{sec:ext_tatb_app}
Renew $S_\sht$, $S_\shaj$, and $ S_w$ as the   scaled between and within whole-plot covariance matrices for $Y_{ws}(z)$ and $Y'_{ws}(z)$ under the general $\ta\times\tb$ split-plot design. 
Renew
$H  = \diag(p_a^{-1})_{a\in \mt_\A}\otimes \bbo{T_\B} - \bbo{|\mt|}$ and 
$ H_w = \diag(p_a^{-1})_{a\in \mt_\A} \otimes \{  \diag(q_{wb}^{-1})_{b\in\mt_\B} - \bbo{T_\B}\}$ with 
 $p_a = W_a/W$  and $q_{wb} = M_{wb}/M_w$. 
Corollary \ref{clt_g} below extends Lemma \ref{Vmat} and Theorem \ref{clt} to the general $\ta\times \tb$ split-plot design.  The proof is identical to that of the $2^2$ case and thus omitted.  

\begin{corollary}\label{clt_g}
Lemma \ref{Vmat} and Theorem \ref{clt} hold under the $\ta\times \tb$ split-plot randomization and a generalized version of Condition  \ref{asym} for $a\in \mt_\textsc{a}$, $b \in \mt_\textsc{b}$, and $z = (ab)\in\mt$. 
\end{corollary}

The unit regression \eqref{lm_t} extends to the 
$T_\text{A}\times T_{\text{B}}$ split-plot design as 
\beginy\label{eq:lm_n_tatb}
Y_{ws} \sim \sumz 1(Z_{ws} = z). 
\endy
Lemma \ref{lem:wls} below states a numeric result on the invariance of least-squares fits of \eqref{eq:lm_n_tatb} to treatment-specific scaling of the fitting weights. 
The result ensures that the theory we derived under the ``\wls'' fitting scheme with inverse probability weighting extends to a much larger class of fitting weights with no need of modification. See Remark \ref{rmk:alternative_weights} in the main text for an example. 
We state the lemma in terms of general multi-armed experiment to highlight its generality.  
The proof follows from direct algebra and is thus omitted. 

\begin{lemma}\label{lem:wls}
For a general experiment with treatment levels $\mt = \{1,\ldots, Q\}$ and units $i=1,\ldots, N$, 
let $Y_i$ denote the observed outcome, $Z_i \in \mt$ denote the treatment assignment, $w_i > 0$ denote an arbitrary weight, and $\rho_i > 0$ denote a scaling factor that is a function of $Z_i$ only. 
The least-squares fits of $Y_i \sim \sumz 1(Z_i = z)$ with weights $(w_i)_{i=1}^N$ and $(w_i\rho_i)_{i=1}^N$ yield identical coefficients, heteroskedasticity-robust covariance, and cluster-robust covariance by arbitrary clustering rule. 
\end{lemma}

\subsection{Fisher randomization test}\label{sec:frt}
 We give in this subsection the details on the Fisher randomization test under split-plot randomization. 
Assume a general $\ta\times \tb$ split-plot experiment with 
potential outcomes  $\{Y_{ws}(z): z \in \mt\}_{ \ws \in \ms}$. 
The  weak null hypothesis concerns $$\HN: G\bbY = 0$$ for some full-row-rank contrast matrix $G$ with rows orthogonal to $1_{|\mt|}$.
Given observed data $Z = (Z_\ws)_{\ws\in\ms}$, $Y = (Y_{ws})_{\ws\in\ms}$, $X = (x_\ws)_{\ws\in\ms}$, and $\mZ$ as the set of all possible values that $Z$ can take under the split-plot randomization restriction,
we can pretend to be testing
\begina
\HF: Y_{ws}(z) = Y_{ws} \qquad  \text{for all} \ \ \text{$\ws \in \ms$ and  $z\in\mt$}
\enda
as a strong null hypothesis that is compatible with $\HN$,  
and compute the $p$-value as 
$$p_\frt = |\mZ|^{-1}\sum_{Z' \in \mZ}1\big\{ t(Z',   Y, X ) \geq t(Z,Y,X)\big\}$$
for some arbitrary test statistic $t(Z, Y, X)$. 
Of interest is the operating characteristics of $p_\frt$ when only $\HN$ holds. 

Renew $\mathcal{B} =  \{\tbeta_\wls,  \tbeta_{\wls,\dmd};\, \tbeta_\ag ,  \tbeta_{\ag ,\dmd}, \tbeta_{\ag ,\dmd}(\alpha,v):  \dmd  = \fisher, \lin  \}$ as the collection of {\reg} estimators of $\bY$ that are consistent  under the general $\ta\times\tb$ split-plot design.
Let 
$
t^2(\tbeta) = 
(G \tbeta)^\T   (G\tbV  G^\T)^{-1}  G\tbeta
$ 
be the robustly studentized test statistic based on $\tbeta\in \mathcal{B}$, 
with $\tbV$ as the corresponding cluster-robust covariance.
The robust studentization ensures that the resulting $p_\frt$ controls the type one error rates asymptotically in the sense of 
\begina
\lim_{W\to\infty} \pr(  p_\frt \leq \alpha   ) \leq \alpha \qquad \text{for all} \ \ \alpha \in (0,1)
\enda
 for all $\tbeta \in \mathcal{B}$ under Conditions \ref{asym}--\ref{asym2}. 
 The Fisher randomization test with $t^2(\tbeta)$ is therefore finite-sample exact for testing the strong null hypothesis and asymptotically valid for testing the weak null hypothesis under split-plot randomization  for all $\tbeta \in\mathcal{B}$. 
The duality between confidence interval and hypothesis testing further ensures that the test based on $t^2(\tbeta_{\ag,\lin} (\alpha,v))$ has the highest power asymptotically when the generalized version of \eqref{cond:eff} holds. 
The same guarantees also extend to $\tbeta \in  \{\tbeta_\ols,\ \tbeta_\of,\ \tbeta_\ol  \}$ under Condition \ref{balanced}.

\section{Simulation}\label{sec:simu}
Define a {\it regression scheme}  as the combination of model specification  and fitting scheme.
We illustrate in this section the validity and efficiency of thirteen regression schemes for estimating the standard factorial effects $\tauA$, $\tauB$, and $\tauAB$ under the $2^2$ split-plot design, summarized in Table \ref{tb:simu}.

\begin{table}[t]\caption{\label{tb:simu} Thirteen regression schemes based on the  unadjusted, additive (``\fisher''), and fully-interacted (``\lin'') factor-based specifications. 
Recall 
$f_{ws} = ( A_{w}-1/2 ,  B_{ws} -1/2 ,  (A_{w}-1/2)(B_{ws}-1/2) )^\T$
and 
$ f_{w}(A_wb) = (  A_{w}-1/2 ,  b -1/2 ,  (A_{w}-1/2)(b-1/2) )^\T$
as the vectors of the non-intercept regressors from \eqref{dlm1_a} and \eqref{dlm3_a}, respectively. 
Prefixes ``ols'', ``wls'', and ``ag'' indicate the {\fts}.  
Suffixes ``x.\fisher'' and ``x.\lin'' denote the additive and fully-interacted specifications for covariate adjustment, respectively. 
For the aggregate regressions, we use ``x'', ``m'', and ``xm'' to indicate three choices of covariate combinations: (i) uses the scaled whole-plot total covariates $ \hv_w(A_wb) = \alpha_w \hat x_w(A_wb)$ (``x''), (ii) uses the whole-plot size factor $\alpha_w$ (``m''), and (iii) uses both (``xm''). }
\begin{center}
\renewcommand{\arraystretch}{1.1}
\begin{tabular}{|l|c|l|}\hline
 \multicolumn{1}{|c|}{regression}& fitting &   \\
 \multicolumn{1}{|c|}{scheme} & scheme & \multicolumn{1}{c|}{model specification}    \\\hline
ols & \multirow{3}{*}{ols}& $1+f_\ws$ \\
ols.x.\fisher && $1+f_\ws + x_\ws$\\ 
ols.x.\lin && $1+f_\ws + x_\ws + f_\ws\otimes  x_\ws$\\\hline
wls & \multirow{3}{*}{wls}& $1+f_\ws$ \\
wls.x.\fisher && $1+f_\ws + x_\ws $\\ 
wls.x.\lin && $1+f_\ws + x_\ws + f_\ws \otimes  x_\ws$\\\hline
ag  & \multirow{7}{*}{ag}& $1+f_w(A_wb)$\\
ag.x.\fisher && $1+f_w(A_wb) + \hat v_w(A_wb) $\\
ag.x.\lin&& $1+f_w(A_wb) + \hat v_w(A_wb) + f_w(A_wb) \otimes \hat v_w(A_wb)$\\
ag.m.\fisher && $1+f_w(A_wb) + (\alpha_w-1)   $\\
ag.m.\lin && $1+f_w(A_wb) + (\alpha_w-1) + f_w(A_wb)  (\alpha_w-1) $\\
ag.xm.\fisher && $1+f_w(A_wb) + (a_w-1) + \hat v_w(A_wb)  $\\
ag.xm.\lin && $1+f_w(A_wb) + (a_w-1) + \hat v_w(A_wb) + f_w(A_wb) (\alpha_w-1) + f_w(A_wb) \otimes  \hat v_w(A_wb)$\\\hline
\end{tabular}
\end{center}
\end{table}
Consider a $2^2$ split-plot experiment with a study population nested in $W=300$ whole-plots. 
We set $(W_0, W_1) =(0.7W,0.3W)$ and generate $(M_{w0}, M_{w1}, M_w)_{w=1}^W$ as $
M_{w0} = \max(2, \zeta_{w0})$, $ M_{w1} = \max(2, \zeta_{w1})$, and $M_w  = M_{w0} + M_{w1}$, respectively, 
with the $\zeta_{w0}$'s being i.i.d.\ $\textrm{Poisson}(5)$ and the $\zeta_{w1}$'s being i.i.d.\  $\textrm{Poisson}(3)$. 
For each $w = \ot{W}$, 
we draw a scalar group-level covariate $x_{w}$ from $\mn(0.2,0.5)$, and set $x_{ws} =   x_w$ for $s = \ot{M_w}$. 
The potential outcomes are then generated as  
\begina
\begin{array}{ll}
Y_{ws}(00) = \theta_w +0.5 + 2 x_{ws}^2 + \ep_{ws},&\qquad
Y_{ws}(01) = -0.5 \theta_w +1 +  x_{ws}^2 + \ep_{ws},\\
Y_{ws}(10) = 0.5 \theta_w +1 -  x_{ws}^2 + \ep_{ws}, &\qquad
Y_{ws}(11) = \theta_w + 2 + 2 x_{ws}^2 + \ep_{ws} 
\end{array}
\enda 
for $\ws \in \ms$, 
where the $\theta_w$'s are independent $\mn(2M_w/M_{\max}, 0.2)$ with $M_{\max} =\max_{w=\ot{W}} M_w$  and the $\ep_{ws}$'s are i.i.d.\ $\textrm{Uniform}(-1,1)$. 
Fix $\{Y_\ws(z), x_\ws: z\in\mt  \}_{\ws\in\ms}$ in simulation. 
We draw a random permutation of $W_1$ $1$'s and $W_0$ $0$'s to assign factor A at the whole-plot level, and then, for each $w = \ot{W}$, draw a random permutation of $M_{w1}$ $1$'s and $M_{w0}$ $0$'s to assign factor B in whole plot $w$.  

The procedure is repeated $2,000$ times, with the biases (``bias''), true standard deviations (``sd''), average cluster-robust standard errors (``ese''), and coverage rates of the $95\%$ confidence intervals based on the cluster-robust standard errors (``coverage'') for all three standard effects summarized in Figure \ref{fig:simu_xw}. 
We separate the results into ``unadjusted vs.\ additive regressions'' and ``unadjusted vs.\ fully-interacted regressions'' for ease of display.

Figure \ref{fig:simu_xw}(a) shows the comparison between the unadjusted and additive regressions.
The first row illustrates the biases in the \olss estimators under {\aasym} split-plot designs. 
The second row illustrates the efficiency gain by covariate adjustment for estimating the whole-plot factor effect. 
The covariate-adjusted regressions ``ols.x'', ``wls.x'', ``ag.m'', and ``ag.mx''  yield less variable estimators than their unadjusted counterparts under all three {\fts}.
The comparison between ``ag'', ``ag.m'', ``ag.x'', and ``ag.mx'' under the ``ag'' fitting scheme further highlights the importance of whole-plot size adjustment for improving efficiency. 
The results for the sub-plot factor effect and interaction, on the other hand, remain unchanged under the ``wls'' and ``ag'' fitting schemes as Proposition \ref{prop:no change} suggests.

The third row shows the average cluster-robust standard errors for estimating the true standard deviations. Compare it with the second row to see the conservativeness that is coherent with Theorems \ref{covEst}--\ref{thm_fisher}. 
The last row illustrates the validity of the regression-based Wald-type inference.
The overall conservativeness is, again, coherent with Theorems \ref{covEst}--\ref{thm_fisher}.  

Figure \ref{fig:simu_xw}(b) shows the comparison between the unadjusted and fully-interacted regressions.
In addition to the same observations as in Figure \ref{fig:simu_xw}(a), 
it also illustrates the efficiency gain by covariate adjustment for the sub-plot factor effect and interaction as well. 
The ``a.mx.\lin'' regression scheme, as the theory suggests, secures the highest efficiency overall.

Lastly, Figure \ref{fig:simu_xws} gives the results when the covariates are not constant within each whole-plot. 
Inherit all settings from above except that we now generate $x_{ws}$ as
$
x_{ws} = x_w + \ep_{ws}
$, where the $\ep_{ws}$'s are i.i.d.\ $\mathcal{N}(0, 0.5)$.
The gain in efficiency by covariate adjustment for estimating the sub-plot factor effect and interaction now surfaces under the additive regressions as well. 

\begin{figure}[ht]\caption{\label{fig:simu_xw}Comparison of the least-squares estimators with $x_{ws} = x_w$. }
\centering
(a) Comparison between the unadjusted and additive regressions. We suppress the suffix ``.\fisher'' in the names of the additive regressions to save some space. 
\includegraphics[width = \textwidth]{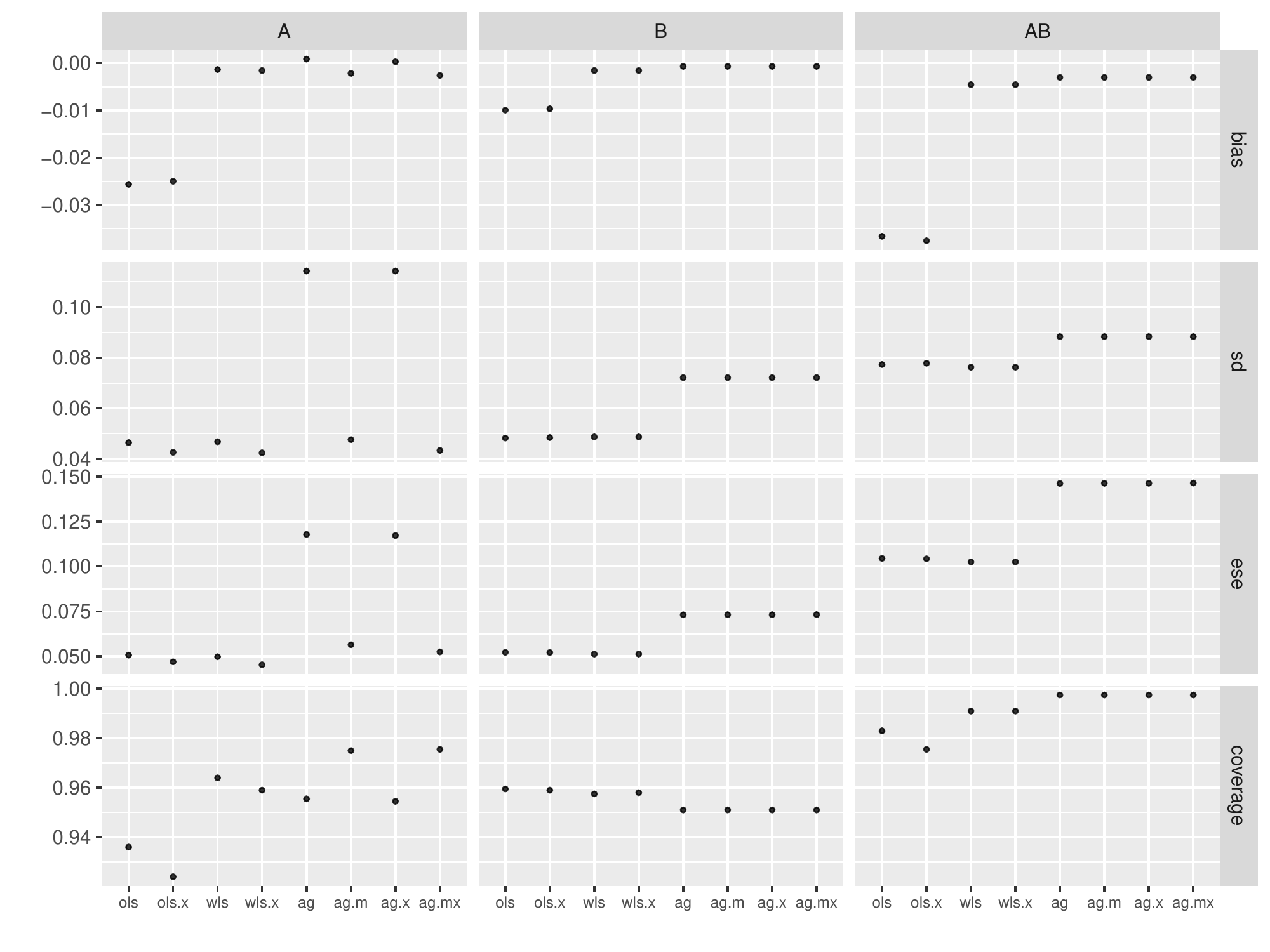}

(b) Comparison between the unadjusted and fully-interacted regressions. We suppress the suffix ``.\lin'' in the names of the fully-interacted regressions to save some space. 

\includegraphics[width =\textwidth]{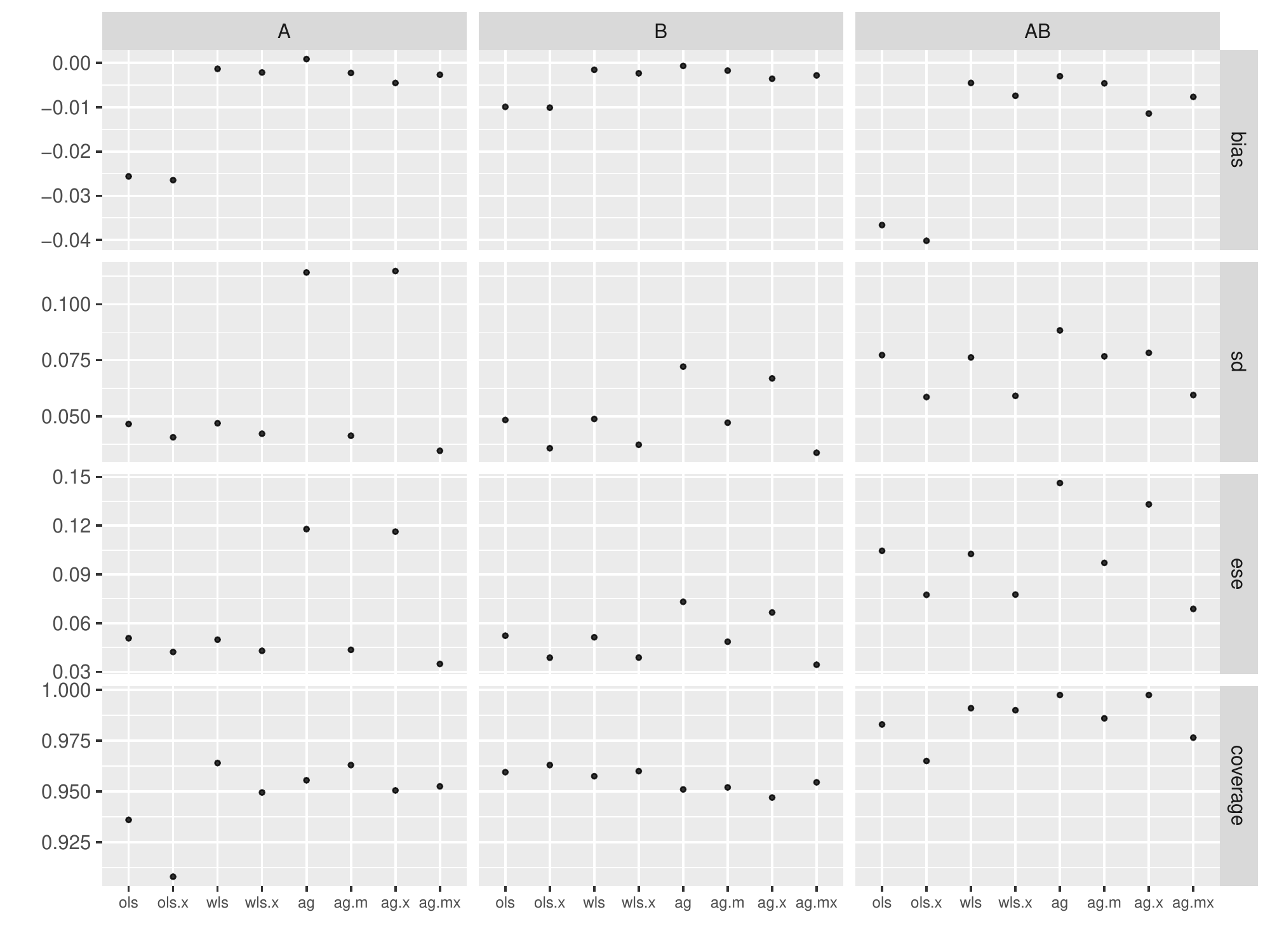}
\end{figure}

\begin{figure}[ht]\caption{\label{fig:simu_xws}Comparison of the least-squares estimators with varying $x_{ws}$ within each whole-plot. }
\centering
(a) Comparison between the unadjusted and additive regressions. We suppress the suffix ``.\fisher'' in the names of the additive regressions to save some space. 
\includegraphics[width = \textwidth]{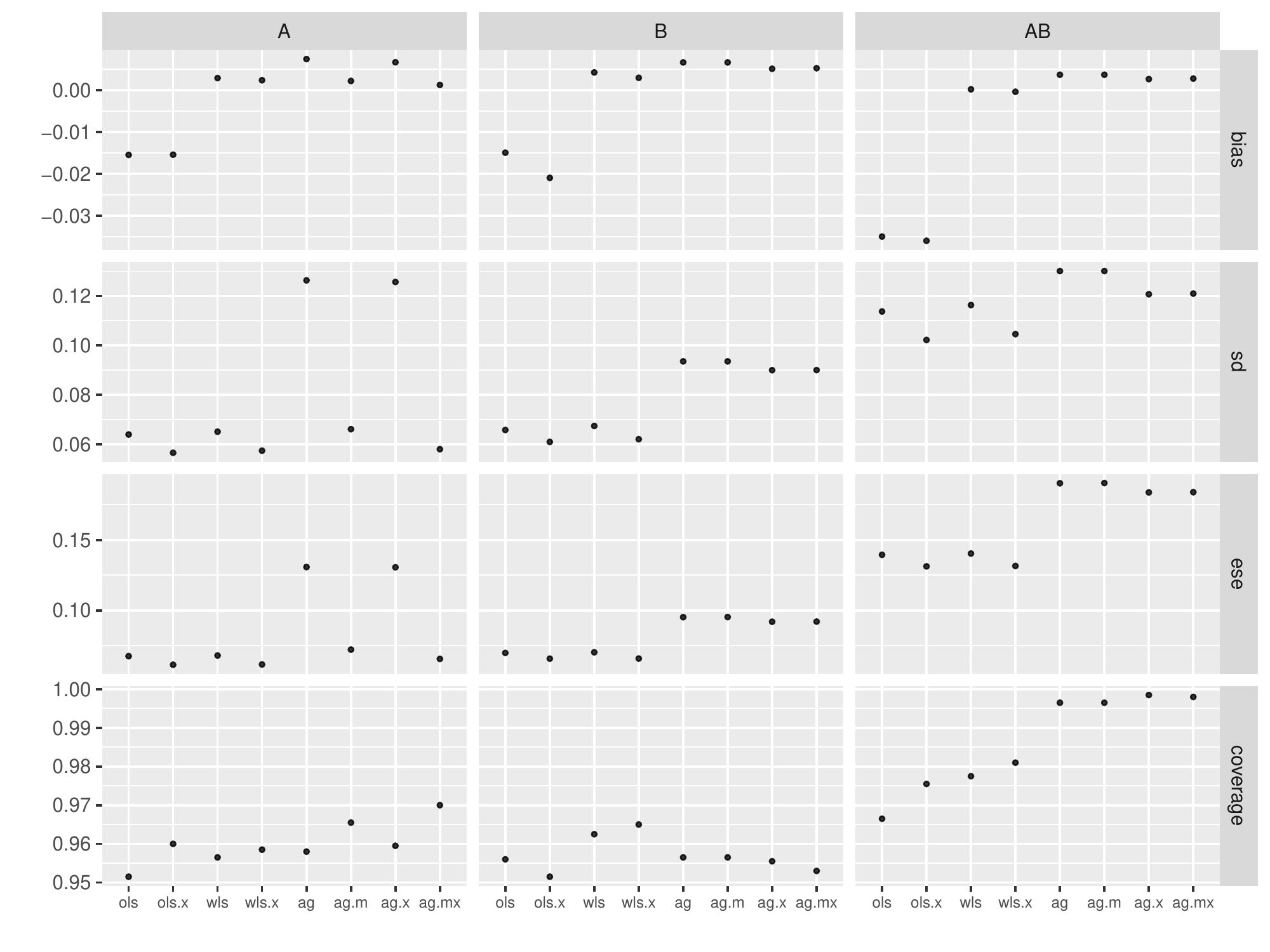}

(b) Comparison between the unadjusted and fully-interacted regressions. We suppress the suffix ``.\lin'' in the names of the fully-interacted regressions to save some space. 

\includegraphics[width =\textwidth]{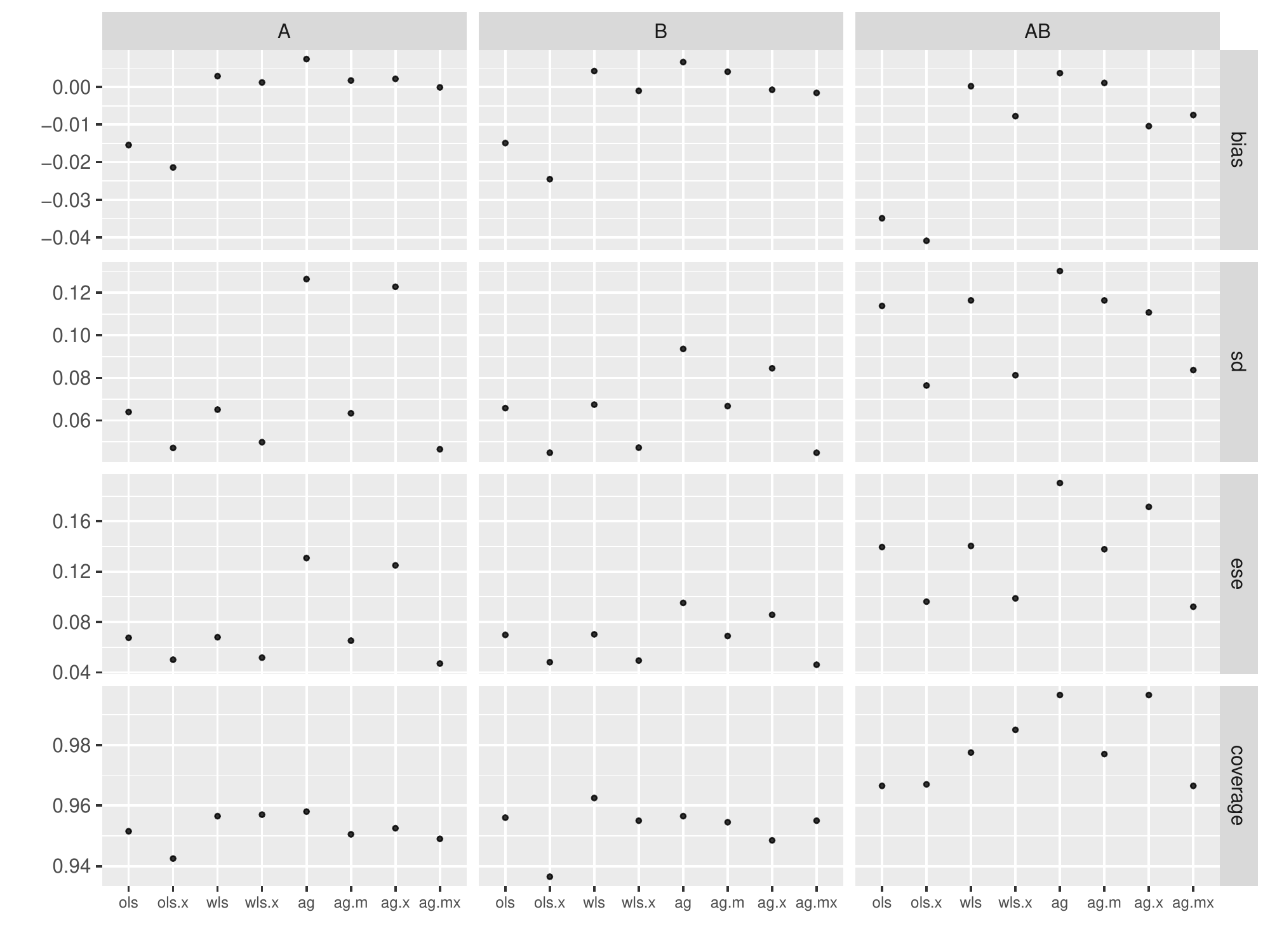}
\end{figure}

\end{document}